\documentclass[10pt,twocolumn,twoside]{IEEEtran}

\usepackage[english]{babel}
\usepackage[utf8]{inputenc}
\usepackage[T1]{fontenc}  
\usepackage{enumerate}
\usepackage{color}
\usepackage[T1]{fontenc}
\usepackage{subfigure}
\usepackage{dsfont}
\usepackage[final]{graphicx}
\usepackage[T1]{fontenc}
\usepackage{amsmath}
\usepackage{mathtools}
\usepackage{amsthm}
\usepackage{amstext}
\usepackage{amssymb}
\usepackage{mathrsfs}
\usepackage{cite}
\usepackage{mathtools}
\usepackage{tikz}
\usetikzlibrary{arrows,positioning}
\usepackage{xfrac}
\usepackage{cancel}

\usepackage{soul,xcolor}

\usepackage[shortlabels]{enumitem}

\newcommand{\me}{\mathrm{e}}
\newcommand{\nn}{\nonumber}
\newcommand{\mm}{m}

\newtheorem{theorem}{Theorem}
\newtheorem{corollary}{Corollary}
\newtheorem{fact}{Fact}
\newtheorem{proposition}{Proposition}
\newtheorem{lemma}{Lemma}
\theoremstyle{remark}
\newtheorem{example}{Example}
\newtheorem{remark}{Remark}
\newtheorem{assumption}{Assumption}
\newtheorem{definition}{Definition}

\newcommand{\ekRev}[1]{{#1}}

\allowdisplaybreaks

\begin{document}

\title{Signaling Games for Log-Concave Distributions: Number of Bins and Properties of Equilibria}

\author{
Ertan~Kaz{\i}kl{\i}, Serkan~Sar{\i}ta\c{s}, Sinan~Gezici, Tam\'{a}s~Linder, and~Serdar~Y{\"u}ksel \thanks{
Part of this work was presented at the 2019 IEEE International Symposium on Information Theory (ISIT), July 7-12, 2019, Paris, France \cite{NumberOfBinsISIT2019}. E. Kaz{\i}kl{\i}, T. Linder and S. Y\"uksel are with the Department of Mathematics and Statistics, Queen's University, K7L 3N6, Kingston, Ontario, Canada. Emails: ertan.kazikli@queensu.ca and \{linder,yuksel\}@mast.queensu.ca. 
S. Sar{\i}ta\c{s} is with the Department of Electrical and Electronics Engineering, Middle East Technical University (METU), 06800, Ankara, Turkey. Email: ssaritas@metu.edu.tr. S. Gezici is with the Department of Electrical and Electronics Engineering, Bilkent University, 06800, Ankara, Turkey, Email: gezici@ee.bilkent.edu.tr.}}

\maketitle

\begin{abstract}
We investigate the equilibrium behavior for the decentralized cheap talk problem for real random variables \ekRev{and quadratic cost criteria} in which an encoder and a decoder have misaligned objective functions. In prior work, it has been shown that the number of bins in any equilibrium has to be countable, generalizing a classical result due to Crawford and Sobel who considered sources with density supported on $[0,1]$. In this paper, we first refine this result in the context of log-concave sources. For sources with two-sided unbounded support, we prove that, for any finite number of bins, there exists a unique equilibrium. In contrast, for sources with semi-unbounded support, there may be a finite upper bound on the number of bins in equilibrium depending on certain conditions stated explicitly. Moreover, we prove that for log-concave sources, the expected costs of the encoder and the decoder in equilibrium decrease as the number of bins increases. Furthermore, for strictly log-concave sources with two-sided unbounded support, we prove convergence to the unique equilibrium under best response dynamics which starts with a given number of bins, making a connection with the classical theory of optimal quantization and convergence results of Lloyd’s method. In addition, we consider more general sources which satisfy certain assumptions on the tail(s) of the distribution and we show that there exist equilibria with infinitely many bins for sources with two-sided unbounded support. Further explicit characterizations are provided for sources with exponential, Gaussian, and compactly-supported probability distributions.
\end{abstract}

\begin{IEEEkeywords}
Cheap talk, signaling games, Nash equilibrium, optimal quantization, Lloyd–Max algorithm, payoff dominant equilibria.
\end{IEEEkeywords}

\section{Introduction}

Signaling games and cheap talk are concerned with a class of Bayesian games where a privately informed player (encoder or sender) transmits information (signal) to another player (decoder or receiver), who knows the probability distribution of the private information observed at the encoder. In these games/problems, the objective functions of the players are not aligned, unlike in classical communication problems. The cheap talk problem was introduced in the economics literature by Crawford and Sobel \cite{SignalingGames}, who obtained the striking result that under some technical conditions on the cost functions, the cheap talk problem only admits equilibria that involve quantized encoding policies, i.e., the observation space is partitioned into intervals and the encoder reveals the interval its observation lies in, rather than revealing the observation completely. This is in significant contrast to the usual communication/information theoretic case where the objective functions are aligned. Therefore, as indicated in \cite{SignalingGames}, the amount of information that can be revealed by the encoder depends on the similarity of the players' interests (objective functions). In other words, the message about the private information should be strategically designed and transmitted by the encoder. In this paper, we discuss extensions and generalizations of some results concerning strategic information transmission and cheap talk \ekRev{under quadratic cost criteria} by focusing on log-concave sources as well as more general sources which satisfy certain assumptions (rather than sources with a density supported on $[0,1]$ as studied in \cite{SignalingGames}).

\subsection{Problem Definition}

The focus of this paper is to address the following problems:

\subsubsection{Number of Bins}
In a previous work \cite{tacWorkArxiv}, it is shown that, since the distances between the optimal decoder actions are lower bounded due to \cite[Theorem 3.2]{tacWorkArxiv}, the quantized nature of an equilibrium holds for arbitrary scalar sources, rather than only for sources with a density supported on $[0,1]$ as studied in the seminal paper by Crawford and Sobel \cite{SignalingGames}. Hence, for bounded sources, it can easily be deduced that the number of bins at the equilibrium must be bounded. For example, for a uniform source on $[0,1]$ and quadratic objective functions, \cite{SignalingGames} provides an upper bound on the number of quantization bins as a function of the bias $b$\ekRev{, which appears in the objective function of the encoder and quantifies the degree of misalignment between the objectives of the encoder and decoder.} Accordingly, for unbounded sources, the following problems are of interest: 
\begin{itemize}
\item For sources with either semi-unbounded or two-sided unbounded support on the real line, is there an upper bound on the number of bins at an equilibrium as a function of the bias $b$? As a special case, is it possible to have only a non-informative equilibrium; i.e., the upper bound on the number of bins is one?
\item Is it possible to have an equilibrium with infinitely many bins?
\item Is the equilibrium unique for a given number of quantization bins?
\end{itemize}
At this point, one can ask why bounding the number of bins is important. Finding such a bound is useful since if one can show that there only exists a finite number of bins, and if the equilibrium is unique for a given number of bins, then the total number of equilibria is finite; this will allow for a feasible setting where the decision makers can coordinate their policies. \ekRev{It is also interesting to investigate upper bounds on the number of bins in equilibrium from a communication theoretic perspective. In particular, even though the talk is cheap, the communication between the encoder and decoder becomes limited by the upper bound on the number of bins in equilibrium. This can be useful in network design where communication resources may not be wasted when they are not needed.}

Furthermore, in a recent work, where signaling games and cheap talk problems are generalized to dynamic (multi-stage) setups, a crucial property that allowed the generalization was the assumption that the number of bins for each stage equilibrium, conditioned on the past actions, is uniformly bounded \cite[Theorem~5]{dynamicGameArxiv}. In view of this, showing that the number of bins is finite is a useful technical result. 

\subsubsection{Equilibrium Costs}
Attaining the upper bound $N$ on the number of bins at an equilibrium implies that there exists at least one equilibrium with each of $1,2,\ldots,N$ bins for the bounded support scenario due to \cite[Theorem~1]{SignalingGames}, and thus, a new question arises: among these multiple equilibria, which ones induce smaller expected cost simultaneously for both of the players? Results in \cite{SignalingGames} show that under certain assumptions an equilibrium with more bins induces smaller expected costs for both the encoder and the decoder for a source with a density supported on $[0,1]$. Accordingly, for more general sources, the problem is to see how the expected costs of the players in equilibrium behave with respect to number of bins. This is known as the payoff dominant equilibrium selection property \cite{eqSelectionBook}. This also allows for a well-posed coding problem as one would like to design a coding scheme that is payoff dominant for both the encoder and the decoder.

\subsubsection{Convergence to Equilibria}
The interaction between the encoder and decoder can be viewed as a quantization game where the encoder decides on quantization bins and the decoder decides on reconstruction values. One might wonder how an equilibrium point is reached in this quantization game. For instance, the players may act sequentially by computing their best responses, which corresponds to a repeated play of the considered quantization game. Then, a problem of interest is to see whether these best response iterations converge to an equilibrium. Another related issue is that if they start from an arbitrary quantization policy with $N$ bins, is it the case that these best response iterations converge to an equilibrium with $N$ bins?

\subsection{Preliminaries}

\ekRev{In this paper, we consider the communication model introduced by Crawford and Sobel \cite{SignalingGames} and specialize to the quadratic cost setup as explained below. In this model,} there are two players with misaligned objective functions. An informed player (encoder) knows the value of an $\mathbb{M}$-valued random variable $M$ and transmits an $\mathbb{X}$-valued random variable $X$ to another player (decoder), who generates an $\mathbb{M}$-valued decision $U$ upon receiving $X$. The policies of the encoder and the decoder are assumed to be deterministic; i.e., $x=\gamma^e(m)$ and $u=\gamma^d(x)=\gamma^d(\gamma^e(m))$. Let $c^e(m,u)$ and $c^d(m,u)$ denote the cost functions of the encoder and the decoder, respectively, when the action $u$ is taken for the corresponding message $m$. Then, given the encoding and decoding policies, the encoder's induced expected cost is $J^e\left(\gamma^e,\gamma^d\right) = \mathbb{E}\left[c^e(M, U)\right]$, whereas the decoder's induced expected cost is $J^d\left(\gamma^e,\gamma^d\right) = \mathbb{E}\left[c^d(M, U)\right]$. Here, we assume real valued random variables and quadratic cost functions; i.e., $\mathbb{M}=\mathbb{X}=\mathbb{R}$, $c^e\left(m,u\right) = \left(m-u-b\right)^2$ and $c^d\left(m,u\right) = \left(m-u\right)^2$, where $b\in\mathbb{R}$ denotes a bias term which is common knowledge between the players. \ekRev{The encoder essentially wishes to introduce a certain amount of bias for the action taken by the decoder.} We assume that the encoder and the decoder announce their policies at the same time. Then, a pair of policies $(\gamma^{*,e}, \gamma^{*,d})$ is said to be a Nash equilibrium (e.g., \cite{basols99}) if
\begin{align}
\begin{split}
J^e(\gamma^{*,e}, \gamma^{*,d}) &\leq J^e(\gamma^{e}, \gamma^{*,d})  \text{ for all } \gamma^e \in \Gamma^e ,\\
J^d(\gamma^{*,e}, \gamma^{*,d}) &\leq J^d(\gamma^{*,e}, \gamma^{d})  \text{ for all }\gamma^d \in \Gamma^d ,
\label{eq:nashEquilibrium}
\end{split}
\end{align}
where $\Gamma^e$ and $\Gamma^d$ are the sets of all deterministic (and Borel measurable) functions from $\mathbb{M}$ to $\mathbb{X}$ and from $\mathbb{X}$ to $\mathbb{M}$, respectively. As observed from the definition in \eqref{eq:nashEquilibrium}, under a Nash equilibrium, each individual player chooses an optimal strategy given the strategy chosen by the other player. 

The quantized nature of Nash equilibria for cheap talk problem \cite{SignalingGames} motivates the following definition of a scalar quantizer.
\begin{definition}\label{def:quantizer}
An $N$-cell scalar quantizer, $q$, is a (Borel) measurable mapping from $\mathbb{M}=\mathbb{R}$ to the set $\{1,2,\dots,N\}$ characterized by a measurable partition $\{B_1,B_2,\dots,B_N\}$ such that $B_i=\{x\,|\,q(x)=i\}$ for $i=1,2,\dots,M$ and that bin probabilities are strictly positive. The $B_i$ are called the bins (or cells) of $q$. 
\end{definition}

Due to results obtained in \cite{SignalingGames} and \cite{tacWorkArxiv}, we know that the encoder policy consists of convex bins at a Nash equilibrium\footnote{We note that, unlike Crawford and Sobel's simultaneous Nash equilibrium formulation, if one considers a Stackelberg formulation (see \cite[p.~133]{basols99} for a definition), then the problem reduces to a classical communication problem since the encoder is then committed a priori and the equilibrium is not quantized; i.e., there exist affine equilibria \cite{tacWorkArxiv,dynamicGameArxiv,CedricWork,akyolITapproachGame,omerHierarchial}.}. Namely, at a Nash equilibrium, the encoder must employ quantization policies with the cells $B_i$ in Definition~\ref{def:quantizer} being intervals.

Now consider an equilibrium with $N$ bins. At this equilibrium, let $k$th bin be the interval $[\mm_{k-1},\mm_k)$ for $k=1,\dots,N$ where $\mm_0<\mm_1<\ldots<\mm_N$ denote bin edges. Let $l_k$ denote the length of the $k$th bin; i.e., $l_k=\mm_k-\mm_{k-1}$ for $k=1,2,\ldots,N$. Note that for an equilibrium with $N$ bins, we have $\mm_0=m_L$ and $\mm_N=m_U$ for sources with bounded support with $m_L$ and $m_U$ denoting lower or upper boundaries, respectively. In the case of sources with semi-unbounded support, we have $\mm_0=m_L$ and $\mm_N=+\infty$, or $\mm_0=-\infty$ and $\mm_N=m_U$ whereas in the case of sources with two-sided unbounded support, we have $\mm_0=-\infty$ and $\mm_N=+\infty$. Here, the encoder reports the interval its observation lies in. This can be represented by a quantization policy with $\gamma^e(m)=k$ when $[\mm_{k-1},\mm_k)$ for $k=1,\dots,N$. In this representation, the decoder knows that receiving $k$ means the encoder observed $m\in [\mm_{k-1},\mm_k)$. \ekRev{Due to \cite[Theorem~1]{SignalingGames} and \cite[Theorem~3.2]{tacWorkArxiv}, we have the following equilibrium conditions for our problem.} At an equilibrium, decoder's best response to encoder's action is characterized by
\begin{align} \label{centroid}
u_k = \mathbb{E}[M|\mm_{k-1}\leq M<\mm_k] 
\end{align}
for $k=1,\dots,N$; i.e., the optimal decoder action is the centroid for the corresponding bin. From the encoder's point of view, the best response of the encoder to decoder's action is determined by the nearest neighbor condition\footnote{Although this condition can be viewed as a nearest neighbor condition with a bias term, we simply use the term nearest neighbor condition as in the quantization theory literature.} as follows:
\begin{align} 
u_{k+1}-\mm_k = (\mm_k - u_k) - 2b \Leftrightarrow \mm_k = \frac{u_k + u_{k+1}}{2}+b
\label{eq:centroidBoundaryEq}
\end{align}
for $k=1,\dots,N-1$. Due to the definition of Nash equilibrium in \eqref{eq:nashEquilibrium}, these best responses in \eqref{centroid} and \eqref{eq:centroidBoundaryEq} must match each other, and only then can the equilibrium be characterized; i.e., for a given number of bins, the positions of the bin edges are chosen by the encoder, and the centroids are determined by the decoder. Alternatively, the problem can be considered as a quantization game in which the boundaries are determined by the encoder and the reconstruction values are determined by the decoder.

Based on the above, the problems we consider in this paper can be formulated more formally as follows:

\subsubsection{Number of Bins}
For a given finite (or infinite) $N$, does there exist an equilibrium with $N$ bins; i.e., is it possible to find optimal encoder actions (the boundaries of the bins) $\mm_0,\mm_1,\ldots,\mm_N$ and decoder actions (the centroids of the bins) $u_1,u_2,\ldots,u_N$ which satisfy \eqref{centroid} and \eqref{eq:centroidBoundaryEq} simultaneously? \ekRev{In relation to this problem, we introduce the following definition regarding the maximum possible number of bins in equilibrium for a given problem setup. 
\begin{definition}
For a given source density and a certain bias $b$, we define the maximum possible number of bins in equilibrium as follows:
\begin{align}
N_{\max}(b) \triangleq \sup&\{N\geq 1\,| \,\text{there exists an}\nonumber \\
&\hphantom{\{} \text{equilibrium with }N\text{ bins} \}.\label{def:Nmax}
\end{align}
\end{definition}
\begin{remark}
Note that $N_{\max}(b) < \infty$ implies that there is a finite upper bound on the number of bins in equilibrium, while if $N_{\max}(b) = \infty$, then there exist equilibria with an arbitrarily large number of bins.
\end{remark}
}


\subsubsection{Equilibrium Costs} 
An equilibrium is \textbf{more informative} (than another one) if it results in smaller expected costs for both the encoder and the decoder. Let $J^{e,N}$ and $J^{d,N}$ denote the encoder cost and the decoder cost, respectively, at an equilibrium with $N$ bins. Then, the aim is to see if $J^{d,N}>J^{d,N+1}$ holds or not for all finite $N$.\footnote{It suffices to consider the expected cost of one of the players since Remark~\ref{rem:equilibriumSelection} shows that for a given equilibrium with $N$ bins, we have that $J^{e,N}=J^{d,N}+b^2$.}

\subsubsection{Convergence to Equilibria}Consider an initial arbitrary set of ordered bin edges $m_0,m_1,\dots,m_N$ that does not necessarily lead to an equilibrium where $m_0$ and $m_N$ correspond to boundaries of the support for the source distribution. Suppose that the encoder and decoder iteratively compute their best responses by taking the centroid conditions and nearest neighbor conditions (i.e., \eqref{centroid} and \eqref{eq:centroidBoundaryEq}) into account by which bin edges $m_1,m_2,\dots,m_{N-1}$ and centroids $u_1,\dots,u_N$ are updated iteratively. Then, the aim is to see whether these bin edges and centroids converge to an equilibrium with $N$ bins as a result of best response iterations. This correspond to modified Lloyd's Method I \cite{LloydIT82} where the modification specific to our setting is due to the bias term in \eqref{eq:centroidBoundaryEq}. 

\subsection{Related Literature}

Cheap talk and signaling game problems find applications in networked control systems when a communication channel/network is present among competitive and non-cooperative decision makers \cite{basols99,misBehavingAgents}. Also, there have been a number of related results in the economics and control literature in addition to the seminal work by Crawford and Sobel, which are reviewed in \cite{tacWorkArxiv,dynamicGameArxiv} (see \cite{sobelSignal} for an extensive survey). We note that although Crawford and Sobel's simultaneous-move Nash formulation is considered in this paper, the signaling games literature also focuses on the sequential-move Stackelberg formulation, which is also referred to as the Bayesian persuasion problem in the economics literature \cite{bayesianPersuasion,CedricWork,akyolITapproachGame,omerHierarchial,LeTreustAllerton16,LeTreustJET19,LeTreustArxiv20,subjectiveBiasSIT}.

In addition to seminal work in \cite{SignalingGames}, there have been many contributions to \ekRev{signaling game} problems from a variety of aspects in the economics literature \cite{BattagliniEconometrica2002,ComparativeCheap07,MiuraGAB2014,ambrus2008multi,blume2007noisy,LongCheapTalk03,DynamicSIT14,KartikJET07,GescheGEB21,CheapTalkTwoSenderInformative}. For instance, a set of related works analyzes multidimensional cheap talk problems \cite{BattagliniEconometrica2002,MiuraGAB2014,ComparativeCheap07,ambrus2008multi}. Interestingly, in contrast to one-dimensional cheap talk model of Crawford and Sobel \cite{SignalingGames}, in certain scenarios, the sender may reveal certain information to the decoder fully for the cheap talk problems with multiple senders considered in \cite{MiuraGAB2014,BattagliniEconometrica2002,ambrus2008multi}. The work in \cite{blume2007noisy} considers a cheap talk setup where with some probability the decoder does not observe the message transmitted by the encoder and instead observes a message coming from a distribution which is statistically independent of the encoder's transmission and the source. This setup reduces to Crawford and Sobel's formulation \cite{SignalingGames} when the probability of error is zero. They show the quantized nature of the equilibria with convex bins (i.e., intervals) in such a noisy communication setting. The work in \cite{CheapTalkTwoSenderInformative} considers a cheap talk setup with two senders communicating with a receiver and an interesting observation is that more information revelation by one sender incentives more information revelation by the other sender. \ekRev{Furthermore, in \cite{KartikJET07}, the authors consider a strategic communication setting where the cost function of the sender depends on the conveyed message, which leads to a departure from the cheap talk setup of Crawford and Sobel \cite{SignalingGames}. In particular, it is costly for the sender to convey inaccurate information (i.e., there is a certain cost for lying). Unlike in \cite{SignalingGames}, there exists fully revealing equilibria under certain conditions in \cite{KartikJET07}.}

The quantized nature of the equilibrium connects game theory with the quantization theory. For a comprehensive survey regarding the history of  quantization and results on the optimality and convergence properties of different quantization techniques (including Lloyd's methods), we refer to \cite{quantizationSurvey}. In particular, \cite{Fleischer64} shows that, for sources with a log-concave density, Lloyd’s Method I converges to the unique optimal quantizer. It is shown in \cite{trushkin82} and \cite{kieffer83} that Lloyd’s Method I converges to the globally optimal quantizer if the source density is continuous and log-concave, and if the error weighting function is convex and symmetric. For sources with bounded support, the condition on the source is relaxed to include all continuous and positive densities in \cite{wu92}, and convergence of Lloyd’s Method I to a (possibly) locally optimal quantizer is proven. The number of bins of an optimal entropy-constrained quantizer is investigated in \cite{gyorgyLinder2003}, and conditions under which the number of bins is finite or infinite are presented. As an application to smart grids, \cite{csLloydMax} considers the design of signaling schemes between a consumer and an electricity aggregator with finitely many messages (signals); the best responses are characterized and the maximum number of messages (i.e., quantization bins) are found using Lloyd's Method II via simulation. The work in \cite{QuantizationGameTSP21} studies evolution of language in social networks by modeling the problem as a quantization game where information spreads over a network of strategic decision makers. The authors show convergence to equilibria via a distributed version of the Lloyd-Max algorithm for the considered setup.

The existence of multiple quantized equilibria necessitates a theory to specify which equilibrium point is the solution of a given game. Two different approaches can be taken to achieve a unique equilibrium. One of them reduces the multiplicity of equilibria by requiring that off-the-equilibrium-path beliefs satisfy an additional restriction (e.g., by shrinking the set of players' rational choices) \cite{eqSelectionRefinement}, \cite{sobelSignal}. As introduced in \cite{eqSelectionBook}, the other approach presents a theory that selects a unique equilibrium point for each finite game as its solution; i.e., one and only one equilibrium points out of the set of all equilibrium points of this kind (e.g., see \cite{eqSelectionSignaling} for the application). Our results have implications on an equilibrium selection criterion known as payoff dominant equilibria \cite{eqSelectionBook}. In particular, we show that for the signaling game setup with log-concave sources and quadratic cost structure, an equilibrium with more bins is more informative since an equilibrium with more bins leads to reduced expected costs for both of the players. \ekRev{The work \cite{SobelSelectionEconometrica08} introduces the notion of no incentive to separate (NITS) condition for selection among multiple equilibria for general cheap talk problems where they associate a special meaning to the smallest source realization in certain applications in economics.}

As mentioned earlier, signaling game problems are also investigated under Stackelberg equilibria where one of the players is committed a priori \cite{bayesianPersuasion,CedricWork,akyolITapproachGame,omerHierarchial,LeTreustAllerton16,LeTreustJET19,LeTreustArxiv20,VoraCDC20,VoraISIT20,subjectiveBiasSIT,AkyolITW15,AkyolISIT16}. \ekRev{For instance, the works in \cite{AkyolITW15,AkyolISIT16,akyolITapproachGame} consider problems under the Stackelberg equilibrium concept with a biased sender where the cost functions are quadratic, and the source and bias are modeled as jointly Gaussian random variables.} The problem of strategical coordination is considered in \cite{LeTreustAllerton16,LeTreustJET19,LeTreustArxiv20}; specifically, the information design and a point-to-point strategic source-channel coding problems (originated from the Bayesian persuasion game \cite{bayesianPersuasion}) between an encoder and a decoder with non-aligned utility functions is investigated under the Stackelberg equilibrium. In addition, \cite{VoraCDC20,VoraISIT20} investigate communication scenarios between a sender and a receiver such that the receiver is the Stackelberg leader and aims to recover the source sequence as correct as possible. Due to the strategic nature of the sender, not all the transmitted information is truthful. Under this setting, the authors investigate how much true information can be recovered by the receiver from such a sender. They propose the notion of information extraction capacity from the strategic sender, which quantifies the growth rate of the number of recovered sequences with the block-length.

\subsection{Contributions}
\begin{enumerate}[(i)]
\item For log-concave sources with two-sided unbounded support, we prove that for any $N\in\mathbb{N}$, there exists a unique equilibrium with $N$ bins regardless of the value of $b$, i.e., there is no \ekRev{finite} upper bound on the number of bins in equilibrium (Theorem~\ref{thm:LogConcaveTwoSidedExistence}).
\item For log-concave sources with semi-unbounded support, we show that when the support extends to $+\infty$ and $b<0$, or the support extends to $-\infty$ and $b>0$, there exists a \ekRev{finite} upper bound on the number of quantization bins in equilibrium and that for any $N$ less than or equal to this upper bound there exists a unique equilibrium with $N$ bins (Theorem~\ref{thm:LogConcaveOneSidedExistence}). Otherwise, for any $N\in\mathbb{N}$, there exists a unique equilibrium with $N$ bins.
\item We prove that an equilibrium with more bins is more informative for log-concave sources, i.e., the equilibrium costs of the encoder and the decoder decrease as the number of bins increases (Theorem~\ref{thm:LogConcaveCost}).
\item We show that for a source with a strictly log-concave density and two-sided unbounded support, if the encoder and the decoder iteratively compute their best responses starting from an initial set of bin edges with $N$ bins, then it is guaranteed that these bin edges converge to the unique equilibrium with $N$ bins regardless of the value of $b$ (Theorem~\ref{thm:convergence}).
\item For sources with semi-unbounded support extending to $+\infty$ and under certain assumptions, we prove that there exists a finite upper bound on the number of bins in equilibrium when $b<0$ whereas for $b>0$ there exist equilibria with infinitely many bins (Theorem~\ref{FinThm}).
\item For sources with two-sided unbounded support and under certain assumptions, we show that there exist equilibria with infinitely many bins (Theorem~\ref{thm:twoSidedInfinite}).
\item For log-concave sources, we prove that there exist equilibria with infinitely many bins in the cases when there does not exist a finite upper bound on the number of bins in equilibrium (Corollary~\ref{cor:infManyLogConcaveOneSided} and Corollary~\ref{cor:infManyLogConcaveTwoSided}).
\item For exponential sources, we obtain an upper bound on the number of bins in equilibrium for $b<0$ (Proposition~\ref{prop:expProp1}). On the other hand, for $b>0$, we prove that there exists a unique equilibrium with infinitely many bins (Theorem~\ref{thm:expExistence}). Furthermore, the equilibrium costs are shown to achieve their minimum at the equilibrium with infinitely many bins (Theorem~\ref{thm:expMoreInformative}). 
\end{enumerate}

\section{Some Relevant Prior Results}\label{sec:PriorResults}

In this section, we review and discuss closely related prior results in the literature. In this paper, we touch upon all of these listed results and extend them. To begin with, in their seminal paper \cite{SignalingGames}, Crawford and Sobel prove the quantized nature of Nash equilibria under certain technical assumptions for sources with a density supported on $[0,1]$. It should be emphasized that the result of Crawford and Sobel holds for more general cost functions which satisfy certain conditions. These conditions hold for the quadratic criteria analyzed in this paper. In the next theorem, $u^e(m)$ and $u^d(m)$ are defined as $u^e(m) \triangleq \arg \min_u c^e(m,u)$ and $u^d(m) \triangleq \arg \min_u c^d(m,u)$ where $c^e(m,u)$ and $c^d(m,u)$ denote the cost function of the encoder and decoder, respectively, and the former term depends also on the bias term $b$. As noted earlier, here we only consider the Nash (simultaneous) setup, and not the Stackelberg (Bayesian persuasion) setup.

\begin{theorem}[Crawford and Sobel {\cite[Theorem~1]{SignalingGames}} ]\label{thm:CS}
Let $M$ be a real-valued random variable which admits a density supported on $[0,1]$. Suppose that $u^e(m)\neq u^d(m)$ for all $m$. Then, there exists at least one equilibrium with a quantization policy at the encoder where there are $N$ quantization bins with $1\leq N \leq N_{\max}(b)$ and $N_{\max}(b)\in \mathbb{N}$ denoting the upper bound on the number of bins.\ekRev{\footnote{\ekRev{In \cite{CSCorrectionEconometrica21}, the authors show that while this statement is correct, the proof in \cite{SignalingGames} relies on some incorrect statements and gives an example for which this proof fails. A correct version of the proof is given in \cite{CSCorrectionEconometrica21}.}}} Furthermore, an encoding policy at an equilibrium is equivalent to a quantized encoder policy in terms of the performance at the equilibrium. 
\end{theorem}

It is also possible to consider arbitrary scalar valued random variables and to show the quantized nature of Nash equilibria in this case, as well. Namely, the following result applies to arbitrary scalar valued random variables, not necessarily sources on $[0,1]$ that admit densities. Here, the costs functions are assumed to be quadratic, which is also the case in this paper.  

\begin{theorem}[Sar{\i}ta\c{s} et al. {\cite[Theorem~3.2 and Theorem~3.4]{tacWorkArxiv}}] \label{thm:SSTac}
Let $M$ be a real-valued random variable with an arbitrary probability measure. Let the strategy set of the encoder consists of the set of all measurable (deterministic) functions from $\mathbb{M}$ to $\mathbb{X}$. Then, an equilibrium encoder policy has to be quantized almost surely, that is, it is equivalent to a quantized policy for the encoder in the sense that the performance of any equilibrium encoder policy is equivalent to the performance of a quantized encoder policy. Furthermore, the quantization bins are convex. In addition, the quantized nature of equilibria does not necessarily hold when the source is vector-valued.
\end{theorem}

A set of related results involves multi-stage (dynamic) cheap talk game setup considered in \cite{dynamicGameArxiv}. In particular, \cite[Theorem~3]{dynamicGameArxiv} establishes the quantized nature of the last stage in a repeated cheap talk game and \cite[Theorem~5]{dynamicGameArxiv} makes further assumptions on the single-stage setup to prove that encoding policies at all stages must be quantization policies with a finite number of bins at each stage. In this latter result, it is assumed that there exists an upper bound on the number of bins at the equilibria in addition to an upper bound on the number of equilibria for a given number of bins. It is noted that this paper draws conclusions regarding these assumptions for log-concave sources. 

Moreover, our results in this paper are closely related to some foundational results by Kieffer \cite{kieffer83} where optimal quantization in the classical sense (i.e., in a team theoretic setup\footnote{A setup is referred to as team theoretic if all the decision makers in the system wish to minimize a common objective function, and thus share a common goal.}) is considered. This problem can be viewed as a communication problem where an encoder decides on quantization bins and a decoder determines reconstruction values where the encoder and the decoder wishes to minimize a common cost function. In this setting, \cite{kieffer83} shows that if the source is log-concave there exists a unique locally optimal quantization policy for a given number of bins under a set of assumptions on the common cost function. The uniqueness result makes the corresponding quantizer globally optimal for a given number of bins. 

\begin{theorem}[Kieffer \cite{kieffer83}]\label{thm:kieffer}
If the source is log-concave and the common cost function is convex, strictly increasing and continuously differentiable, then there exists a unique locally optimal quantizer and Lloyd-Max iterations (i.e., Lloyd's Method I) converge to this quantizer. 
\end{theorem}

Note that the special case of quadratic cost criteria with $b\neq 0$ is analyzed in this paper. For the case of $b=0$ leading to a team theoretic setup, Theorem~\ref{thm:kieffer} quoted above proves the uniqueness of the optimal quantizer for a given number of bins and in addition to convergence to this optimal quantizer via Lloyd-Max iterations. In that respect, this paper generalizes these uniqueness and convergence results to the game theoretic setup where $b\neq 0$. We will use Kieffer's analysis in some crucial steps of our paper.

\section{Nash Equilibria for Sources with Log-Concave Distributions}\label{sec:LogConcave}

Before presenting our results, we make the following definitions.

\begin{definition}\label{defn:logConcave}
A probability density function $f$ on $\mathbb{R}$ is said to be (strictly) log-concave if $\log f$ is a (strictly) concave function on the support of $f$ where the support is an interval. 
\end{definition}

We note that distributions such as Gaussian, exponential, Laplace, Rayleigh and uniform, which are commonly encountered in information and communication theoretic applications, are log-concave (see Section~\ref{sec:specialLogConcave} for more specific results for the cases where the source is exponential, Gaussian and \ekRev{uniform).} The reader is referred to \cite{bagnoli2005log} to see a list of common distributions which are log-concave or log-convex.

\begin{definition}\label{defn:support}
For a probability density function $f$, we say that:
\begin{enumerate}[(i)]
\item $f$ has a two-sided unbounded support if its support is $\mathbb{R}$.
\item $f$ has a semi-unbounded support if its support is either an interval in the form $(-\infty,m_U]$ or an interval in the form $[m_L,\infty)$ for some $m_U,m_L\in\mathbb{R}$.
\item $f$ has a bounded support if $f(x)=0$ for all $x\notin [m_L,m_U]$ for some $-\infty<m_L<m_U<\infty$. 
\end{enumerate}
\end{definition}

\subsection{Existence and Uniqueness of Equilibria}

In the following theorem, we present existence and uniqueness results for sources with a log-concave density and two-sided unbounded support. Note that an equilibrium $N=1$ bins is equivalent to the case where the encoder and decoder do not communicate, as remarked below.

\begin{remark}\label{rem:babbling}
If there is only one bin; i.e., $N=1$, then an equilibrium is called a non-informative (babbling) equilibrium \cite{SignalingGames}. Under such an equilibrium, the encoder transmits a message that is independent of the source and the decoder takes an action only by considering prior distribution of the source (since the received message is useless).
\end{remark}

\begin{theorem}\label{thm:LogConcaveTwoSidedExistence}
Consider a source that has a log-concave density with two-sided unbounded support. Then, \ekRev{$N_{\max}(b)=\infty$ and }for any $N\geq 1$, there exists a unique equilibrium with $N$ bins.
\end{theorem}
\begin{proof}
See Section~\ref{ProofLogConcaveTwoSidedExistence}.
\end{proof}

\begin{remark}
In the proof of Theorem~\ref{thm:LogConcaveTwoSidedExistence}, we employ the property that the mean of a source with a log-concave density is finite. This finite mean property follows from the fact that the tails of a log-concave density are at most exponential \cite[Lemma~1]{CuleSamworthEJS10}. In the following, we restate this result for the special case of densities defined on $\mathbb{R}$.
\end{remark}
\begin{lemma}[{\cite[Lemma~1]{CuleSamworthEJS10}} ] \label{lem:expTail}
For a log-concave density $f$ defined on $\mathbb{R}$, there exists $a>0$ and $b\in\mathbb{R}$ such that $f(x)\leq \exp(-a|x|+b)$ for all $x\in \mathbb{R}$.
\end{lemma}

\begin{remark}
It is noted that in the proof of Theorem~\ref{thm:LogConcaveTwoSidedExistence} the source distribution is assumed to be strictly log-concave. In fact, the result holds for distributions which are only log-concave. In that case, for the result to hold, it is required that the density is not log affine everywhere. This is indeed true for a log-concave distribution with two-sided unbounded support.
\end{remark}

\begin{remark}
For a source log-concave density with two-sided unbounded support, there also exist equilibria with infinitely many bins (see Corollary~\ref{cor:infManyLogConcaveTwoSided}). In order to prove this result, we show that a log-concave density satisfies the assumptions made in Section~\ref{sec:general} where the proof for existence of infinitely many bins essentially relies on Tychonoff's fixed-point theorem \cite{fixedPointBook}.
\end{remark}

Theorem~\ref{thm:LogConcaveTwoSidedExistence} reveals that there does not exist a \ekRev{finite} upper bound on the number of bins for an equilibrium considering log-concave sources with two-sided unbounded support. Namely, there exist countably infinite number of distinct equilibria for log-concave sources with two-sided unbounded support. 

\ekRev{In the proof of Theorem~\ref{thm:LogConcaveTwoSidedExistence}, we build on the idea that one can obtain equilibria with higher number of bins by varying the left-most or right-most bin edge depending on the sign of the bias $b$. A similar approach is taken in \cite[Theorem~1]{SignalingGames} for sources with a bounded support where the authors first prove the existence of an upper bound on the number of bin in equilibrium denoted by $N_{\max}(b)$ and then show the existence of an equilibrium with $N$ bins for each $N\in\{1,\dots,N_{\max}(b)\}$. In contrast, Theorem~\ref{thm:LogConcaveTwoSidedExistence} already reveals that there does not exist a finite upper bound on the number of bins in equilibrium in the case of log-concave sources with two-sided unbounded support. Note also that in the case of sources with a bounded support, the problem reduces to proving the existence of solutions to difference equations with given initial and terminal conditions, which do not exist in the case of sources with two-sided unbounded support. With these observations, we begin with proving the existence of an equilibrium with $N=2$ bins and then by varying the left-most or right-most bin edge we increment $N$ one by one. We also note that in the case of sources with a bounded support an additional monotonicity assumption (M) with regard to the behavior of bin edges guarantees the uniqueness of an equilibrium for a given number of bins as stated in \cite[Lemma~3]{SignalingGames}. In contrast, in Theorem~\ref{thm:LogConcaveTwoSidedExistence} of this paper, we prove the existence and uniqueness results together in the case of (log-concave) sources with two-sided unbounded support.}

The uniqueness result of Theorem~\ref{thm:LogConcaveTwoSidedExistence} is related to the classical results for optimal quantization in \cite{trushkin82} and \cite{kieffer83} where the aim is to find quantization bins and reconstruction values that minimize a common cost function (i.e., team theoretic setup). \ekRev{As stated in Theorem~\ref{thm:kieffer} of this paper, \cite{kieffer83} particularly focuses on optimal quantization of log-concave sources and proves the uniqueness in such a team theoretic setup.} 

The existence and uniqueness results of Theorem~\ref{thm:LogConcaveTwoSidedExistence} can be generalized to sources with semi-unbounded support. In contrast, for sources with semi-unbounded support, there may be a \ekRev{finite} upper bound on the number of bins depending on the sign of $b$.

\begin{theorem}\label{thm:LogConcaveOneSidedExistence}
\begin{enumerate}[(i)]
\item Consider a source that has a log-concave density with a support on $[m_L,\infty)$ for some $m_L\in\mathbb{R}$. \ekRev{If $b>0$, then $N_{\max}(b)=\infty$, and for any $N\geq 1$, there exists a unique equilibrium with $N$ bins.}
\item Consider a source that has a log-concave density with a support on $[m_L,\infty)$ for some $m_L\in\mathbb{R}$. \ekRev{If $b<0$, then $N_{\max}(b)<\infty$, and for any $N$ satisfying $1\leq N \leq N_{\max}(b)$, there exists a unique equilibrium with $N$ bins.}
\item Consider a source that has a log-concave density with a support on $(-\infty,m_U]$ for some $m_U\in\mathbb{R}$. \ekRev{If $b<0$, then $N_{\max}(b)=\infty$, and for any $N\geq 1$, there exists a unique equilibrium with $N$ bins.}
\item Consider a source that has a log-concave density with a support on $(-\infty,m_U]$ for some $m_U\in\mathbb{R}$. \ekRev{If $b>0$, then $N_{\max}(b)<\infty$, and for any $N$ satisfying $1\leq N \leq N_{\max}(b)$, there exists a unique equilibrium with $N$ bins.}
\end{enumerate}
\end{theorem}

\begin{proof}
See Section~\ref{ProofLogConcaveOneSidedExistence}.
\end{proof}

\begin{remark}
In the case of semi-unbounded support, the result of Theorem~\ref{thm:LogConcaveOneSidedExistence} holds even for densities which are log affine everywhere on its support. In fact, the exponential distribution, which is investigated in more detail in the paper, is an example of such a distribution.
\end{remark}

\begin{remark}
For a source with a log-concave density and semi-unbounded support, in the cases when there does not exist a finite upper bound on the number of bins in equilibrium, there also exist equilibria with infinitely many bins (see Corollary~\ref{cor:infManyLogConcaveOneSided}). This result is obtained by showing that a log-concave density satisfies the assumptions made in Section~\ref{sec:general}.
\end{remark}

The following theorem specifies a necessary and sufficient condition for the existence of an informative equilibrium, which is proven in a similar manner to Theorem~\ref{thm:LogConcaveOneSidedExistence}.

\begin{theorem}\label{thm:logConcaveNoninformative}
\begin{enumerate}
\item[(i)] For a source that has a log-concave density with a support on $[m_L,\infty)$ for some $m_L\in\mathbb{R}$ and has a mean $\mathbb{E}[M]=\mu$, the only equilibrium is non-informative if $2b\leq -(\mu-m_L)$. Otherwise, there exists an (unique) equilibrium with two bins. 
\item[(ii)] For a source that has a log-concave density with a support on $(-\infty,m_U]$ for some $m_U\in\mathbb{R}$ and has a mean $\mathbb{E}[M]=\mu$, the only equilibrium is non-informative if $2b\geq (m_U-\mu)$. Otherwise, there exists an (unique) equilibrium with two bins. 
\end{enumerate}
\end{theorem}

\begin{proof}
See Section~\ref{ProofLogConcaveNoninformative}.
\end{proof}

We note that for sources with a density supported on $[0,1]$ Crawford and Sobel prove that for a given $b$ there exists a \ekRev{finite} upper bound on the number of bins expressed as $N_{\max}(b)$ and that for any $1\leq N \leq N_{\max}(b)$ there exists an equilibrium with $N$ bins in \cite[Theorem~1]{SignalingGames} (see Theorem~\ref{thm:CS} of this paper). This existence result is valid even for a source without a log-concave density. As mentioned earlier, \cite[Lemma~3]{SignalingGames} presents a uniqueness result for sources with a bounded support under an additional monotonicity assumption (M). \ekRev{In fact, \cite[Theorem~1]{Szalay12} shows that this monotonicity condition (M) holds for compactly supported log-concave source distributions, and thus, establishes the uniqueness of equilibrium for a given number of bins by using \cite[Lemma~3]{SignalingGames}. We also note that our result in Lemma~\ref{lem:monotonicity} when applied to compactly supported log-concave source distributions also reveals that the monotonicity condition (M) holds.}


\subsection{Equilibrium Costs}\label{sec:equilibriumSelection}
So far, the number of bins at the equilibria has been investigated for sources with two-sided unbounded support and semi-unbounded support. At this point, it is interesting to examine the behavior of the encoder's and decoder's expected costs in equilibrium with respect to number of bins. In the following theorem, we show that an equilibrium with $N$ bins payoff dominates an equilibrium with $K$ bins if $K<N$ \cite{eqSelectionBook}. In other words, an equilibrium with $N$ bins leads to a smaller expected cost for both of the player than an equilibrium with $K$ bins if $K<N$. Before presenting this result, we first make the following remark regarding the relation between the expected costs of the encoder and decoder in equilibrium.

\begin{remark}\label{rem:equilibriumSelection}
Note that at an equilibrium of this quantization game, the relation between the encoder cost and the decoder cost can be expressed as follows:
\begin{align*}
&J^e(\gamma^{*,e}, \gamma^{*,d}) = \sum_{i=1}^{N}\mathrm{Pr}(\mm_{i-1}\leq M<\mm_i)\nn\\
&\times\mathbb{E}\left[(M-\mathbb{E}[M|\mm_{i-1}\leq M<\mm_i]-b)^2|\mm_{i-1}\leq M<\mm_i\right] \nn\\
&=\sum_{i=1}^{N}\mathrm{Pr}(\mm_{i-1}\leq M<\mm_i)\Big(b^2\nn\\
&+\mathbb{E}\left[(M-\mathbb{E}[M|\mm_{i-1}\leq M<\mm_i])^2|\mm_{i-1}\leq M<\mm_i\right]\Big) \nn\\
&= J^d(\gamma^{*,e}, \gamma^{*,d})+b^2 .
\end{align*}
Thus, the difference between the expected cost of the encoder and the decoder in equilibrium is always $b^2$ regardless of the number of bins under the quadratic cost assumption. This implies that if an equilibrium with more bins induces smaller expected cost for one of the players, then it also induces smaller expected cost for the other player.
\end{remark}

\begin{theorem}\label{thm:LogConcaveCost}
For sources with a log-concave density, an equilibrium with more bins induces smaller expected costs for both of the players. In other words, denoting expected costs of the encoder and the decoder by $J^{e,N}$ and $J^{d,N}$ at the unique equilibrium with $N$ bins, we have $J^{e,N}<J^{e,K}$ and $J^{d,N}<J^{d,K}$ if $K<N$.
\end{theorem}
\begin{proof}
See Section~\ref{ProofLogConcaveCost}.
\end{proof}

We note that for sources with a support on $[0,1]$ Crawford and Sobel establish that an equilibrium with more bins induces smaller expected costs in \cite[Theorem~3]{SignalingGames} and \cite[Theorem~5]{SignalingGames} from the perspective of encoder and decoder, respectively, under the monotonicity assumption (M), which involves sources with a bounded support. On the other hand, Theorem~\ref{thm:LogConcaveCost} proves such an informativeness result for log-concave sources with two-sided unbounded support, and this result also generalizes to sources with semi-unbounded support as well as sources with a bounded support.

\ekRev{In the proof of Theorem~\ref{thm:LogConcaveCost}, we again use the idea of varying the left-most or right-most bin edge building on the approach employed in the proof of \cite[Theorem~3]{SignalingGames}, which considers sources with a bounded support and makes the monotonicity assumption (M) mentioned earlier. In our proof, we first show that it is possible to obtain an equilibrium with $N$ bins starting from an equilibrium with $N+1$ by varying the value of left-most or right-most depending on the sign of $b$ and then prove that the expected cost is monotonic during this procedure. Here, the log-concave source assumption ensures that the corresponding left-most or right-most bin edge is monotonic with respect to the number of bins in equilibrium. Moreover, the approach taken in the proof of Theorem~\ref{thm:LogConcaveCost} builds also on the proof of Theorem~\ref{thm:LogConcaveTwoSidedExistence} with a slight modification which enables us to observe the monotonic behavior of the expected costs.}

Theorem~\ref{thm:LogConcaveCost} shows that an equilibrium with more bins leads to a better expected cost in equilibrium for both of the players. When $b=0$ (i.e., team theoretic setup), the problem reduces to the classical optimal quantization problem under quadratic cost structure and the result of Theorem~\ref{thm:LogConcaveCost} is in fact valid also when $b=0$. In other words, the classical optimal quantization setup with the quadratic cost criterion is a special case of Theorem~\ref{thm:LogConcaveCost}.

\begin{remark}\label{rem:exante}
In Theorem~\ref{thm:LogConcaveCost}, we compare expected costs of the players in equilibrium and prove that these costs are monotonically decreasing with respect to number of bins. This directly implies that the decoder always prefers an equilibrium with more bins. In that respect, our result is a generalization of \cite[Theorem~3]{SignalingGames} for log-concave sources. On the other hand, our result implies that the encoder ex ante (before observing a realization of the source) prefers an equilibrium with more bins. This is a generalization of \cite[Theorem~5]{SignalingGames} for log-concave sources. In that respect, our results reveal that if we compare (ex ante) expected costs of the encoder and decoder in equilibrium among all possible Nash equilibria, the one with the maximum possible number of bins gives the smallest expected cost, i.e. the most informative equilibrium.
\end{remark}

\begin{remark}
Without a log-concave assumption, such a monotonic behavior of the expected costs in equilibrium with respect to number of bins may not hold in general. For instance, we provide a simple example where there exists an equilibrium with $N=3$ bins for which the expected costs in equilibrium is worse than an equilibrium with $N=2$ bins. Although we provide an example with a discrete distribution for simplicity, one can also construct an example with a continuous density. 
\end{remark}

\begin{example}
Let the bias term be $b=0.9$. Suppose that the distribution of the source is given by the following discrete distribution: 
\begin{align}
p_M(m) = 
\begin{cases}
0.8& \text{if }m=0\\
0.05& \text{if }m\in\{2,4,6,8\}\\
0 & \text{otherwise}
\end{cases}.
\end{align}
Consider an encoding policy with $N=2$ bins where the quantization bins are given by $[0,m_1)$ and $[m_1,8]$ with $m_1=3.9588$. The corresponding centroids are given by $u_1=\mathbb{E}[M|M<m_1] =  0.1176$ and $u_2=\mathbb{E}[M|m_1\leq M] =  6$. Notice that the nearest neighbor condition $m_1=(u_1+u_2)/2+b$ is satisfied with these centroids. This leads to an expected cost value of $J^d=0.5882$ for the decoder at the characterized equilibrium. Now consider an encoding policy with $N=3$ bins where the quantization bins are given by $[0,m_1)$, $[m_1,m_2)$ and $[m_2,8]$ with $m_1=4.0667$ and $m_2=7.9$. The corresponding centroids are given by $u_1=\mathbb{E}[M|M<m_1] =  0.3333$, $u_2=\mathbb{E}[M|m_1\leq M<m_2] =  6$ and $u_3=\mathbb{E}[M|m_2\leq M] =  8$. With these centroids, the nearest neighbor conditions are satisfied, i.e., $m_1=(u_1+u_2)/2+b$ and  $m_2=(u_2+u_3)/2+b$. This leads to an expected cost value of $J^d=0.9$ for the decoder at the characterized equilibrium. Therefore, the characterized equilibrium with two bins payoff dominates the characterized equilibrium with three bins, which is in contrast to the result of Theorem~\ref{thm:LogConcaveCost} involving log-concave sources.
\end{example}


\subsection{Convergence to Equilibria}

It is interesting to investigate if the best response iterations converge to the unique equilibrium. Before presenting our result, we note that in the classical optimal quantization setup with aligned cost structure (i.e., team theoretic setup) and a log-concave density, the uniqueness of a fixed point and convergence to the unique fixed point via fixed point iterations are proven in \cite{kieffer83} (see Theorem~\ref{thm:kieffer} of this paper). In the following theorem, we present a signaling games counterpart of such a convergence result where the proof is based on the result of \cite{kieffer83}. In particular, when modified Lloyd-Max iterations (i.e., by using \eqref{centroid} and \eqref{eq:centroidBoundaryEq}) are performed starting from an initial set of bin edges, the unique equilibrium (for the given number of bins) is reached for strictly log-concave sources with two-sided unbounded support. These iterations in the classical team theoretic setting is referred to as Lloyd's Method~I \cite{LloydIT82} where bin edges and centroids are updated in a parallel fashion as described below.\footnote{Note that parallel updates in Lloyd's Method~I indeed models the scenario where the encoder and the decoder update their policies iteratively in our game theoretic setting. Lloyd's Method~II \cite{LloydIT82} is another heuristic technique to obtain a quantization policy in communication theoretic settings where bin edges and centroids are updated serially by taking the nearest neighbor and centroid conditions into account, i.e., given a value for the first centroid, compute the first bin edge and then compute the second centroid, and so on.}

A modified Lloyd-Max iteration can be defined as follows: Let $m_1<\dots<m_{N-1}$ be an initial set of bin edges and denote $\boldsymbol{m}\triangleq [m_1,\dots,m_{N-1}]^T$. Given these bin edges, the decoder first determines its best response by computing the centroids via
\begin{align}
u_k = \mathbb{E}[M|\mm_{k-1}\leq M<\mm_k]
\end{align}
for $k=1,\dots,N$, where $m_0=m_L$ and $m_N=m_U$ remain the same during these iterations. We denote this operation by $\mathrm{BR}_d(\boldsymbol{m})= \boldsymbol{u}$ where $\boldsymbol{u}\triangleq [u_1,\dots,u_{N}]^T$. Then, the encoder computes the nearest neighbors for the resulting centroids $\boldsymbol{u}$ via
\begin{align}
m_k=\frac{u_k+u_{k+1}}{2}+b
\end{align}
for $k=1,\dots,N-1$. We denote this operation by $\mathrm{BR}_e(\boldsymbol{u})= \boldsymbol{m}$. With these best response characterizations, a modified Lloyd-Max iteration is defined as
\begin{align}
T(\boldsymbol{m}) \triangleq \mathrm{BR}_e(\mathrm{BR}_d(\boldsymbol{m})).\label{eq:LloydMaxIter}
\end{align}
The operation in \eqref{eq:LloydMaxIter} is applied iteratively to update an initial set of bin edges. In the following theorem, we show that by repeating these iterations, it is guaranteed that a set of bin edges that leads to an equilibrium with $N$ bins is obtained under a strict log-concave source assumption.

\begin{theorem}\label{thm:convergence}
For a strictly log-concave source with two-sided unbounded support, best response iterations (i.e., modified Lloyd-Max iterations) through \eqref{eq:LloydMaxIter} starting with $N$ bins always converge to the unique equilibrium with $N$ bins.
\end{theorem}
\begin{proof}
See Section~\ref{sec:convergenceProof}.
\end{proof}

\begin{remark}
We note that \cite{kieffer83} proves convergence in a team theoretic setup even for log-concave sources which are not necessarily strictly log-concave. To prove this result, \cite{kieffer83} uses the fact that the common cost is strictly decreasing as long as the corresponding set of bin edges is not a fixed point. Since there is no common cost in our signaling games setup, such an approach is not feasible.
\end{remark}

In the case of sources with semi-unbounded support, convergence to the unique equilibrium is guaranteed depending on the sign of $b$.
\begin{theorem}\label{thm:convergenceSemiInf}
\begin{enumerate}[(i)]
\item Consider a source that has a strictly log-concave density with a support on $[m_L,\infty)$ for some $m_L\in\mathbb{R}$. If $b>0$, then best response iterations through \eqref{eq:LloydMaxIter} starting with $N$ bins always converge to the unique equilibrium with $N$ bins. 
\item Consider a source that has a strictly log-concave density with a support on $(-\infty,m_U]$ for some $m_U\in\mathbb{R}$. If $b<0$, then best response iterations through \eqref{eq:LloydMaxIter} starting with $N$ bins always converge to the unique equilibrium with $N$ bins. 
\end{enumerate}
\end{theorem}
\begin{proof}
The result essentially follows from Theorem~\ref{thm:convergence} after observing that under the given assumptions, the bin edges are always greater than the lower boundary or larger than the upper boundary during best response iterations. 
\end{proof}

For sources with semi-unbounded support, if the sign of $b$ does not satisfy the conditions of Theorem~\ref{thm:convergenceSemiInf}, convergence is not guaranteed, see the following remark. This is in contrast to the case of no bias \cite{kieffer83}.

\begin{remark}
For sources with semi-unbounded support, even if there exists an equilibrium with $N$ bins for a given bias $b$, an initial set of bin edges forming a partition with $N$ bins may not converge to the unique equilibrium with $N$ bins. For instance, for a source with a support on $[m_L,\infty)$, if the initial values for the first two bin edges (i.e., $m_1$ and $m_2$ with $m_L<m_1<m_2$) are close to $m_L$ and the bias is negative (i.e., $b<0$), then after the first iteration, the resulting first bin edge may be less than the value of $m_L$. As a result, the sender reduces the number of bin edges during the best response iterations. This implies that even though iterations start with $N$ bins, they may converge to an equilibrium with $K$ bins with $K<N$. In this case, it is guaranteed that best response iterations converge to an informative equilibrium with at least two bins (assuming it exists).   
\end{remark}
A similar result also holds for sources with bounded support.
\begin{remark}
For sources with bounded support, best response iterations starting with $N$ bins may not converge to the unique equilibrium with $N$ bins regardless of the value or the sign of $b$. Nevertheless, best response iterations always converge to an informative equilibrium with at least two bins (assuming it exists).
\end{remark}

\section{Nash Equilibria for Sources with More General Distributions}\label{sec:general}

In this section, the number of bins at the equilibria is investigated for sources with more general distributions satisfying certain assumptions. First, sources with semi-unbounded support are considered, then the results on the distributions with two-sided unbounded support are presented. 

\subsection{Sources with Semi-Unbounded Support}

Before the analysis, we make the following assumption on the source\footnote{Even though we consider sources supported on $[a,\infty)$, the results in this section can be extended for sources supported on the interval $(-\infty,a]$.}:
\begin{assumption}\label{assum:semiInfinite}
The distribution of the source $M$ satisfies the following conditions:	
\begin{enumerate}[(i)]
\item \textit{(Continuous density)} The source admits a continuous density. 
\item \textit{(Finite mean with semi-unbounded support)} The source is supported on the interval $[a,\infty)$ and $\mathbb{E}[M]=\mu$ with $\mu\in\mathbb{R}$.
\item \textit{(Monotonicity of the centroid at the tail)} There exist $K\geq a$ and $\eta\geq0$ such that for any $t \geq K$, we have $\mathbb{E}[M|M\geq t] \leq t + \eta$.
\item \textit{(Monotonicity of the pdf at the tail)} $M$ has a monotonically decreasing pdf for $M\geq K$; in particular, for any $t \geq K$ and $h\geq 0$, we have $\mathbb{E}[M| t \leq M \leq t+h] \leq t+h/2 $.
\end{enumerate}
\end{assumption}

\begin{theorem} \label{FinThm}
Under Assumption~\ref{assum:semiInfinite}, at an equilibrium,
\begin{enumerate}[(i)]
\item \ekRev{if $b<0$, there can be at most $\lfloor \frac{\eta + (K-a)}{2|b|}+3 \rfloor$ bins and the bins in the interval $[K,\infty)$ have monotonically increasing bin-lengths.}
\item if $b>0$, there exist equilibria with infinitely many bins. 
\end{enumerate}
\end{theorem}
\begin{proof}
See Section~\ref{ProofFinThm}.
\end{proof}

Although Theorem~\ref{FinThm} provides an upper bound on the number of bins at the equilibria, it does not give a necessary condition for the non-informative equilibrium. The following theorem provides the details:
\begin{theorem} \label{FinThm2}
Under Assumption~\ref{assum:semiInfinite}, if $b\leq-{K+\eta-a\over2}$, there cannot be any informative equilibria; i.e., there exist only non-informative equilibria.
\end{theorem}
\begin{proof}
See Section~\ref{ProofFinThm2}.
\end{proof}

\subsection{Sources with Two-Sided Unbounded Support}

Before the analysis, we make the following assumption on the source:
\begin{assumption}\label{assum:doubleInfinite}
The distribution of the source $M$ satisfies the following conditions:	
\begin{enumerate}[(i)]
\item \textit{(Continuous density)} The source admits a continuous density. 
\item \textit{(Finite mean with two-sided unbounded support)} The source is supported on the interval $(-\infty,\infty)$ and $\mathbb{E}[M]=\mu$ with $\mu\in\mathbb{R}$.
\item \textit{(Monotonicity at the positive tail)} There exist real $K$ and $\eta\geq0$ such that for any $t \geq K$, we have that $\mathbb{E}[M|M\geq t] \leq t + \eta$, and $M$ has a monotonically decreasing pdf for $M\geq K$; in particular, for any $t \geq K$ and $h\geq 0$, we have that $\mathbb{E}[M| t \leq M \leq t+h]  \leq t+h/2$.
\item \textit{(Monotonicity at the negative tail)} There exist real $S\leq K$ and $\nu\geq0$ such that for any $t \leq S$, we have that $\mathbb{E}[M|M\leq t] \geq t - \nu$, and $M$ has a monotonically increasing pdf for $M\leq S$; in particular, for any $t \leq S$ and $h\geq 0$, we have that $\mathbb{E}[M| t-h \leq M \leq t] \geq t +h/2$.  
\end{enumerate}
\end{assumption}

The following result proves the bounds on the number of bins and monotonicity of bin-lengths for the intervals $(-\infty,S]$ and $[K,\infty)$ depending on the sign of the bias $b$:
\begin{proposition} \label{prop:twoSidedMonotonicity}
Under Assumption~\ref{assum:doubleInfinite}, in an equilibrium,
\begin{enumerate}[(i)]
\item if $b>0$, the number of bins in the interval $(-\infty,S]$ is upper bounded, and the bin-lengths are monotonically decreasing for those bins. 
\item if $b<0$, the number of bins in the interval $[K,\infty)$ is upper bounded, and the bin-lengths are monotonically increasing for those bins.
\end{enumerate}
\end{proposition}
\begin{proof}
See Section~\ref{ProofTwoSidedMonotonicity}.
\end{proof}

\ekRev{For sources with two-sided unbounded support satisfying Assumption~\ref{assum:doubleInfinite}, there does not exist a finite upper bound on the number of bins in equilibrium as shown in the following:}
\begin{theorem} \label{thm:twoSidedInfinite}
Under Assumption~\ref{assum:doubleInfinite}, there exist equilibria with infinitely many bins. 
\end{theorem}
\begin{proof}
See Section~\ref{ProofTwoSidedInfinite}.
\end{proof}

\subsection{Existence of Equilibria with Infinitely Many Bins for Sources with Log-Concave Distributions}

In Theorem~\ref{thm:LogConcaveTwoSidedExistence} and Theorem~\ref{thm:LogConcaveOneSidedExistence}, we investigate existence of equilibria with finitely many bins for sources with a log-concave density. In fact, in the case when there does not exist a finite upper bound on the number of bins in equilibrium, the existence of equilibria with infinitely many bins can be shown by using the results in this section. Towards that goal, we show that the assumptions made in this section are satisfied for a log-concave density. 

\begin{lemma}\label{lem:logConcaveMoreGenaral}
\begin{enumerate}[(i)]
\item A log-concave density with a support on $[a,\infty)$ for $a\in\mathbb{R}$ satisfies Assumption~\ref{assum:semiInfinite}. 
\item A log-concave density with two-sided unbounded support satisfies Assumption~\ref{assum:doubleInfinite}. 
\end{enumerate}

\end{lemma}
\begin{proof}
The following considers sources with semi-unbounded support and the proof for two-sided unbounded support is similar. We know that a log-concave density must be continuous. In addition, due to Lemma~\ref{lem:expTail}, tails of a log-concave density are at most exponential. This implies that the mean must be finite. Moreover, since $(\mathbb{E}[M|M\geq t] - t) $ is decreasing in $t$ \cite{bagnoli2005log}, it follows that $\mathbb{E}[M|M\geq t] \leq t+(\mu-a) $ where $\mu=\mathbb{E}[M]$. Finally, the fourth assumption follows from the fact that a log-concave density is unimodal \cite{IbragimovUnimodal56}. 
\end{proof}

The following is a corollary of Theorem~\ref{FinThm} and Lemma~\ref{lem:logConcaveMoreGenaral}.
\begin{corollary}\label{cor:infManyLogConcaveOneSided}
Consider a source with a log-concave density supported on $[a,\infty)$ for some $a\in\mathbb{R}$. Then, there always exist equilibria with infinitely many bins.
\end{corollary}

The following is a corollary of Theorem~\ref{thm:twoSidedInfinite} and Lemma~\ref{lem:logConcaveMoreGenaral}.

\begin{corollary}\label{cor:infManyLogConcaveTwoSided}
Consider a source with a log-concave density with two-sided unbounded support. Then, there always exist equilibria with infinitely many bins.
\end{corollary}

\begin{remark}
It is seen that the class of distributions satisfying Assumption~\ref{assum:semiInfinite} or Assumption~\ref{assum:doubleInfinite} are more general distributions than the class of log-concave distributions. For instance, Assumption~\ref{assum:semiInfinite} is satisfied when log-concavity is assumed only at the tail with an additional continuity and finite mean assumptions. Therefore, one can construct a density which is not log-concave but satisfies Assumption~\ref{assum:semiInfinite} or Assumption~\ref{assum:doubleInfinite}. 
\end{remark}

\section{Nash Equilibria for Sources with Special Log-Concave Distributions: Exponential, Gaussian and Uniform Cases}\label{sec:specialLogConcave}

\ekRev{In Section~\ref{sec:LogConcave}}, we have investigated the cheap talk problem for general log-concave sources. In this section, we now focus on special distributions with a log-concave density and present more specific results.

\subsection{Exponential Distribution} \label{sec:exp}

In this subsection, the source is assumed to be exponential and the number of bins at the equilibria is investigated. Before delving into the technical results, we observe the following fact:
\begin{fact} \label{fact:exponential}
Let $M$ be an exponentially distributed random variable with a positive parameter $\lambda$, i.e., the probability density function (pdf) of $M$ is $f(m)=\lambda\me^{-\lambda m}$ for $m\geq0$. The expectation and the variance of an exponential random variable conditioned on the interval $[a,b]$ are $\mathbb{E}[M|a<M<b]={1\over\lambda}+a-{b-a\over\me^{\lambda (b-a)}-1}$ and $\mathrm{Var}\left(M|a<M<b\right)={1\over\lambda^2} - {(b-a)^2\over\me^{\lambda (b-a)}+\me^{-\lambda (b-a)}-2}$, respectively.
\end{fact}

Since the exponential distribution is log-concave, the following result is a corollary of Theorem~\ref{thm:LogConcaveOneSidedExistence}.

\begin{corollary}\label{cor:expExistence}
\ekRev{Consider an exponentially distributed source with parameter $\lambda$. If $b>0$, then $N_{\max}(b)=\infty$, and for any $N\geq1$, there exists a unique equilibrium with $N$ bins. On the other hand, if $b<0$, then $N_{\max}(b)<\infty$, and for any $N$ satisfying $1\leq N\leq N_{\max}(b)$, there exists a unique equilibrium with $N$ bins.}
\end{corollary}

In the case of a negative bias, the number of bins at the equilibrium is upper bounded, as stated in Corollary~\ref{cor:expExistence}, and the following proposition provides an upper bound on the number of bins in this case. In fact, this proposition presents an alternative proof for the fact that the number of bins at the equilibrium is upper bounded when $b<0$. Here, $\lfloor x\rfloor$ denotes the largest integer less than or equal to $x$.

\begin{proposition} \label{prop:expProp1}
Suppose $M$ is exponentially distributed with parameter $\lambda$. Then, for $b<0$, any Nash equilibrium has at most $\lfloor -\frac{1}{2b\lambda} + 1\rfloor$ bins.
\end{proposition}
\begin{proof}
See Section~\ref{sec:expProp1}.
\end{proof}

Next, we investigate the structure of the bin edges at the equilibrium. 

\begin{proposition}\label{prop:expProp2}
Consider an exponentially distributed source with parameter $\lambda$. Let $N\geq 1$ be given and suppose that the parameters $b$ and $\lambda$ are such that there exists an equilibrium with $N$ bins. Then, at this unique equilibrium, the following holds: 
\begin{enumerate}[(i)]
\item The bin-lengths $l_1=m_1$, $l_2=m_2-m_1$, $\dots$, $l_{N-1}=m_{N-1}-m_{N-2}$ satisfy the following:
\begin{subequations}
\begin{gather}
g(l_{N-1})={2\over\lambda}+2b ,\label{eq:compactRecursionLast}\\
g(l_k) 	= {2\over\lambda}+2b-h(l_{k+1}), \quad\text{for } k=1,2,\dots,N-2,\label{eq:compactRecursion}
\end{gather}
\end{subequations}
where $g(x) \triangleq \frac{xe^{\lambda x}}{e^{\lambda x}-1}$ and $h(x) \triangleq \frac{x}{e^{\lambda x}-1}$. 
\item The bin-lengths are monotonically increasing, i.e., $l_1< l_2< \dots< l_{N-1}$.
\end{enumerate}
\end{proposition}
\begin{proof}
See Section~\ref{sec:expProp2}.
\end{proof}

When the upper bound on the number of bins is investigated further, it is possible to derive conditions for the existence of equilibria with two or more bins. For instance, Theorem~\ref{thm:logConcaveNoninformative} gives a necessary and sufficient condition for the existence of informative equilibria for general log-concave sources in terms of the mean. In the special case of exponential sources, this translates to the condition $b>-{1\over2\lambda}$, which ensures the existence of equilibria with two bins. The following theorem states this result and further refines it for the existence of an equilibrium with three bins where the result is obtained by characterizing the equilibrium. In addition, it is possible to construct an equilibrium with infinitely many bins when $b>0$, as proven in the following theorem.

\begin{theorem}\label{thm:expExistence}
When the source has an exponential distribution with parameter $\lambda$, the following holds:
\begin{enumerate}[(i)]
\item There exists an equilibrium with at least two bins if and only if $b>-{1\over2\lambda}$. Otherwise, the only equilibrium is non-informative. 
\item There exists an equilibrium with at least three bins if and only if $b>-{1\over2\lambda}{e-2\over e-1}$.
\item For $b>0$, there exists a unique equilibrium with infinitely many bins. In particular, all bins must have a length of $l^*$, where $l^*$ is the solution to $g(l^*) = {2\over\lambda}+2b-h(l^*)$, with $g(x) = \frac{xe^{\lambda x}}{e^{\lambda x}-1}$ and $h(x) = \frac{x}{e^{\lambda x}-1}$.
\end{enumerate}
\end{theorem}
\begin{proof}
See Section~\ref{sec:expExistence}.
\end{proof}

Theorem~\ref{thm:LogConcaveCost} shows that an equilibrium with more bins is more informative for any log-concave density as the equilibrium costs of the encoder and the decoder monotonically decreases with the number of bins. The following extends this result considering an equilibrium with finitely many bins when $b>0$.

\begin{theorem} \label{thm:expMoreInformative}
When the source has an exponential distribution with parameter $\lambda$, the smallest expected equilibrium costs are attained with the maximum possible number of bins:
\begin{enumerate}[(i)]
\item for equilibria with $K$ and $N$ bins where $N>K$, the equilibrium with $N$ bins induces smaller expected costs in equilibrium for both of the players regardless of the value of $b$.
\item for $b>0$, the equilibrium with infinitely many bins yields the smallest expected costs in equilibrium for both of the players. 
\end{enumerate} 
\end{theorem}

\begin{proof}
See Section~\ref{sec:expMoreInformative}.
\end{proof}

Theorem~\ref{thm:expMoreInformative} implies that the unique equilibrium with maximum number of bins payoff dominates all other equilibria \cite{eqSelectionBook}.

\subsection{Gaussian Distribution}\label{sec:Gaussian}

Let $M$ be a Gaussian random variable with mean $\mu$ and variance $\sigma^2$; i.e., $M\sim\mathcal{N}(\mu,\sigma^2)$. Let $\phi(m)={1\over\sqrt{2\pi}}e^{-{m^2\over2}}$ be the pdf of a standard Gaussian random variable, and let $\Phi(b)=\int_{-\infty}^{b}\phi(m){\rm{d}}m$ be its cumulative distribution function (cdf). Then, the expectation of a truncated Gaussian random variable is the following:
\begin{fact}\label{fact:Gauss}
The mean of a Gaussian random variable $M\sim\mathcal{N}(\mu,\sigma^2)$ conditioned on the interval $[a,b]$ is $\mathbb{E}[M|a<M<b]=\mu - \sigma{\phi({b-\mu\over\sigma})-\phi({a-\mu\over\sigma})\over\Phi({b-\mu\over\sigma})-\Phi({a-\mu\over\sigma})}$.
\end{fact}

Since the Gaussian pdf is log-concave with two-sided unbounded support, we state the following result as a corollary of Theorem~\ref{thm:LogConcaveTwoSidedExistence}, \ekRev{Corollary~\ref{cor:infManyLogConcaveTwoSided}} and Theorem~\ref{thm:LogConcaveCost}.
\begin{corollary}\label{cor:gauss}
When the source has a Gaussian distribution as $M\sim\mathcal{N}(\mu,\sigma^2)$, for any $N\geq 1$, there always exists a unique equilibrium with $N$ bins regardless of the value of $b$. \ekRev{In addition, there always exist equilibria with infinitely many bins.} Moreover, the expected costs of the encoder and the decoder in equilibrium \ekRev{decrease} when more bins are used.
\end{corollary}

Since the pdf of a Gaussian random variable is symmetrical about its mean $\mu$, and monotonically decreasing in the interval $[\mu,\infty)$, the following can be obtained in a similar manner to Proposition~\ref{prop:twoSidedMonotonicity}:

\begin{proposition} \label{prop:gaussMonotonicity}
Consider the unique equilibrium with $N$ bins for a Gaussian source $M\sim\mathcal{N}(\mu,\sigma^2)$. Then,
\begin{enumerate}[(i)] 
\item if $b<0$, the bin-lengths in the interval $[\mu,\infty)$ are monotonically increasing and the number of bins in this interval is upper bounded by $\sigma\over2|b|$,
\item if $b>0$, the bin-lengths in the interval $(-\infty,\mu]$ are monotonically decreasing and the number of bins in this interval is upper bounded by $\sigma\over2b$.
\end{enumerate}
\end{proposition}
\begin{proof}
See Section~\ref{ProofGaussMonotonicity}.
\end{proof}

%
%

\ekRev{As Corollary~\ref{cor:gauss} states, there exist equilibria with infinitely many bins when the source is Gaussian. The following states a property related to the bin-lengths for such equilibria.}

\begin{remark}\label{rem:gaussInfBinLength}
At the equilibrium with infinitely many bins, as the bin edges get very large in absolute value (i.e., $m_i\rightarrow\infty$ for $b>0$ and $m_i\rightarrow-\infty$ for $b<0$ as $i\to\infty$), the bin-lengths converge to $2|b|$.
\end{remark}
\begin{proof}
See Section~\ref{ProofGaussInfBinLength}.
\end{proof}

\subsection{Uniform Distribution}\label{sec:uniform}

As an example for source with a compactly supported density, we restate results by Crawford and Sobel \cite[Section~4]{SignalingGames} for the interesting scenario with a uniformly distributed source. The following result provides an expression for the upper bound on the number of bins. Here, $\lceil x\rceil$ denotes the smallest integer greater than or equal to $x$.

\begin{theorem}[Crawford and Sobel \cite{SignalingGames}]\label{thm:uniformExistenceUniqueness}
When $M$ is uniformly distributed in the interval $[0,1]$, the upper bound on the number of bins at the equilibrium is given by $N_{\max}(b)=\bigg\lceil -\frac{1}{2}+\frac{1}{2}\Big(1+\frac{2}{|b|}\Big)^{1/2}\bigg\rceil$, and for any $N\in\{1,\dots,N_{\max}(b)\}$, there exists a unique equilibrium with $N$ bins. In particular, when $|b|\geq 1/4$, the only equilibrium is non-informative with a single bin. 
\end{theorem}

Since the uniform distribution can be viewed to be log-concave, an equilibrium with more bins induces smaller expected costs for both the encoder and decoder compared to an equilibrium with less bins, as stated in the following:

\begin{theorem}[Crawford and Sobel \cite{SignalingGames}]\label{thm:uniformCost}
When $M$ is uniformly distributed in the interval $[0,1]$, the expected costs of the encoder and the decoder in equilibrium decrease as the number of bins increases.
\end{theorem}

One way of proving Theorem~\ref{thm:uniformCost} is to express the costs at the equilibrium and analyze the behavior of the cost with respect to the number of bins, as mentioned in \cite{SignalingGames}. Nevertheless, \cite{SignalingGames} shows this result by using the fact that the monotonicity condition (M) is satisfied for the uniform source scenario.

\section{Concluding Remarks}

In this paper, the Nash equilibria of cheap talk have been investigated for general log-concave sources. It has been shown that for any finite $N$, there exists a unique equilibrium with $N$ bins for sources with two-sided unbounded support. Similarly, it has been proven that for sources with semi-unbounded support, when the support extends to $+\infty$ and $b>0$, or the support extends to $-\infty$ and $b<0$, for any finite $N$, there exists a unique equilibrium with $N$ bins. On the other hand, in the converse case for sources with semi-unbounded support, it has been proven that there exists a \ekRev{finite} upper bound on the number of bins in an equilibrium. Furthermore, it has been shown that, as the number of bins increases, the expected costs of the encoder and the decoder in equilibrium \ekRev{decrease}, i.e., a more informative equilibrium is obtained. Moreover, we have proven that best response iterations converge to the unique equilibrium for sources with strictly log-concave density and two-sided unbounded support, generalizing the classical convergence results in optimal quantization via Lloyd's Method I. In addition, it has been shown that for exponential sources there exists a unique equilibrium with infinitely many bins when $b>0$. Moreover, sources which are not necessarily log-concave but satisfy certain assumptions have also been considered. For these source, the existence of a \ekRev{finite} upper bound on the number of bins has been investigated, and a \ekRev{finite} upper bound has been derived whenever it exists. For the case that there is no \ekRev{finite} upper bound, the existence of equilibria with infinitely many bins has been established. 

\section{Proofs}

\subsection{Proof of Theorem~\ref{thm:LogConcaveTwoSidedExistence}}\label{ProofLogConcaveTwoSidedExistence}

We first discuss the existence of an equilibrium with a single bin. A quantizer with a single bin means that the encoder does not reveal any information related to the source. The optimal receiver action in this case becomes the mean of the source. This pair of encoding and decoding policies leads to a Nash equilibrium. In this case, since the encoder does not convey information related to the source, the resulting equilibrium is a babbling equilibrium. Moreover, such an equilibrium always exists regardless of whether the source distribution is log-concave or not.

In order to prove that there exists an equilibrium with $N$ bins where $N\geq 2$, we take the following approach. We first show the existence of an equilibrium with two bins and then show that one can obtain an equilibrium with $N+1$ bins starting from an equilibrium with $N$ bins. These reveal that there exists an equilibrium with any finite number of bins. In addition, while obtaining the existence result, we also prove the uniqueness of an equilibrium for a given number of bins. In the following proof, we assume that the source is strictly log-concave. Nevertheless, our result also holds for a source with a log-concave density. In this case, one of the inequalities in Lemma~\ref{lem:monotonicity} must be strict for the result to hold and this is indeed the case for a density with two-sided unbounded support.

To begin with, we restate the centroid condition and nearest neighbor condition that need to be satisfied at an equilibrium. An ordered set of bin edges $m_1<\dots<m_{N-1}$ leads to an equilibrium with $N$ bins if the following conditions are satisfied:
\begin{gather}
\int_{-\infty}^{m_1} (u_1-m)f(m) dm = 0,\label{eq:thm1CentroidFirstBin} \\
\int_{m_i}^{m_{i+1}} (u_{i+1}-m)f(m) dm = 0,\quad \text{for }i=1,\dots,N-2,\label{eq:thm1CentroidithBin}\\
\int_{m_{N-1}}^{\infty} (u_N-m)f(m) dm = 0,\label{eq:thm1CentroidLastBin}\\
u_{i+1} = 2m_{i}-u_{i}-2b,\quad \text{for }i=1,\dots,N-1 \label{eq:thm1Indifference}
\end{gather}
where $u_1,\dots,u_N$ denote receiver actions. The aim is to show the existence of bin edges $m_1,\dots,m_{N-1}$ and receiver actions $u_1,\dots,u_N$ that satisfy \eqref{eq:thm1CentroidFirstBin}-\eqref{eq:thm1Indifference} for any $N\geq 2$ and also to prove that these bin edges and receiver actions are indeed unique. To that aim, we need some auxiliary results which are presented in the following lemmas. The first lemma derives a set of inequalities which always holds for a strictly log-concave density.  

\begin{lemma} \label{lem:inequalities}
For a random variable $M$ with a strictly log-concave density $f$, the following inequalities are always satisfied:
\begin{align}
\mathrm{Pr}(M<u) > &(u-\mathbb{E}[M | M < u]) f(u)\quad\text{for any }u\in\mathbb{R} 
\label{eq:Lemma1FirstBin}\\
\mathrm{Pr}(M>u) > &(\mathbb{E}[M | u< M]-u) f(u)\quad\text{for any }u\in\mathbb{R}
\label{eq:Lemma1LastBin}\\
\mathrm{Pr}(v<M<u)&> (u-\mathbb{E}[M | v< M < u])f(u) \nonumber \\  
&+ (\mathbb{E}[M | v< M < u]-v)f(v)
\label{eq:Lemma1ithBin}
\end{align}
for any $u,v\in\mathbb{R}$ with $v<u$.
\end{lemma}
Before presenting the proof, we note that a variant of the inequalities \eqref{eq:Lemma1FirstBin} and \eqref{eq:Lemma1LastBin} can be found in the literature considering log-concave densities with bounded support \cite{bagnoli2005log}. On the other hand, these inequalities need to be derived for the general case of densities with two-sided unbounded support for our purpose. Moreover, it is required to prove the additional property in \eqref{eq:Lemma1ithBin}, which is crucial in obtaining the monotonic behavior of the bin edges $m_2,\dots,m_N$ with respect to a change in $m_1$ in Lemma~\ref{lem:monotonicity}.

\begin{proof}
We note that the inequalities \eqref{eq:Lemma1FirstBin} and \eqref{eq:Lemma1LastBin} follow from \eqref{eq:Lemma1ithBin} and Lemma~\ref{lem:expTail}, where the latter states that a log-concave density has at most exponential tails. This exponential tail property implies that $uf(u)\to 0$ as $u\to\infty$ and that $vf(v)\to 0$ as $v\to-\infty$. In addition, since a source with a log-concave density has a finite mean, for a given $v\in\mathbb{R}$, we have that $(\mathbb{E}[M | v< M < u])f(u)\to 0$ as $u\to\infty$. Similarly, for a given $u\in\mathbb{R}$, we have that $(\mathbb{E}[M | v< M < u])f(v)\to 0$ as $v\to-\infty$. Hence, by taking $u\to\infty$ in \eqref{eq:Lemma1ithBin}, the inequality in \eqref{eq:Lemma1LastBin} is proven. Similarly, by taking $v\to-\infty$ in \eqref{eq:Lemma1ithBin}, we obtain the inequality in \eqref{eq:Lemma1FirstBin}. Therefore, in order to prove the lemma, it suffices to prove that \eqref{eq:Lemma1ithBin} holds. In the proof, we use the property that for a differentiable and strictly log-concave density $f$, we have that $\frac{f'(x)}{f(x)}$ is strictly monotonically decreasing in $x$, which directly follows from the definition of being log-concave. In the following, we first prove the result under an additional differentiable density assumption and then show that the result holds for any log-concave density.

It is seen that the inequality in \eqref{eq:Lemma1ithBin} is equivalent to 
\begin{align}
g(u) &\triangleq (F(u)-F(v))^2 +(u-v)\Big(F(v)f(u)-F(u)f(v)\Big) \nonumber \\
&-f(u)\int_{v}^{u} F(m)dm 
+f(v)\int_{v}^{u} F(m)dm
>0. \label{eq:defg}
\end{align}
For fixed $v$, the derivative of $g(u)$ with respect to $u$ is given by
\begin{align}
g'(u)
&= f(u)(F(u)-F(v)) - f'(u)\int_{v}^{u} F(m) dm \nonumber \\
&\hphantom{=}+(u-v) \Big(F(v)f'(u)-f(u)f(v)\Big) \nonumber \\
&=
\int_{v}^{u}
h(m) 
dm,\label{eq:derivativeOfg}
\end{align}
where 
\begin{align*}
h(m) \triangleq 
f(u)f(m) - f'(u)F(m) + F(v)f'(u) - f(u)f(v).
\end{align*}
Next, we show that the integrand $h(m)$ is always positive for $m>v$:
\begin{align}
h(m) 
&= f(u)\big(f(m)-f(v)\big)
- f'(u) \int_{v}^{m}f(t) dt \nonumber \\
&= f(u)\big(f(m)-f(v)\big)
- f(u) \frac{f'(u)}{f(u)} \int_{v}^{m}f(t) dt  \nonumber \\
&> f(u)\big(f(m)-f(v)\big)
- f(u)  \int_{v}^{m}\frac{f'(t)}{f(t)}f(t) dt  \nonumber \\
& = f(u)\big(f(m)-f(v)\big)
- f(u) \big(f(m)-f(v)\big) = 0, \label{eq:ineqforh}
\end{align}
where the inequality follows from the fact that $f$ is a strictly log-concave density. This implies that $g'(u)>0$ when $u>v$. Since $g(u)=0$ when $u=v$ and $g'(u)>0$ when $u>v$, it follows that $g(u)>0$ for all $u>v$. This shows that \eqref{eq:Lemma1ithBin} is satisfied for all $u>v$.

Now, we show that the result holds for any strictly log-concave density. Since $f$ is log-concave, it follows that $\log f$ is left and right differentiable. This implies that $f$ is left and right differentiable, as well. Denote the left and right derivative of $f$ by $f'_{-}$ and $f'_{+}$, respectively. Let $\log f(x) \triangleq \ell(x)$ for notational convenience. Since $f$ is strictly log-concave, we have that $\ell'_{-}(x)= \frac{f'_{-}(x)}{f(x)}$ and $\ell'_{+}(x)= \frac{f'_{+}(x)}{f(x)}$ are strictly monotonically decreasing in $x$. We also know that $g$ defined in \eqref{eq:defg} has left and right derivatives since $f$ has left and right derivatives. Then, by taking the left derivative of $g$ with respect to $u$ for fixed $v$, we get
\begin{align*}
g'_{-}(u) = \int_{v}^{u} \Big( &f(u)f(m) -f'_{-}(u)F(m)\\
& + F(v)f'_{-}(u)-f(u)f(v) \Big)dm.
\end{align*}
Then, a similar analysis as in \eqref{eq:ineqforh} yields $g'_{-}(u)>0$ when $u>v$, where we use the fact that $\frac{f'_{-}(x)}{f(x)}$ is strictly monotonically decreasing and that $\int_v^m f'_{-}(t)dt = f(m)-f(v)$ since $\log f$, and therefore $e^{\log f}$, is an absolutely continuous function in the interior of the support of $f$ (see e.g. \cite[p.~34-35]{niculescu2018convex} for the absolute continuity of convex/concave functions). Also, we analogously obtain $g'_{+}(u)>0$ when $u>v$. We also know that $g(u)=0$ if $u=v$. Hence, we get $g(u)>0$ for all $u>v$, as desired.
\end{proof}

Next, we use the inequalities derived in Lemma~\ref{lem:inequalities} to establish a monotonicity property of centroids and subsequent bin edges with respect to a shift in the first bin edge. Before presenting the lemma, for notational convenience, we make the following definitions which specify a collection consisting of a set of bin edges either in an ascending or descending order such that the nearest neighbor conditions with the corresponding centroids are satisfied at all bin edges except at the last bin edge. 

\begin{definition}
\begin{enumerate}
\item[(i)] Let $m_1<m_2<\cdots<m_{N}$, $u_1 = \mathbb{E}[M|M<m_1]$ and $u_i = \mathbb{E}[M|m_{i-1}\leq M <m_i]$ for $i=2,\dots,N$. Define $\mathcal{M}_A^N$ as the collection of all $\{m_i\}_{i=1}^N$ for which the nearest neighbor conditions at the boundaries $m_1,\dots,m_{N-1}$ are satisfied with the corresponding centroids $u_1,\dots,u_N$, i.e., $2m_i = u_i+u_{i+1}+2b$ for $i=1,\dots,N-1$.
\item[(ii)] Let $m_1>m_2>\cdots>m_N$, $u_1 = \mathbb{E}[M|m_1\leq M]$ and $u_i = \mathbb{E}[M|m_i\leq M <m_{i-1}]$ for $i=2,\dots,N$. Define $\mathcal{M}_D^N$ as the collection of all $\{m_i\}_{i=1}^N$ for which the nearest neighbor conditions at the boundaries $m_1,\dots,m_{N-1}$ are satisfied with the corresponding centroids $u_1,\dots,u_N$, i.e., $2m_i = u_i+u_{i+1}+2b$ for $i=1,\dots,N-1$. 
\end{enumerate}
\end{definition}

Notice that when a set of bin edges belongs to $\mathcal{M}_A^N$ or $\mathcal{M}_D^N$, it does not necessarily mean that this set of bin edges leads to an equilibrium. For this set of bin edges to form an equilibrium, the nearest neighbor condition at the last bin edge should also be satisfied. In our proof, we show that by varying the value of the first bin edge, it is possible to satisfy the last nearest neighbor condition. In fact, in the proof, it will be apparent that there exists a unique value for the first bin edge that leads to an equilibrium with $N$ bins and these unique values are monotone in $N$. This monotonicity property of the bin edges is also used to prove monotonically decreasing behavior of the expected costs with respect to the number of bins in Theorem~\ref{thm:LogConcaveCost}.

\begin{lemma}\label{lem:monotonicity}
Suppose that there exists a set of bin edges $\{m_i\}_{i=1}^N$ which belongs to $\mathcal{M}_A^N$ (resp. $\mathcal{M}_D^N$) with the corresponding set of centroids $\{u_i\}_{i=1}^N$. Let $\bar{u}_{N+1} = \mathbb{E}[M|m_N\leq M]$ (resp. $\bar{u}_{N+1} = \mathbb{E}[M|M < m_N]$) denote the centroid of the last bin and let $\tilde{u}_{N+1} = 2 m_N -u_N-2b$ denote the receiver action that satisfies the nearest neighbor condition at the bin edge $m_N$. Then, the following hold:
\begin{gather}
\frac{dm_i}{du_i} > 1,\quad \text{for }i=1,\dots,N,\label{eq:Lemma2FirstIneq}\\
\frac{du_{i+1}}{dm_i}> 1,\quad \text{for }i=1,\dots,N-1,\label{eq:Lemma2SecondIneq}\\
\frac{d\tilde{u}_{N+1}}{dm_N} > 1,\label{eq:Lemma2ThirdIneq}\\
\frac{dm_N}{d\bar{u}_{N+1}} > 1.\label{eq:Lemma2FourthIneq}
\end{gather}
\end{lemma}

\begin{proof}
We prove the result for the case of $\{m_i\}_{i=1}^N \in \mathcal{M}_A^N$ and the case of $\{m_i\}_{i=1}^N \in \mathcal{M}_D^N$ can be proven in a similar manner. The centroid condition requires that
\begin{gather}
\int_{-\infty}^{m_1} (u_1-m)f(m) dm = 0, \label{eq:Lemma2FirstCentroid}\\
\int_{m_i}^{m_{i+1}} (u_{i+1}-m)f(m) dm = 0,\quad \text{for }i=1,\dots,N-1,\label{eq:Lemma2ithCentroid}\\
\int_{m_N}^{\infty} (\bar{u}_{N+1}-m)f(m) dm = 0.\label{eq:Lemma2LastCentroid}
\end{gather}
Since it is assumed that the nearest neighbor conditions at $m_1,\dots,m_{N-1}$ are satisfied, we have
\begin{align}
u_{i+1} = 2m_{i}-u_{i}-2b,\quad \text{for }i=1,2,\dots,N-1.\label{eq:Lemma2Indifference}
\end{align}
By differentiating \eqref{eq:Lemma2FirstCentroid} with respect to $u_1$, we get
\begin{align}
F(m_1) + u_1f(m_1)\frac{dm_1}{du_1} - \frac{d}{du_1}  \int_{-\infty}^{m_1}mf(m)dm =0, \label{eq:Lemma2FirstIntegral}
\end{align}
where the integral is finite since $\mathbb{E}[M]$ is finite. From \eqref{eq:Lemma2FirstIntegral}, it follows that
\begin{align}
\frac{dm_1}{du_1} 
&= 
\frac{\int_{-\infty}^{m_1}f(m)dm}{(m_1-u_1)f(m_1)}>1,
\end{align}
where $\frac{dm_1}{du_1}$ is finite and the inequality follows from \eqref{eq:Lemma1FirstBin} in Lemma~\ref{lem:inequalities}. Then, $\frac{dm_1}{du_1} >1$ and \eqref{eq:Lemma2Indifference} imply that $\frac{du_2}{dm_1}>1$. Next, we show that if $\frac{du_{i+1}}{dm_i}>1$, then $\frac{dm_{i+1}}{du_{i+1}}>1$  by using \eqref{eq:Lemma2ithCentroid}. Differentiating \eqref{eq:Lemma2ithCentroid} with respect to $u_{i+1}$ leads to 
\begin{align}
\int_{m_i}^{m_{i+1}} f(m)dm 
=&(u_{i+1}-m_i) f(m_i) \frac{dm_i}{du_{i+1}}\nonumber \\
&+(m_{i+1}-u_{i+1}) f(m_{i+1}) \frac{dm_{i+1}}{du_{i+1}}.\label{eq:Lemma2Eq1}
\end{align}
From Lemma~\ref{lem:inequalities}, we know that 
\begin{align}
\int_{m_i}^{m_{i+1}} f(m)dm 
>&(u_{i+1}-m_i) f(m_i)  \nonumber \\
&+(m_{i+1}-u_{i+1}) f(m_{i+1}).\label{eq:Lemma2Eq2}
\end{align}
By combining \eqref{eq:Lemma2Eq1} and \eqref{eq:Lemma2Eq2}, we get
\begin{align}
&(u_{i+1}-m_i) f(m_i) \frac{dm_i}{du_{i+1}} 
+(m_{i+1}-u_{i+1}) f(m_{i+1}) \frac{dm_{i+1}}{du_{i+1}} \nn \\
&\hphantom{PHANT}>(u_{i+1}-m_i) f(m_i)  
+(m_{i+1}-u_{i+1}) f(m_{i+1}) \nn \\
&\Leftrightarrow 
\left( \frac{dm_i}{du_{i+1}} -1\right) (u_{i+1}-m_i) f(m_i) \nn \\
&\hphantom{PHANT}>\left( 1-\frac{dm_{i+1}}{du_{i+1}}\right)(m_{i+1}-u_{i+1}) f(m_{i+1}).\label{eq:Lemma2Eq3}
\end{align}
From \eqref{eq:Lemma2Eq3}, it is seen that $\frac{du_{i+1}}{dm_i}>1$ leads to $\frac{dm_{i+1}}{du_{i+1}}>1$ since
$(u_{i+1}-m_i) f(m_i)>0$ and $(m_{i+1}-u_{i+1}) f(m_{i+1})>0$. Next, $\frac{dm_{i+1}}{du_{i+1}}>1$ and \eqref{eq:Lemma2Indifference} imply that $\frac{du_{i+2}}{dm_{i+1}}>1$. By iteratively using \eqref{eq:Lemma2Indifference}, \eqref{eq:Lemma2Eq1} and \eqref{eq:Lemma2Eq2} in this manner, one can show \eqref{eq:Lemma2FirstIneq} and \eqref{eq:Lemma2SecondIneq}. Moreover, $\frac{dm_N}{du_N}>1$ and $\tilde{u}_{N+1} = 2m_N-u_N-2b$ lead to \eqref{eq:Lemma2ThirdIneq}. Finally, differentiating \eqref{eq:Lemma2LastCentroid} with respect to $\bar{u}_{N+1}$, we get
\begin{align}
\frac{dm_N}{d\bar{u}_{N+1}} = 
\frac{\int_{m_N}^{\infty}f(m)dm}{(\bar{u}_{N+1}-m_N)f(m_N)}>1,
\end{align}
where $\frac{dm_N}{d\bar{u}_{N+1}}$ is finite and the inequality follows from Lemma~\ref{lem:inequalities}.
\end{proof}

Now, equipped with Lemma~\ref{lem:monotonicity}, we are ready to prove the theorem. We first prove the existence of a unique equilibrium with two bins. In this case, for a bin edge $m$, the centroids of the resulting bins are denoted by $u_1(m)=\mathbb{E}[M|M< m]$ and $u_2(m)=\mathbb{E}[M|m\leq M]$. The point where the sender is indifferent between these receiver actions is given by
\begin{align}
\tilde{m}(m)= \frac{u_1(m)+u_2(m)}{2} +b.\label{eq:mtilde}
\end{align}
For an equilibrium, it is required that $m=\tilde{m}(m)$. Lemma~\ref{lem:monotonicity} states that $\frac{dm}{du_1}>1$ and $\frac{dm}{du_2}>1$ hold due to the fact that the source distribution follows a strictly log-concave density and these imply that $\frac{dm}{d\tilde{m}}>1$ due to \eqref{eq:mtilde}. This last property is crucial in obtaining the uniqueness of the equilibrium. In the following, we construct an interval such that $\tilde{m}(\cdot)$ maps this interval into itself and then use the Brouwer’s fixed point theorem \cite{guide2006infinite} to show the existence of an equilibrium. Towards that goal, let $m^{\alpha}$ be a bin edge for which $m^\alpha \neq \tilde{m}(m^\alpha)$. Without loss of generality, suppose that $m^\alpha <\tilde{m}(m^\alpha)$. Let $m^\beta\triangleq u_2(m^\alpha)-m^\alpha+\mu+2b$. Note that $u_1(m^\alpha)$ and $u_2(m^\alpha)$ are finite for a finite $m^\alpha$ as the source has a finite mean. Observe that $m^\alpha<m^\beta$ holds since 
\begin{align*}
m^\alpha &< 2\tilde{m}(m^\alpha)-m^\alpha \\
&= u_2(m^\alpha)+u_1(m^\alpha)+2b-m^\alpha \\
&< u_2(m^\alpha)+\mu+2b-m^\alpha=m^\beta,
\end{align*}
where the first inequality is by assumption and the second inequality is due to $u_1(m)=\mathbb{E}[M|M<m]<\mathbb{E}[M|M<\infty]=\mu$. We also have $\tilde{m} (m^\beta) <m^\beta$ as shown in the following:
\begin{align*}
2(\tilde{m}(m^\beta)-m^\beta) 
&= u_2(m^\beta)-m^\beta +u_1(m^\beta)-m^\beta + 2b\\
&< u_2(m^\alpha)-m^\alpha +u_1(m^\beta)-m^\beta+2b\\
&< u_2(m^\alpha)-m^\alpha +\mu+2b-m^\beta=0,
\end{align*}
where the first inequality follows from the facts that $m^\alpha<m^\beta$ and $u_2(m)-m$ is monotonically decreasing in $m$ due to \eqref{eq:Lemma2FourthIneq} in Lemma~\ref{lem:monotonicity}, and the second inequality is due to $u_1(m)=\mathbb{E}[M|M<m]<\mathbb{E}[M|M<\infty]=\mu$. By using $\tilde{m} (m^\beta) <m^\beta$ together with the the fact that $\tilde{m}(m)$ is monotone in $m$, we have $\tilde{m}(m)\in(m^\alpha,m^\beta)$ for any $m\in[m^\alpha,m^\beta]$. Therefore, by Brouwer’s fixed point theorem \cite{guide2006infinite}, there exists $m\in (m^\alpha,m^\beta)$ such that $m=\tilde{m}(m)$ holds since $\tilde{m}(m)$ is continuous and maps the interval $[m^\alpha,m^\beta]$ into itself. The uniqueness, on the other hand, follows from $\frac{dm}{d\tilde{m}}>1$, which makes $(\tilde{m}(m)-m)$ monotonically decreasing in $m$. If the initial choice of bin edge $m^\alpha$ satisfies $\tilde{m}(m^\alpha)<m^\alpha$, a similar procedure can be applied. This shows the existence of a unique equilibrium with two bins.

Next, the aim is to prove that it is possible to obtain an equilibrium with $N+1$ bins starting from an equilibrium with $N$ bins. Depending on whether $b$ is positive or negative, we need to take slightly different approaches. In the following, we first focus on the case where $b$ is negative. Let 
\begin{align}
K_A(x) = \max \Big\{k \;|\; &\text{there exists }x=m_1 < \dots <m_{k-1}\nonumber \\ 
&\text{ with }\{m_i\}_{i=1}^{k-1}\in \mathcal{M}_A^{k-1} \Big\}.\label{eq:defK}
\end{align}
Denote the bin edges that attain the maximum in \eqref{eq:defK} for a given $x$ by $m_1^x,\dots,m_{K_A(x)-1}^x$ and denote the centroids with these bin edges by $u_1^x\triangleq\mathbb{E}[M|M<m_1^x]$ and $u_i^x\triangleq\mathbb{E}[M|m_{i-1}^x\leq M<m_i^x]$ for $i=2,\dots,K_A(x)-1$. In the following, an equilibrium with $N+1$ bins is obtained by continuously decreasing $m_1^x=x$ starting from its value for the equilibrium with $N$ bins. Let $m_i^A(N)$ for $i=1,\dots,N-1$ denote bin edges for an equilibrium with $N$ bins where these bin edges are sorted in an ascending order, i.e., $m_1^A(N)<\cdots<m_{N-1}^A(N)$.

To begin with, we show that when $m_1^x<m_1^A(N)$, there exists $m_1^x<\dots<m_N^x$ such that the nearest neighbor conditions at the boundaries $m_1^x,\dots,m_{N-1}^x$ hold with the corresponding centroids, i.e., $x<m_1^A(N)$ implies that $K_A(x)\geq N+1$. Towards that goal, we first show that for all $m_1^x<m_1^A(N)$ there exist $m_1^x<\cdots<m_{N-1}^x$ for which the nearest neighbor conditions at the boundaries $m_1^x,\dots,m_{N-2}^x$ are satisfied. These nearest neighbor conditions are in fact satisfied when $m_1^x=m_1^A(N)$ and our aim is to show the existence of boundaries which still satisfy these nearest neighbor conditions when $m_1^x$ is decreased. These nearest neighbor conditions require that
\begin{align}
u_{i+1}^x -m_i^x = m_i^x - u_i^x-2b,\quad\text{for }i=1,\dots,N-2. \label{eq:Thm1Indifferent}
\end{align}
It is always guaranteed that $u_1^x<m_1^x$ since $u_1^x=\mathbb{E}[M|M< m_1^x]$. From $u_1^x<m_1^x$, $b<0$ and \eqref{eq:Thm1Indifferent}, it follows that $u_1^x<m_2^x$. Then, it is always possible to find $m_2^x<m_2^A(N)$ such that the centroid of the bin $[m_1^x,m_2^x)$ is given by $u_2^x$, where $m_2^x<m_2^A(N)$ is due to the fact that $m_2^x$ is monotonically increasing in $m_1^x$. This implies that while $m_1^x$ is decreased $m_1^x<m_2^x$ is guaranteed. Next, via a similar argument, $u_i^x<m_i^x$, $b<0$ and \eqref{eq:Thm1Indifferent} imply that there exist $m_i^x<m_{i+1}^x<m_{i+1}^A(N)$ such that the nearest neighbor condition at the boundary $m_i^x$ holds for $i=2,\dots,N-2$. Now, we show the existence of the last bin edge $m_N^x$ for which the nearest neighbor condition at $m_{N-1}^x$ holds. Let $\bar{u}_N^x = \mathbb{E}[M | m_{N-1}^x \leq M ]$ and $\tilde{u}_N^x =2m_{N-1}^x-u_{N-1}^x-2b$ where the former term corresponds to the centroid of the last bin and the latter term is the receiver action that makes the nearest neighbor condition at $m_{N-1}^x$ hold. We know that $\bar{u}_N^x=\tilde{u}_N^x$ when $m_1^x=m_1^A(N)$ as this bin edge leads to an equilibrium with $N$ bins. Since $0< \frac{d\bar{u}_N^x}{dx} < \frac{d\tilde{u}_N^x}{dx}$ by Lemma~\ref{lem:monotonicity}, it follows that $\tilde{u}_N^x<\bar{u}_N^x$ for any $m_1^x<m_1^A(N)$. Moreover, from $u_{N-1}^x<m_{N-1}^x$, $\tilde{u}_N^x =2m_{N-1}^x-u_{N-1}^x-2b$ and $b<0$, it is guaranteed that $m_{N-1}^x<\tilde{u}_N^x$. Therefore, for any $m_1^x<m_1^A(N)$, there exists $m_N^x$ such that $\tilde{u}_N^x = \mathbb{E}[M|m_{N-1}^x\leq M < m_N^x]$. These establish the existence of boundaries $m_1^x,\dots,m_N^x$ for which the nearest neighbor conditions at $m_1^x,\dots,m_{N-1}^x$ are satisfied when $m_1^x<m_1^A(N)$, i.e., $x<m_1^A(N)$ leads to $K_A(x)\geq N+1$.

In the previous analysis, the existence of $m_1^x<\dots<m_N^x$ with $\{m_i^x\}_{i=1}^N\in \mathcal{M}_A^N$ is established when $m_1^x<m_1^A(N)$. Now, we prove that by continuously decreasing $m_1^x$ starting from $m_1^A(N)$, it is possible to reach at a point where this set of bin edges $\{m_i^x\}_{i=1}^N$ forms an equilibrium, i.e., the nearest neighbor conditions at all boundaries are satisfied. Since $m_1^x=m_1^A(N)$ leads to an equilibrium with $N$ bins and all of the subsequent bin edges move in a monotone manner with respect to $m_1^x$ by Lemma~\ref{lem:monotonicity}, it follows that $m_N^x\to\infty$ as $m_1^x$ approaches $m_1^A(N)$ from the left. This observation ensures that it is possible to find a left neighborhood of $m_1^A(N)$ such that a value of $m_1^x$ in this neighborhood leads to $\bar{u}_{N+1}^x<\tilde{u}_{N+1}^x$ where $\bar{u}_{N+1}^x = \mathbb{E}[M | m_{N}^x \leq M ]$ and $\tilde{u}_{N+1}^x =2m_{N}^x-u_{N}^x-2b$. In order to see this, we write
\begin{align}
\tilde{u}_{N+1}^x-\bar{u}_{N+1}^x 
&= 2m_N^x - u_N^x - 2b -\bar{u}_{N+1}^x \nonumber \\
&= m_N^x - u_N^x -(\bar{u}_{N+1}^x-m_N^x) - 2b.\label{eq:uTildeMinusuBar}
\end{align}
Since $(\mathbb{E}[M|x\leq M]-x)$ is a decreasing function of $x$ by Lemma~\ref{lem:monotonicity}, the term $(\bar{u}_{N+1}^x-m_N^x)$ decreases as $m_1^x$ increases. Moreover, since $u_N^x\rightarrow \mathbb{E}[M|m_{N-1}^A(N)<M]$ as $m_1^x$ approaches $m_1^A(N)$ from the left and $u_N^x$ is monotonically increasing in $m_1^x$, $u_N^x$ is upper bounded by $\mathbb{E}[M|m_{N-1}^A(N)<M]$. Hence, the last three terms in \eqref{eq:uTildeMinusuBar} can be upper bounded by a finite value. Thus, by choosing a value of $m_1^x$ sufficiently close to $m_1^A(N)$, one can make $m_N^x$ sufficiently large, which leads to $\tilde{u}_{N+1}^x-\bar{u}_{N+1}^x>0$. As a result, it holds that $\tilde{u}_{N+1}^x>\bar{u}_{N+1}^x$ for $m_1^x$ at a left neighborhood of $m_1^A(N)$. On the other hand, a sufficiently small value of $m_1^x$ yields $\tilde{u}_{N+1}^x<\bar{u}_{N+1}^x$. To see this, observe that since $\frac{d\tilde{u}_{N+1}^x}{dm_1^x}>1$ by Lemma~\ref{lem:monotonicity}, $\tilde{u}_{N+1}^x$ can be made sufficiently small by taking $m_1^x$ sufficiently small. However, $\bar{u}_{N+1}^x$ is lower bounded by the mean. Thus, a sufficiently small value of $m_1^x$ makes $\tilde{u}_{N+1}^x<\bar{u}_{N+1}^x$. In addition, we know that $(\tilde{u}_{N+1}^x-\bar{u}_{N+1}^x)$ is continuous and monotonically increasing in $m_1^x$. Therefore, there exists a particular value of $m_1^x$ for which $\tilde{u}_{N+1}^x=\bar{u}_{N+1}^x$ holds, and hence, this particular value leads to an equilibrium with $N+1$ bins.

\begin{figure}
\centering
\includegraphics[width=0.95\linewidth]{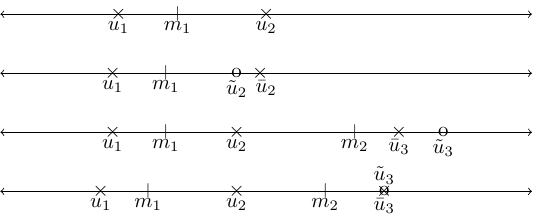}
\caption{Illustration of the proof technique while showing the existence of equilibria with more bins when $b<0$. The first plot shows an equilibrium with two bins. In order to obtain an equilibrium with three bins, we decrease the value of bin edge $m_1$. In the second plot, $u_1$ is the centroid of the bin $(-\infty,m_1)$ and $\bar{u}_2$ is the centroid of the bin $[m_1,\infty)$, and $\tilde{u}_2$ is where the nearest neighbor condition with $u_1$ is satisfied at $m_1$. Since $\tilde{u}_2<\bar{u}_2$, we need to introduce a second bin edge $m_2$ so that $\tilde{u}_2$ is the centroid of the bin $[m_1,m_2)$ and this is shown in the third plot. With this bin edge, we obtain a new centroid $\bar{u}_3$ of the bin $[m_2,\infty)$. In order for this centroid to coincide with $\tilde{u}_3$, we decrease $m_1$ further. This leads to the final plot where an equilibrium with three bins is obtained.}
\label{fig:proof}
\end{figure}

To sum up, it is shown that there always exists an equilibrium with two bins and that if there exists an equilibrium with $N$ bins, it is possible to construct an equilibrium with $N+1$ bins by continuously decreasing the left-most bin edge. Fig.~\ref{fig:proof} illustrates this technique to obtain equilibria with more bins. These results imply that for any $N\geq 2$, there exists an equilibrium with $N$ bins. Moreover, the proof for the existence also reveals that an equilibrium with $N$ bins is unique. To see this, observe that $K_A(x)$ is a piecewise constant and monotonically non-increasing function with jumps by one at discontinuities where these discontinuities lead to equilibrium scenarios, i.e., when $x\in [m_1^A(N+1),m_1^A(N))$, we have that $K_A(x)=N+1$ for $N\in\{2,3,\dots\}$. Hence, this monotonic behavior of $K_A(x)$ ensures the uniqueness of an equilibrium for a given number of bins.

Finally, we briefly describe the generalization of the existence and uniqueness results to the case of $b>0$. We take bin edges in a descending order, i.e., $m_1>\dots>m_{N-1}$, and continuously increase $m_1$ starting from $m_1^D(N)$ to obtain an equilibrium with $N+1$ bins where $m_1^D(N)>\cdots>m_{N-1}^D(N)$ denote bin edges for an equilibrium with $N$ bins. Let
\begin{align}
K_D(x) = \max \Big\{k \;|\; &\text{there exists }x=m_1 > \dots >m_{k-1} \nonumber \\
&\text{ with }\{m_i\}_{i=1}^{k-1}\in \mathcal{M}_D^{k-1} \Big\}.\label{eq:defKPositiveb}
\end{align}
Denote the bin edges that attain the maximum in \eqref{eq:defKPositiveb} for a given $x$ by $m_1^x,\dots,m_{K_D(x)-1}^x$ and denote the centroids with these bin edges by $u_1^x\triangleq\mathbb{E}[M|m_1^x\leq M]$ and $u_i^x\triangleq\mathbb{E}[M|m_{i}^x\leq M<m_{i-1}^x]$ for $i=2,\dots,K_D(x)-1$. Via a similar reasoning, it can be shown that when $m_1^x>m_1^D(N)$ there exists $m_1^x>\cdots>m_N^x$ such that the nearest neighbor conditions at $m_1^x,\dots,m_{N-1}^x$ are satisfied. Moreover, there exists a right neighborhood of $m_1^D(N)$ such that a value of $m_1^x$ in this neighborhood leads to $\tilde{u}_{N+1}^x<\bar{u}_{N+1}^x$ where $\bar{u}_{N+1}^x = \mathbb{E}[M | M<m_{N}^x ]$ and $\tilde{u}_{N+1}^x =2m_{N}^x-u_{N}^x-2b$. On other hand, a sufficiently large value of $m_1^x$ results in $\tilde{u}_{N+1}^x>\bar{u}_{N+1}^x$. Since $(\tilde{u}_{N+1}^x-\bar{u}_{N+1}^x)$ is continuous and monotonically increasing in $m_1^x$, there exists a particular value of $m_1^x$ leading to an equilibrium with $N+1$ bins by making $\tilde{u}_{N+1}^x=\bar{u}_{N+1}^x$. These imply that $K_D(x)$ is a piecewise constant monotonically non-decreasing function with jumps by one at discontinuities which correspond to equilibrium scenarios, i.e., we have that $K_D(x) = N+1$ when $x\in (m_1^D(N),m_1^D(N+1)]$ for $N\in \{2,3,\dots\}$. Therefore, for any $N\geq 2$, there exists a unique equilibrium with $N$ bins.  
\hspace*{\fill}\qed

\subsection{Proof of Theorem~\ref{thm:LogConcaveOneSidedExistence}}\label{ProofLogConcaveOneSidedExistence}

For the cases when there is no finite upper bound on the number of bins, the existence and uniqueness results can be obtained in a similar manner to Theorem~\ref{thm:LogConcaveTwoSidedExistence}. Namely, for a density with a support on $[m_L,\infty)$ and $b>0$, the approach is to increase the value of the right-most bin edge to construct an equilibrium with $N+1$ bins starting from an equilibrium with $N$ bins. On the other hand, for a density with a support on $(-\infty,m_U]$ and $b<0$, the value of the left-most bin edge is decreased to obtain equilibria with large number of bins.

We now prove that there exists a finite upper bound on the number of bins in the case of support on $[m_L,\infty)$ and $b<0$. Note that if there does not exist an informative equilibrium, then the upper bound on the number of bins is equal to one. Suppose that the bias term is such that there exists at least one informative equilibrium. Let $m_1$ be the right-most bin edge and $u_1=\mathbb{E}[M|m_1\leq M]$. For the nearest neighbor condition at $m_1$ to be satisfied, there must be a receiver action $\tilde{u}_2 = 2m_1-u_1-2b$. We know that $(\mathbb{E}[M|m_1\leq M]-m_1)$ monotonically decreases as $m_1$ increases due to Lemma~\ref{lem:monotonicity} and approaches zero as $m_1\to \infty$. Therefore, existence of an informative equilibrium and monotonic behavior of $(\mathbb{E}[M|m_1\leq M]-m_1)$ guarantee that there exists a particular value of $m_1$, say $\bar{m}_1$, which makes $u_1-\bar{m}_1=\mathbb{E}[M|\bar{m}_1\leq M]-\bar{m}_1 = -2b$. Since $(\mathbb{E}[M|m_1\leq M]-m_1)$ monotonically decreases as $m_1$ increases, we have $(m_1-\mathbb{E}[M|m_1\leq M]-2b)>0$ for all $m_1> \bar{m}_1$. This implies that for all $m_1>\bar{m}_1$, we have $(m_1-\mathbb{E}[M|m_1\leq M]-2b)=\tilde{u}_2-m_1>0$. Therefore, there cannot be an equilibrium with a right-most bin edge $m_1$ satisfying $m_1>\bar{m}_1$ since it is required to find $m_2<m_1$ such that the centroid of the bin $[m_2,m_1)$ is given by $\tilde{u}_2>m_1$, which is not possible. Thus, there exists an upper bound on the value of the right-most bin edge in equilibrium, which reveals that the number of bins at an equilibrium is upper bounded. Via a construction similar to Theorem~\ref{thm:LogConcaveTwoSidedExistence}, the existence and uniqueness of an equilibrium with number of bins $N\in\{1,\dots,N_{\text{max}}\}$ can be proven. The result can also be extended to the case of support on $(-\infty,m_U]$ and $b>0$.
\hspace*{\fill}\qed

\subsection{Proof of Theorem~\ref{thm:logConcaveNoninformative}}\label{ProofLogConcaveNoninformative}
\begin{enumerate}[(i)]
\item Suppose first that $2b\leq -(\mu-m_L)$. Recall that an equilibrium with two bins requires that $\tilde{m}(m)=m$ holds for a bin edge $m$ where
\begin{align}
\tilde{m}(m) = \frac{\mathbb{E}[M|m\leq M] + \mathbb{E}[M|M<m]}{2}+b.
\end{align}
Note that for an equilibrium with two bins we need $m>m_L$ as otherwise there is no information conveyed related to the source. If we take $m_L=m$, we get
\begin{align}
\tilde{m}(m_L) ={\mu + m_L\over 2} +b \leq m_L,
\end{align}
where the inequality is by assumption. Moreover, we know that $(m-\tilde{m}(m))$ is strictly monotonically increasing in $m$ due to Lemma~\ref{lem:monotonicity}. Therefore, for any bin edge $m>m_L$, we have $m>\tilde{m}(m)$. Therefore, there does not exist a bin edge that leads to an equilibrium with two bins. On the other hand, when $2b>-(\mu-m_L)$, it is possible to construct an interval for which $\tilde{m}(\cdot)$ maps this interval into itself in a similar manner to Theorem~\ref{thm:LogConcaveTwoSidedExistence} and then apply Brouwer’s fixed point theorem \cite{guide2006infinite}.
\item The proof is similar to (i). 
\end{enumerate}

\subsection{Proof of Theorem~\ref{thm:LogConcaveCost}}\label{ProofLogConcaveCost}

We first focus on the case of $b<0$ and then generalize the result to the case of $b>0$. Theorem~\ref{thm:LogConcaveTwoSidedExistence} obtains an equilibrium with $N+1$ bins starting from an equilibrium with $N$ bins by continuously decreasing the value of the left-most bin edge. In the following analysis, we instead start from an equilibrium with $N+1$ bins and continuously decrease the value of the right-most bin edge to reach an equilibrium with $N$ bins. Then, we show that the objective function for the decoder is monotonically decreasing in the value of the right-most bin edge. Note that the expected costs of the encoder and decoder in equilibrium differ only by a constant as pointed out in Remark~\ref{rem:equilibriumSelection}. Thus, it suffices to prove the monotonic behavior for the expected cost of the decoder.  

Recall that $m_1^D(N)>\cdots>m_{N-1}^D(N)$ denote the bin edges arranged in a descending order for the unique equilibrium with $N$ bins. We first show that $m_{1}^D(N)<m_{1}^D(N+1)$. When $m_1^x=m_1^D(N)$, we have that $\tilde{u}_{N}^x=\bar{u}_N^x$ where $\bar{u}_N^x=\mathbb{E}[M|M<m_{N-1}^x]$ and $\tilde{u}_{N}^x=2m_{N-1}^x-u_{N-1}^x-2b$. Since Lemma~\ref{lem:monotonicity} implies that $\frac{dm_{i+1}^x}{dm_i^x}>1$ and $\frac{du_{i+1}^x}{dm_i^x}>1$ hold for $i=1,\dots,N-2$, it is guaranteed that $m_{N-1}^x<\dots<m_1^x$ holds as $m_1^x$ is decreased. Moreover, by Lemma~\ref{lem:monotonicity}, we have $0< \frac{d\bar{u}_N^x}{d\tilde{m}_1^x} <\frac{d\tilde{u}_N^x}{d\tilde{m}_1^x}$. Thus, $m_1^x<m_1^D(N)$ leads to $\tilde{u}_N^x<\bar{u}_N^x$. These observations imply that $x<m_1^D(N)$ leads to $K_D(x)\leq N$. Since $K_D(x)=N+1$ when $x=m_1^D(N+1)$, it follows that $m_{1}^D(N)<m_{1}^D(N+1)$. Now, we take $m_1^x=m_1^D(N+1)$ and decrease until $m_1^x=m_1^D(N)$. Since all the subsequent bin edges decrease monotonically as $m_1^x$ is decreased, it follows that 
$m_i^x\rightarrow m_i(N)$ for $i=2,\dots,N-1$ and $m_{N}^x\rightarrow -\infty$ as $m_1^x$ approaches $ m_1^D(N)$ from the right. Namely, in order to obtain an equilibrium with $N$ bins, one can decrease $m_1^x$ starting from its value for the unique equilibrium with $N+1$ bins, i.e., $m_1^D(N+1)$, until its value for the unique equilibrium with $N$ bins, i.e., $m_1^D(N)$. 

Next, the aim is to show that when $m_1^x\in (m_1^D(N),m_1^D(N+1))$ the objective function of decoder is monotonically decreasing in $m_1^x$ where the sender employs the bin edges $m_N^x<\cdots<m_1^x$ with $m_1^x=x$ such that the nearest neighbor conditions at the boundaries $m_1^x,\dots,m_{N-1}^x$ hold with the centroids $u_1^x,\dots,u_{N+1}^x$ computed based on these bin edges. With these bin edges, the objective function of the receiver can be written as
\begin{align*}
J^d(x) =& 
\int_{m_1^x}^{\infty} (u_1^x-m)^2f(m)dm  \nonumber \\
&+\sum_{i=1}^{N-1}\int_{m_{i+1}^x}^{m_i^x} (u_{i+1}^x-m)^2f(m)dm \nonumber \\
&+\int_{-\infty}^{m_N^x} (u_{N+1}^x-m)^2f(m)dm,
\end{align*} 
where $u_1^x$, $u_{N+1}^x$ and $u_i^x$ denote the centroid of the bins $[m_1^x,\infty)$, $(-\infty,m_N^x)$ and $[m_{i}^x,m_{i-1}^x)$ for $i=2,\dots,N$, respectively. If we take the derivative with respect to $x$, we get
\begin{align*}
\frac{d J^d(x)}{dx}
=\sum_{i=1}^N \frac{dm_i^x}{dx}f(m_i^x) \Big(
(u_{i+1}^x-m_i^x)^2 - (u_i^x-m_i^x)^2
\Big).
\end{align*}
Since the nearest neighbor conditions at $m_1^x,\dots,m_{N-1}^x$ are satisfied, the following holds:
\begin{align*}
(u_{i+1}^x-m_i^x)^2 - (u_i^x-m_i^x)^2 = (u_i^x-u_{i+1}^x) 2b <0
\end{align*}
for $i=1,\dots,N-1$, where the inequality follows from $b<0$. Moreover, if $m_1^x=m_1^D(N+1)$, then
\begin{align*}
(u_{N+1}^x-m_N^x)^2 - (u_N^x-m_N^x)^2 = (u_N^x-u_{N+1}^x) 2b<0.
\end{align*}
As $m_1^x$ decreases, the term $(u_{N+1}^x-m_N^x)^2 = (m_N^x-\mathbb{E}[M|M<m_N^x])^2$ decreases by Lemma~\ref{lem:monotonicity}. On the other hand, $(u_N^x-m_N^x)^2$ increases as $m_1^x$ decreases because $\frac{dm_N^x}{du_N^x}>1$ by Lemma~\ref{lem:monotonicity}. These ensure that when $m_1^x$ is decreased, we still have
\begin{align}
(u_{N+1}^x-m_N^x)^2 - (u_N^x-m_N^x)^2<0.
\end{align}
As a result, since $f(m_i^x)$ and $\frac{dm_i^x}{dx}$ are all positive and $\big(
(u_{i+1}^x-m_i^x)^2 - (u_i^x-m_i^x)^2
\big)$ are all negative, it follows that $\frac{dJ^d(x)}{dx}<0$. This shows that the expected cost monotonically increases as $x=m_1^x$ decreases. Since we obtain the equilibrium with $N$ bins by decreasing $m_1^x$, the equilibrium with $N+1$ bins induces smaller expected cost than the equilibrium with $N$ bins.   


We now generalize the result to the case of $b>0$ where we take bin edges in an ascending order, i.e., $m_1^x<\cdots<m_{N}^x$. We know that $m_1^A(N+1)<m_1^A(N)$ in this case, which can be proven in a similar manner to the case of $b<0$. Due to the monotonicity of the subsequent bin edges with respect to the first bin edge, we know that  $m_i^x\rightarrow m_i^A(N)$ for $i=2,\dots,N-1$ and $m_N^x\rightarrow \infty$ as $m_1^x$ approaches $m_1^A(N)$ from the left. Therefore, by increasing $m_1^x$ from $m_1^A(N+1)$ to $m_1^A(N)$, one can obtain the unique equilibrium with $N$ bins starting from the unique equilibrium with $N+1$ bins. 

If we write the objective function of the decoder as a function of $m_1^x=x$, we get 
\begin{align*}
J^d(x) =& \int_{-\infty}^{m_1^x} (u_1^x-m)^2f(m)dm \\
&+ \sum_{i=1}^{N-1}\int_{m_i^x}^{m_{i+1}^x} (u_{i+1}^x-m)^2f(m)dm\\
&+ \int_{m_N^x}^{\infty} (u_{N+1}^x-m)^2f(m)dm.
\end{align*} 
Taking the derivative with respect to $x$ yields
\begin{align*}
\frac{dJ^d(x)}{dx}=\sum_{i=1}^N \frac{dm_i^x}{dx}f(m_i^x) \Big(
(u_i^x-m_i^x)^2-(u_{i+1}^x-m_i^x)^2 
\Big). 
\end{align*}
Since the nearest neighbor conditions at $m_1^x,\dots,m_{N-1}^x$ are satisfied, we have that
\begin{align}
(u_i^x-m_i^x)^2-(u_{i+1}^x-m_i^x)^2  = (u_{i+1}^x-u_i^x) 2b >0
\end{align}
for $i=1,\dots,N-1$. Moreover, when $m_1^x=m_1^A(N+1)$, we have
\begin{align*}
(u_N^x-m_N^x)^2-(u_{N+1}^x-m_N^x)^2  = (u_{N+1}^x-u_N^x) 2b>0.
\end{align*}
As $m_1^x$ increases, the term $(u_{N+1}^x-m_N^x)^2 = (\mathbb{E}[M|m_N^x<M]-m_N^x)^2$ decreases by Lemma~\ref{lem:monotonicity}. On the other hand, $(u_N^x-m_N^x)^2$ increases as $m_1^x$ increases since $\frac{dm_N^x}{du_N^x}>1$ by Lemma~\ref{lem:monotonicity}. These imply that when $m_1^x$ is increased, we still have
\begin{align}
(u_{N+1}^x-m_N^x)^2 - (u_N^x-m_N^x)^2>0.
\end{align}
As a result, we get $\frac{dJ^d(x)}{dx}>0$ since $f(m_i^x)$, $\frac{dm_i^x}{dx}$ and $\big(
(u_i^x-m_i^x)^2-(u_{i+1}^x-m_i^x)^2 
\big)$ are all positive. This shows that the expected cost monotonically increases as $m_1^x$ increases. Since we obtain the equilibrium with $N$ bins by increasing $x=m_1^x$, the equilibrium with $N+1$ bins induces smaller expected cost than the equilibrium with $N$ bins.   
\hspace*{\fill}\qed

\subsection{Proof of Theorem~\ref{thm:convergence}}\label{sec:convergenceProof}
Given an initial set of bin edges $\boldsymbol{m} = [m_1,\dots,m_{N-1}]^T$ with $m_1<\cdots<m_{N-1}$, a modified Lloyd-Max iteration defined in \eqref{eq:LloydMaxIter} is expressed as
\begin{align*}
&T_1(\boldsymbol{m})=\frac{\mathbb{E}[M|M<m_1] + \mathbb{E}[M|m_1\leq M< m_2]}{2}+b\\
&T_2(\boldsymbol{m})=\\
&\frac{\mathbb{E}[M|m_1\leq M< m_2] + \mathbb{E}[M|m_2\leq M< m_3]}{2}+b\\
&\hphantom{T_1(\boldsymbol{m})..}\vdots\\
&T_{N-1}(\boldsymbol{m})\\
&=\frac{\mathbb{E}[M|m_{N-2}\leq M<m_{N-1}] + \mathbb{E}[M|m_{N-1}\leq M]}{2}+b.
\end{align*}
Notice that as a result of an iteration it is guaranteed that the resulting bin edges maintain the order, i.e., $T_1(\boldsymbol{m})<\cdots<T_{N-1}(\boldsymbol{m})$. Thus, after each iteration, the corresponding edges form a valid partition. Let $\lVert \boldsymbol{x}\rVert=\max_i |x_i|$. Let $\boldsymbol{m}$ and $\tilde{\boldsymbol{m}}$ be two set of bin edges, which do not necessarily lead to an equilibrium, with $m_1<\cdots<m_{N-1}$ and $\tilde{m}_1<\cdots<\tilde{m}_{N-1}$. Observe that
\begin{align}
\big\lVert T(\boldsymbol{m})-T(\tilde{\boldsymbol{m}}) \big\rVert 
&= \big\lVert \big(T(\boldsymbol{m})-b\big) - \big(T(\tilde{\boldsymbol{m}})-b\big) \big\rVert\nn \\
&< \big\lVert \boldsymbol{m}-\tilde{\boldsymbol{m}} \big\rVert \label{eq:convergenceIneq}
\end{align}
where the inequality follows from \cite[Lemma~6]{kieffer83} since $(T(\boldsymbol{m})-b)$ may be viewed as a fixed point iteration considering the team theoretic setup with no bias. Notice that the inequality in \eqref{eq:convergenceIneq} is strict as the source distribution is assumed to be strictly log-concave. We also know that there exists a unique fixed point of $T(\cdot)$ from Theorem~\ref{thm:LogConcaveTwoSidedExistence}. Let $\boldsymbol{m}^*$ denote the unique fixed point. By utilizing \eqref{eq:convergenceIneq}, we get
\begin{align*}
\big\lVert T(\boldsymbol{m})-\boldsymbol{m}^* \big\rVert  
= \big\lVert T(\boldsymbol{m})-T(\boldsymbol{m}^*) \big\rVert 
< \big\lVert \boldsymbol{m}-\boldsymbol{m}^* \big\rVert,
\end{align*}
which implies that $\lVert T^n(\boldsymbol{m})-\boldsymbol{m}^* \rVert$ is a non-negative and monotonically decreasing sequence. Hence, the limit $\lim_{n\to\infty}\lVert T^n(\boldsymbol{m})-\boldsymbol{m}^* \rVert$ exists. In the following, we show that this limit is zero. To prove this, suppose otherwise, that is,  $\lim_{n\to\infty}\lVert T^n(\boldsymbol{m})-\boldsymbol{m}^* \rVert=c$ for some $c>0$. Since $\{ T^n(\boldsymbol{m})\}$ is bounded, there exists a convergent subsequence, i.e., $T^{n_k}(\boldsymbol{m}) \to \tilde{\boldsymbol{m}}$ as $k\to\infty$ where $\tilde{\boldsymbol{m}}$ denotes the limit point. Now, observe that
\begin{align*}
c 
&= \lim_{n\to\infty} \lVert T^n(\boldsymbol{m})-\boldsymbol{m}^* \rVert\\
&= \lim_{k\to\infty} \lVert T^{n_k}(\boldsymbol{m})-\boldsymbol{m}^* \rVert\\
&= \lVert \lim_{k\to\infty} T^{n_k}(\boldsymbol{m})-\boldsymbol{m}^* \rVert\\
&= \lVert \tilde{\boldsymbol{m}}-\boldsymbol{m}^* \rVert,
\end{align*} 
where the third equality is due to the continuity of the norm. By using \eqref{eq:convergenceIneq}, we get $\lVert T(\tilde{\boldsymbol{m}})-\boldsymbol{m}^* \rVert<\lVert \tilde{\boldsymbol{m}}-\boldsymbol{m}^* \rVert=c$. This contradicts with the assumption that the sequence $\lVert T^n(\boldsymbol{m})-\boldsymbol{m}^*\rVert$ converges to $c>0$ as a monotonically decreasing sequence. Thus, $\lVert T^n(\boldsymbol{m})-\boldsymbol{m}^* \rVert$ converges to zero as $n$ goes to infinity, which implies that the fixed point iterations converge to the unique fixed point.
\hspace*{\fill}\qed

\subsection{Proof of Theorem~\ref{FinThm}} \label{ProofFinThm}
\begin{enumerate}[(i)]
\item Note that $N \leq \lfloor 1 + \frac{K-a}{2|b|}\rfloor$ is obtained for $\mm_{N-1}<K$ since the distance between the optimal decoder actions (reconstruction values) must be at least $2|b|$, and the proof is completed. Thus, it can be assumed that $\mm_{N-1}\geq K$. Then, $u_N = \mathbb{E}[M|M\geq \mm_{N-1}] \leq \mm_{N-1} + \eta$ holds by Assumption~\ref{assum:semiInfinite}-(iii), and it follows that 		
\begin{align*}
\eta \geq u_{N} - \mm_{N-1} 	&= (\mm_{N-1} - u_{N-1}) - 2b\\
&\ge (u_{N-1} - \mm_{N-2}) - 2b\\
&=  (\mm_{N-2} - u_{N-2}) - 2(2b)\\
&\hphantom{\;\;}\vdots\\
&\ge u_{k+1} - \mm_k -(N-k-1)(2b)\\
&\ge -(N-k-1)(2b),
\end{align*}
where $\mm_k$ is the smallest bin edge greater than or equal to $K$, and the inequalities follow from Assumption~\ref{assum:semiInfinite}-(iv). Thus, these inequalities show that the bin-lengths are monotonically increasing for the bins in the interval $[K,\infty)$. Moreover, it follows that the number of bins in the interval $(\mm_k,\infty)$ is upper bounded by $N-k \leq \lfloor \frac{\eta}{2|b|}+1\rfloor$. Regarding the interval $[a,K)$, there can be at most $\lfloor 1 + \frac{K-a}{2|b|}\rfloor$ bins. Further, due to the definition of $\mm_k$, there can be at most one bin in $[K,\mm_k]$, which implies $k\leq 1+\lfloor 1 + \frac{K-a}{2|b|}\rfloor$. Thus, we get $N \leq \lfloor \frac{K-a+\eta}{2|b|}+3\rfloor$.

\item Assume that $K$ is in the $r$-th bin; i.e., $\mm_{r-1}\leq K<\mm_r$. First, consider the bins on the right-side of the $r$-th bin; i.e., the $k$-th bin for $k>r$; in other words, the bin is in the interval $[K,\infty)$. Due to (\ref{centroid}), (\ref{eq:centroidBoundaryEq}) and Assumption~\ref{assum:semiInfinite}-(iii), the following is obtained:
\begin{align}
\mm_{k+1} - u_{k+1} &= u_{k+2}-\mm_{k+1} + 2b \nn\\
&= \mathbb{E}[M| \mm_{k+1} \leq M < \mm_{k+2}]-\mm_{k+1} + 2b \nn\\
&\leq \mathbb{E}[M| \mm_{k+1} \leq M < \infty]-\mm_{k+1} + 2b \nn\\
&\leq \eta+2b, 
\end{align}
which implies $2b\leq\mm_{k+1} - u_{k+1}\leq\eta+2b$.\footnote{Note that non-strict inequalities are preferred over the strict ones in the proof in order to obtain \ekRev{a convex and compact set}, which will be used in the fixed-point theorem to show the existence of an equilibrium.} \ekRev{In addition, from Assumption~\ref{assum:semiInfinite}-(iii), we have
\begin{align}
u_{k+1}&=\mathbb{E}[M| \mm_{k} \leq M < \mm_{k+1}]\nn\\
&\leq \mathbb{E}[M| \mm_{k} \leq M <\infty]\nn\\
&\leq m_k+\eta,
\end{align}	
which implies that $0\leq u_{k+1}-\mm_{k}\leq\eta$. Thus, $2b\leq\mm_{k+1}-\mm_{k}\leq2\eta+2b$ is obtained.}

\ekRev{Now consider the bins in the interval $[a,K)$; i.e., the $k$-th bin for $k\leq r$. If $r=1$, it suffices the consider the length of the left-most bin edge as other bins are in the interval $[K,\infty)$. In this case, we get $\mm_{1} - u_{1}=u_{2}-\mm_{1}+2b\leq\eta+2b$ and $u_{1}=\mathbb{E}[M| a \leq M < \mm_{1}]\leq \mathbb{E}[M| a \leq M < \infty]=\mu$, which imply $2b\leq m_1-m_0 \leq \mu+\eta+2b-a$. If $r\geq2$, then $\mm_{r} - u_{r}=u_{r+1}-\mm_{r}+2b\leq \eta+2b$. If we utilize $u_{r}-\mm_{r-1}=\mm_{r-1} - u_{r-1}-2b\Rightarrow u_r=2\mm_{r-1} - u_{r-1}-2b\leq 2K-a-2b$, we obtain $\mm_{r} \leq u_{r}+\eta+2b \leq 2K-a+\eta$. Thus, we get $\mm_{r}-\mm_{0}\leq 2K+\eta-2a$, which provides an upper bound for the lengths of the left-most bins. Hence, for any bin, the length of the bin, $l_k\triangleq\mm_k-\mm_{k-1}$, satisfies $2b\leq l_k\leq\max\{2\eta+2b,\mu+\eta+2b-a,2K+\eta-2a\}\triangleq\Delta$. Based on the bin-lengths, the bin edges can be represented by $\mm_{i}=a+\sum_{j=1}^{i}l_{j}$. Observe that the set $\mathscr{K} \triangleq [a+2b, a+\Delta] \times [a+2(2b), a+2\Delta] \times [a+3(2b), a+3\Delta] \times \cdots$ (where $\{m_1, m_2, m_3, \cdots\} \in \mathscr{K}$) is a convex and compact set by Tychonoff's theorem \cite{fixedPointBook}. Furthermore, at the equilibrium, the best responses of the encoder and the decoder in \eqref{centroid} and \eqref{eq:centroidBoundaryEq} can be combined to define a mapping as follows:
\begin{align}
\mathbf{\mm}
\triangleq\begin{bmatrix}
\mm_{1} \\
\mm_{2} \\
\vdots \\
\mm_{k}\\
\vdots
\end{bmatrix} 
=
\begin{bmatrix}
{u_1+u_2\over2} \\
{u_2+u_3\over2} \\
\vdots \\
{u_{k}+u_{k+1}\over2} \\
\vdots
\end{bmatrix}+b \triangleq\mathscr{T}(\mathbf{\mm}),
\label{eq:fixedPointInf}
\end{align}
where $u_k=\mathbb{E}[M|m_{k-1}\leq M < m_k]$ for $k=1,2,\dots$. Since the source density is continuous, the mapping $\mathscr{T}(\mathbf{\mm}):\mathscr{K}\rightarrow\mathscr{K}$ is continuous under the point-wise convergence, and hence, under the product topology (the result follows by incrementing the dimension of the product one-by-one and showing the continuity at each step). Further, since an (countably) infinite product of real intervals is a locally convex vector space, $\mathscr{K}$ is a convex and compact space. Hence, there exists a fixed point for the mapping $\mathscr{T}$ such that $\mathbf{\mm^*}=\mathscr{T}(\mathbf{\mm^*})$ by Tychonoff's fixed-point theorem \cite{fixedPointBook}, which implies that there exists an equilibrium with infinitely many bins.}

\end{enumerate}
\hspace*{\fill}\qed

\subsection{Proof of Theorem~\ref{FinThm2}} \label{ProofFinThm2}
For $t\geq K$, $\mathbb{E}[M|M\ge t]-t\leq\eta$ holds due to Assumption~\ref{assum:semiInfinite}-(iii). If $t<K$, observe the following:
\begin{align*}
\mathbb{E}[M|M\ge t]-t \leq \mathbb{E}[M|M\ge K]-a \leq K+\eta-a .
\end{align*}
Thus, for any $t\geq a$, it holds that $\mathbb{E}[M|M\ge t]\leq t+K+\eta-a$. Suppose that there is an equilibrium with at least two bins. For the $k$-th bin, it holds that
\begin{align*}
&m_k = \frac{u_k + u_{k+1}}{2}+b \nn\\
&= \frac{\mathbb{E}[M| \mm_{k-1} \leq M < \mm_{k}] + \mathbb{E}[M| \mm_{k} \leq M < \mm_{k+1}]}{2}+b \\
&< \frac{\mm_{k} + \mathbb{E}[M| \mm_{k} \leq M < \infty]}{2}+b \\
&< \frac{\mm_{k} + \mm_{k}+K+\eta-a}{2}+b \\
&\Rightarrow b>-{K+\eta-a\over2} ,
\end{align*}
which proves the statement.
\hspace*{\fill}\qed

\subsection{Proof of Proposition~\ref{prop:twoSidedMonotonicity}} \label{ProofTwoSidedMonotonicity}
\begin{enumerate} 
\item[(i)] Let $\mm_s$ be the greatest bin edge less than or equal to $S$; i.e., $m_s\leq S<m_{s+1}$. Then, noting $\mm_0=-\infty$ and $u_1=\mathbb{E}[M|M\leq \mm_{1}]$, observe the following:
\begin{align*}
\nu \geq \mm_{1} - u_1 &= u_2 - \mm_1 + 2b \nn\\
&\geq  \mm_2 - u_2 + 2b \nn\\
&= u_3 - \mm_2 + 2(2b) \nn\\
&\hphantom{\;\;}\vdots \nn\\
&\geq \mm_s - u_s + (s-1)(2b) \nn\\
&\geq 2(s-1)b \nn\\
&\Rightarrow s\leq 1 + {\nu\over2b},
\end{align*}
which shows that the number of bins in the interval $(-\infty,S]$ is upper bounded. \ekRev{Furthermore, these inequalities reveal that for $m<S$, the bin-lengths are monotonically decreasing.}
\item[(ii)] Let $\mm_k$ be the smallest bin edge greater than or equal to $K$; i.e., $m_{k-1}<K\leq m_{k}$, and $\mm_r$ be the right-most bin edge; i.e., $\mm_{r+1}=\infty$ and $u_{r+1}=\mathbb{E}[m|\mm_r\leq M<\infty]$. Then, observe the following:
\begin{align*}
\eta \geq u_{r+1}-\mm_r &= \mm_{r}-u_{r} - 2b \nn\\
&\geq  u_r-\mm_{r-1} - 2b \nn\\
&= \mm_{r-1}-u_{r-1} - 2(2b) \nn\\
&\hphantom{\;\;}\vdots \nn\\
&\geq u_{k+1}-\mm_k - (r-k)(2b) \nn\\
&\geq - (r-k)(2b) \nn\\
&\Rightarrow r-k+1\leq 1+{\eta\over2|b|},
\end{align*}
which shows that the number of bins in the interval $[K,\infty)$ is upper bounded. \ekRev{Furthermore, from these inequalities, it follows that for $m>K$, the bin-lengths are monotonically increasing.}
\end{enumerate}
\hspace*{\fill}\qed

\subsection{Proof of Theorem~\ref{thm:twoSidedInfinite}} \label{ProofTwoSidedInfinite}
The proof requires separate analyses for the positive and the negative $b$ values. As defined in Proposition~\ref{prop:twoSidedMonotonicity}, let $\mm_s$ be the greatest bin edge less than or equal to $S$; i.e., $m_s\leq S<m_{s+1}$, $\mm_k$ be the smallest bin edge greater than or equal to $K$; i.e., $m_{k-1}<K\leq m_{k}$, and $\mm_r$ be the right-most bin edge; i.e., $\mm_{r+1}=\infty$ and $u_{r+1}=\mathbb{E}[m|\mm_r\leq M<\infty]$.
\begin{enumerate}[(i)]
\item First, assume a positive bias; i.e., $b>0$. For the left-most bin-edge $\mm_1$, consider the following cases:
\begin{enumerate}
\item \underline{$\mm_1\geq K$} : Since $u_2-\mm_1\leq\eta$ and $\mm_1-u_1=u_2-\mm_1+2b$, we have $\mm_1\leq u_1+\eta+2b$. Further, since
\begin{align}
&u_1=\mathbb{E}[M|-\infty<M<\mm_1]\nn\\
&=\mathsf{Pr}(M\leq K|-\infty<M<\mm_1) \nn\\
&\qquad\qquad\times\mathbb{E}[M|-\infty<M\leq K]\nn\\
&\qquad+\mathsf{Pr}(K< M<\mm_1|-\infty<M<\mm_1)\nn\\
&\qquad\qquad\times\mathbb{E}[M|K< M<\mm_1] \nn\\
&\leq\mathsf{Pr}(M\leq K|-\infty<M<\mm_1)\nn\\
&\qquad\qquad\times\mathbb{E}[M|-\infty<M\leq K]\nn\\
&\qquad+\mathsf{Pr}(K< M<\mm_1|-\infty<M<\mm_1)\nn\\
&\qquad\qquad\times\mathbb{E}[M|K< M<\infty] \nn\\
&\leq \mathsf{Pr}(M\leq K|-\infty<M<\mm_1)K \nn\\
&\qquad+ \mathsf{Pr}(K< M<\mm_1|-\infty<M<\mm_1) (K+\eta) \nn\\
&\leq K+\eta,
\end{align} 
we get $\mm_1\leq K+2\eta+2b$. In this case, for any $t\geq 2$, it holds that $\mm_{t}-u_t=u_{t+1}-\mm_t+2b\leq\eta+2b \Rightarrow 2b\leq\mm_{t}-u_t\leq\eta+2b$ and $0\leq u_t-\mm_{t-1}\leq\eta$. Thus, the bin-lengths satisfy $2b\leq l_{t}=\mm_{t}-\mm_{t-1}\leq2\eta+2b$. This shows that these bounds on the bin-lengths hold for every bin in the interval $[K,\infty)$.
\item \underline{$S<\mm_1<K$} : If $\mm_2<K$, then the upper bound on the length of the second bin is $\mm_2-\mm_1\leq K-S$. Now, consider the case when $\mm_2\geq K$ holds. Due to the equilibrium conditions in \eqref{centroid} and \eqref{eq:centroidBoundaryEq}; i.e., ${u_1+u_{2}\over2}+b=\mm_1$ and ${u_{2}+u_{3}\over2}+b=\mm_{2}$, we have ${u_{2}\over2}+b=\mm_1-{u_{1}\over2}=\mm_{2}-{u_{3}\over2}\Rightarrow 2(\mm_{2}-\mm_1)=u_{3}-u_{1} \Rightarrow \mm_{2}-\mm_1=(u_{3}-\mm_{2})+(\mm_1-u_{1})\leq \eta+K-u_1$. Since 
\begin{align}
&u_{1}=\mathbb{E}[M|-\infty<M<\mm_{1}] \nn\\
&=\mathsf{Pr}(-\infty<M\leq S|-\infty<M<\mm_{1}) \nn\\
&\qquad\qquad\times\mathbb{E}[M|-\infty<M\leq S]\nn\\
&\qquad+\mathsf{Pr}(S<M<\mm_{1}|-\infty<M<\mm_{1})\nn\\
&\qquad\qquad\times\mathbb{E}[M|S<M<\mm_{1}] \nn\\
&\geq\mathsf{Pr}(M\leq S|M<\mm_{1})(S-\nu)\nn\\
&\qquad+\mathsf{Pr}(S<M<\mm_{1}|M<\mm_{1})S \nn\\
&\geq S-\nu 
\end{align} 
holds, we have $\mm_{2}-\mm_1\leq \eta+K-S+\nu=K-S+\eta+\nu$, which is the upper bound of the length of the second bin.	

\item \underline{$\mm_1\leq S$} : Since there can be at most one bin in the interval $[\mm_s,S]$, at most ${K-S\over2b}$ bins in the interval $(S,K)$ and at most one bin in the interval $[K,\mm_k]$, we have $k\leq s + {K-S\over2b}+2 \leq {\nu\over2b} + {K-S\over2b} + 3$ by Proposition~\ref{prop:twoSidedMonotonicity}-(i). Furthermore, for any bin in $(-\infty,S]$, $\nu\geq\mm_t-u_t=u_{t+1}-\mm_t+2b \Rightarrow 2b\leq\mm_{t+1}-\mm_t\leq2\nu-2b$ holds (note that if $\nu<2b$, then Case-(a) applies; i.e., there is not any bin in $(-\infty,S]$). \newline
If there is not any bin edge between $\mm_s$ and $\mm_k$; i.e., $k=s+1$, then, due to the equilibrium conditions \eqref{centroid} and \eqref{eq:centroidBoundaryEq}; i.e., ${u_s+u_{s+1}\over2}+b=\mm_s$ and ${u_k+u_{k+1}\over2}+b=\mm_k$, which imply ${u_{s+1}\over2}+b={u_{k}\over2}+b=\mm_s-{u_{s}\over2}=\mm_k-{u_{k+1}\over2} \Rightarrow 2(\mm_k-\mm_s)=u_{k+1}-u_{s} \Rightarrow \mm_k-\mm_s=(u_{k+1}-\mm_k)+(\mm_s-u_{s}) \leq \eta+\nu$. Since $\mm_s\leq S < K\leq\mm_k$, we obtain $\mm_s\geq K-\eta-\nu$ and $\mm_k\leq S+\eta+\nu$. Note that, under this case, the bounds on the length of the bin containing the interval $(S,K)$ are $K-S\leq\mm_k-\mm_s\leq \eta+\nu$. \newline
If there is at least one bin edge between $\mm_s$ and $\mm_k$, we have $\mm_s\leq S < \mm_{s+1}<K$. Due to the equilibrium conditions \eqref{centroid} and \eqref{eq:centroidBoundaryEq}, ${u_s+u_{s+1}\over2}+b=\mm_s$ and ${u_{s+1}+u_{s+2}\over2}+b=\mm_{s+1} \Rightarrow {u_{s+1}\over2}+b=\mm_s-{u_{s}\over2}=\mm_{s+1}-{u_{s+2}\over2}\Rightarrow 2(\mm_{s+1}-\mm_s)=u_{s+2}-u_{s} \Rightarrow \mm_{s+1}-\mm_s=(u_{s+2}-\mm_{s+1})+(\mm_s-u_{s})$. Since 
\begin{align}
&u_{s+2}=\mathbb{E}[M|\mm_{s+1}\leq M<\mm_{s+2}] \nn\\
&\leq \mathbb{E}[M|\mm_{s+1}\leq M<\infty] \nn\\
&=\mathsf{Pr}(\mm_{s+1}\leq M<K|\mm_{s+1}\leq M<\infty)\nn\\
&\qquad\qquad\times\mathbb{E}[M|\mm_{s+1}\leq M<K]\nn\\
&\qquad+\mathsf{Pr}(K\leq M<\infty|\mm_{s+1}\leq M<\infty)\nn\\
&\qquad\qquad\times\mathbb{E}[M|K\leq M<\infty] \nn\\
&\leq\mathsf{Pr}(\mm_{s+1}\leq M<K|\mm_{s+1}\leq M<\infty)K\nn\\
&\qquad+\mathsf{Pr}(K\leq M<\infty|\mm_{s+1}\leq M<\infty)(K+\eta) \nn\\
&\leq K+\eta 
\end{align} 
holds, we have $\mm_s=2\mm_{s+1}-u_{s+2}-(\mm_s-u_{s})\geq 2S - (K+\eta)-\nu \Rightarrow \mm_s\geq 2S-K-\eta-\nu$. Similar to the analysis of the left-end of the interval $(S,K)$ above, we can analyze the right-end of the interval as follows: For the right-end of the interval, we have $S<\mm_{k-1}<K\leq \mm_{k}$. Due to the equilibrium conditions \eqref{centroid} and \eqref{eq:centroidBoundaryEq}; i.e., ${u_k+u_{k+1}\over2}+b=\mm_k$ and ${u_{k-1}+u_{k}\over2}+b=\mm_{k-1}$, which imply ${u_{k}\over2}+b=\mm_k-{u_{k+1}\over2}=\mm_{k-1}-{u_{k-1}\over2}\Rightarrow 2(\mm_{k}-\mm_{k-1})=u_{k+1}-u_{k-1} \Rightarrow \mm_{k}-\mm_{k-1}=(u_{k+1}-\mm_{k})+(\mm_{k-1}-u_{k-1})$. Since 
\begin{align}
&u_{k-1}=\mathbb{E}[M|\mm_{k-2}\leq M<\mm_{k-1}] \nn\\
&\geq \mathbb{E}[M|-\infty< M<\mm_{k-1}] \nn\\
&=\mathsf{Pr}(-\infty< M<S|-\infty< M<\mm_{k-1})\nn\\
&\qquad\qquad\times\mathbb{E}[M|-\infty< M<S]\nn\\
&\qquad+\mathsf{Pr}(S\leq M<\mm_{k-1}|-\infty< M<\mm_{k-1})\nn\\
&\qquad\qquad\times\mathbb{E}[M|S\leq M<\mm_{k-1}] \nn\\
&\geq\mathsf{Pr}( M<S|M<\mm_{k-1})(S-\nu)\nn\\
&\qquad+\mathsf{Pr}(S\leq M<\mm_{k-1}|-\infty< M<\mm_{k-1})S \nn\\
&\geq S-\nu 
\end{align} 
holds, we have $\mm_k=(u_{k+1}-\mm_{k})+2\mm_{k-1}-u_{k-1}\leq \eta+2K-(S-\nu) \Rightarrow \mm_k\leq 2K-S+\eta+\nu$. Therefore, the total length of the bins that contain the interval $(S,K)$ with the minimal number of bins (i.e., $\mm_k-\mm_s$) is bounded by $K-S\leq\mm_k-\mm_s\leq3K-3S+2\eta+2\nu$. 			
\end{enumerate}
For the interval $(-\infty,S]$, there are at most $\lfloor1+{\nu\over2b}\rfloor$ bins as shown in Proposition~\ref{prop:twoSidedMonotonicity} (and the length of the bins is bounded above by $2\nu-2b$, as shown in Case-(c) above). Therefore, since we have $\mm_1\leq K+2\eta+2b$ and $\mm_s\geq \min\{K-\eta-\nu,\, 2S-K-\eta-\nu\}=2S-K-\eta-\nu$, we obtain $K+2\eta+2b\geq\mm_1\geq2S-K-\eta-\nu-\lfloor1+{\nu\over2b}\rfloor(2\nu-2b)$. Furthermore, for any bin, the bin-length $l_t\triangleq\mm_t-\mm_{t-1}$ is between $\Delta_s\triangleq\min\{2b, K-S\}<l_t<\max\{2\nu-2b,2\eta+2b,\eta+\nu,K-S+\eta+\nu,3K-3S+2\eta+2\nu\}\triangleq\Delta_k$. Based on the left-most bin-edge and the bin-lengths, the bin edges can be represented by $\mm_{i}=\mm_1+\sum_{j=1}^{i}l_{j}$ for $i>1$. Similar to the proof of Theorem~\ref{FinThm}, the set $\{\mm_s, l_{1}, l_{2}, \cdots\} \in [2S-K-\eta-\nu-\lfloor1+{\nu\over2b}\rfloor(2\nu-2b), K+2\eta+2b] \times [\Delta_s, \Delta_k] \times [\Delta_s, \Delta_k] \times \cdots$ is a convex and compact set by Tychonoff's theorem \cite{fixedPointBook}. \ekRev{After defining a mapping similar to that in \eqref{eq:fixedPointInf}, it is concluded from Tychonoff's fixed-point theorem \cite{fixedPointBook} that there exists a fixed point. This proves that there exists an equilibrium with infinitely many bins.}

\item Now, assume a negative bias; i.e., $b<0$. A similar approach can be taken for this case, but for the completeness, we include the full proof below. \newline
For the right-most bin edge $\mm_r$, consider the following cases:
\begin{enumerate}
\item \underline{$\mm_r<S$} : Since $\mm_r-u_{r}\leq\nu$ and $u_{r+1}-\mm_r=\mm_r-u_r-2b$, we have $\mm_r\geq u_{r+1}-\nu+2b$. Further, since
\begin{align}
u_{r+1}&=\mathbb{E}[M|\mm_r\leq M<\infty]\nn\\
&=\mathsf{Pr}(\mm_r\leq M<S|\mm_r\leq M<\infty)\nn\\
&\qquad\qquad\times\mathbb{E}[M|\mm_r\leq M<S]\nn\\
&\qquad+\mathsf{Pr}(S\leq M<\infty|\mm_r\leq M<\infty)\nn\\
&\qquad\qquad\times\mathbb{E}[M|S\leq M<\infty] \nn\\
&\geq\mathsf{Pr}(\mm_r\leq M<S|\mm_r\leq M<\infty)\nn\\
&\qquad\qquad\times\mathbb{E}[M|-\infty\leq M<S]\nn\\
&\qquad+\mathsf{Pr}(S\leq M<\infty|\mm_r\leq M<\infty)\nn\\
&\qquad\qquad\times\mathbb{E}[M|S\leq M<\infty] \nn\\
&\geq \mathsf{Pr}(\mm_r\leq M<S|M\geq\mm_r)(S-\nu) \nn\\
&\qquad+ \mathsf{Pr}(M\geq S|M\geq\mm_r)S \nn\\
&\geq S-\nu ,
\end{align} 
$\mm_r\geq S-2\nu+2b$ is obtained. Under this case, for any bin, it holds that $u_{t+1}-\mm_t=\mm_{t}-u_t-2b\leq\nu-2b \Rightarrow -2b\leq u_{t+1}-\mm_t\leq\nu-2b$ and $0\leq \mm_{t}-u_t\leq\nu$. Thus, the bin-lengths satisfy $-2b\leq l_{t+1}=\mm_{t+1}-\mm_{t}\leq2\nu-2b$. Actually, these bounds on the bin-lengths hold for every bin in the interval $(-\infty,S]$.
\item \underline{$S<\mm_r<K$} : If $\mm_{r-1}>S$, then the upper bound of the length of the bin before the last bin is $\mm_r-\mm_{r-1}\leq K-S$. Now, consider the $\mm_{r-1}\leq S$ case. Due to the equilibrium conditions \eqref{centroid} and \eqref{eq:centroidBoundaryEq}; i.e., ${u_r+u_{r+1}\over2}+b=\mm_r$ and ${u_{r-1}+u_{r}\over2}+b=\mm_{r-1}$, which imply ${u_{r}\over2}+b=\mm_r-{u_{r+1}\over2}=\mm_{r-1}-{u_{r-1}\over2}\Rightarrow 2(\mm_{r}-\mm_{r-1})=u_{r+1}-u_{r-1} \Rightarrow \mm_{r}-\mm_{r-1}=(u_{r+1}-\mm_{r})+(\mm_{r-1}-u_{r-1})$. Since 
\begin{align}
u_{r+1}&=\mathbb{E}[M|\mm_r\leq M<\infty] \nn\\
&=\mathsf{Pr}(\mm_r\leq M<K|\mm_r\leq M<\infty) \nn\\
&\qquad\qquad\times\mathbb{E}[M|\mm_r\leq M<K]\nn\\
&\qquad+\mathsf{Pr}(K\leq M<\infty|\mm_r\leq M<\infty)\nn\\
&\qquad\qquad\times\mathbb{E}[M|K\leq M<\infty] \nn\\
&\leq\mathsf{Pr}(\mm_r\leq M<K|M\geq\mm_{r})K\nn\\
&\qquad+\mathsf{Pr}(M\geq K|M\geq\mm_{r})(K+\eta) \nn\\
&\leq K+\eta 
\end{align} 
holds, we have $\mm_{r}-\mm_{r-1}\leq K+\eta-S+\nu=K-S+\eta+\nu$, which is the upper bound of the length of the bin before the last bin.	
\item \underline{$\mm_r>K$} : For any bin in $[K,\infty)$, $\nu\geq u_{t+1}-\mm_t=\mm_t-u_t-2b \Rightarrow -2b\leq u_{t+1}-\mm_t\leq\nu$ and $0\leq \mm_t-u_t\leq\nu+2b$, which results in $-2b\leq\mm_{t+1}-\mm_t\leq2\eta+2b$ (if $\eta<-2b$, then Case-(a) applies; i.e., there is not any bin in $[K,\infty)$). \newline
If there is not any bin edge between $\mm_s$ and $\mm_k$; i.e., $k=s+1$, the analysis for $b>0$ holds: $\mm_s\geq K-\eta-\nu$, $\mm_k\leq S+\eta+\nu$ and $K-S\leq\mm_k-\mm_s\leq \eta+\nu$. Similarly, if there is at least one bin edge between $\mm_s$ and $\mm_k$, the analysis for $b>0$ holds again: $\mm_s\geq 2S-K-\eta-\nu$, $\mm_k\leq 2K-S+\eta+\nu$ and $K-S\leq\mm_k-\mm_s\leq3K-3S+2\eta+2\nu$. 
\end{enumerate}
For the interval $[K,\infty)$, there are at most $\lfloor1-{\eta\over2b}\rfloor$ bins as shown in Proposition~\ref{prop:twoSidedMonotonicity} (and the length of the bins is bounded above by $2\eta+2b$, as shown in Case-(a) above). Therefore, since we have $\mm_r\geq S-2\nu+2b$ and $\mm_k\geq \max\{S+\eta+\nu,\, 2K-S+\eta+\nu\}=2K-S+\eta+\nu$, we obtain $S-2\nu+2b\leq\mm_r\leq2K-S+\eta+\nu+\lfloor1-{\eta\over2b}\rfloor(2\eta+2b)$. Furthermore, for any bin, the bin-length $l_t\triangleq\mm_t-\mm_{t-1}$ is between $\Delta_s\triangleq\min\{-2b, K-S\}<l_t<\max\{2\nu-2b,2\eta+2b,\eta+\nu,K-S+\eta+\nu,3K-3S+2\eta+2\nu\}\triangleq\Delta_k$. Based on the right-most bin-edge and the bin-lengths, the other bin edges can be represented by $\mm_{r-i}=\mm_r-\sum_{j=1}^{i}l_{r-j}$. Similar to the proof of Theorem~\ref{FinThm}, the set $\{\mm_r, l_{r-1}, l_{r-2}, \cdots\} \in [S-2\nu+2b, 2K-S+\eta+\nu+\lfloor1-{\eta\over2b}\rfloor(2\eta+2b)] \times [\Delta_s, \Delta_k] \times [\Delta_s, \Delta_k] \times \cdots$ is a convex and compact set by Tychonoff's theorem \cite{fixedPointBook}. \ekRev{After defining a mapping similar to that in \eqref{eq:fixedPointInf}, it follows from Tychonoff's fixed-point theorem \cite{fixedPointBook} that there exists an equilibrium with infinitely many bins.}

\end{enumerate}
\hspace*{\fill}\qed

\subsection{Proof of Proposition~\ref{prop:expProp1}}\label{sec:expProp1}
Since $u_N = \mathbb{E}[M|\mm_{N-1}\leq M] = \mm_{N-1} + {1\over\lambda}$, it follows that
\begin{align*}
\frac{1}{\lambda} = u_{N} - \mm_{N-1} 	&= (\mm_{N-1} - u_{N-1}) - 2b\\
&> (u_{N-1} - \mm_{N-2}) - 2b\\
&=  (\mm_{N-2} - u_{N-2}) - 2(2b)\\
&\;\;\vdots\\
&> u_1-\mm_0 - (N-1)(2b)\\
&> -(N-1)(2b) .
\end{align*}
Here, the inequalities follow from the fact that the exponential pdf is monotonically decreasing. Thus, for $b<0$,
\begin{align*}
\frac{1}{\lambda} > -(N-1)2b
\Rightarrow
&(N-1)2b > \frac{1}{-\lambda}\\
\Rightarrow
&(N-1) <\frac{1}{-2b\lambda}\\
\Rightarrow
& N \leq \bigg\lfloor -{1\over2b\lambda} + 1 \bigg \rfloor \;,
\end{align*}
is obtained, which also proves that the number of bins at the equilibrium is upper bounded. 
\hspace*{\fill}\qed

\subsection{Proof of Proposition~\ref{prop:expProp2}}\label{sec:expProp2}
For the last two bins, we have
\begin{align}
\mm_{N-1}&={u_{N-1}+u_N\over2}+b\nn\\&= {\left(\mm_{N-2}+{1\over\lambda}-{l_{N-1}\over\me^{\lambda l_{N-1}}-1}\right)+\left(\mm_{N-1}+{1\over\lambda}\right)\over2}+b\nn\\
&\Rightarrow  {l_{N-1}}{\me^{\lambda l_{N-1}}\over\me^{\lambda l_{N-1}}-1}={2\over\lambda}+2b. \label{eq:bin_N-1}
\end{align}
For the other bins; i.e., the $k$-th bin for $k=1,2,\ldots,N-2$, observe the following:    
\begin{align}
&u_{k+1} - \mm_k = \mm_k - u_k - 2b \nn\\
&\hphantom{u_{k+1} - \mm_k} = (\mm_k-\mm_{k-1}) - (u_k-\mm_{k-1}) - 2b \nn\\
&\Rightarrow  {1\over\lambda}-{l_{k+1}\over\me^{\lambda l_{k+1}}-1} = l_k - {1\over\lambda}+{l_{k}\over\me^{\lambda l_{k}}-1}-2b \nn\\
&\Rightarrow  l_{k}{\me^{\lambda l_{k}}\over\me^{\lambda l_k}-1} = {2\over\lambda}+2b - {l_{k+1}\over\me^{\lambda l_{k+1}}-1} .
\label{eq:bin_k}
\end{align}
From \eqref{eq:bin_N-1} and \eqref{eq:bin_k}, the bin-lengths at the equilibrium can be written in a recursive form as in \eqref{eq:compactRecursionLast} and \eqref{eq:compactRecursion}.

The proof for monotonically increasing bin-lengths is based on induction. Before the induction step, in order to use \eqref{eq:compactRecursion}, we examine the structure of $g(x) = \frac{xe^{\lambda x}}{e^{\lambda x}-1}$ and $h(x) = \frac{x}{e^{\lambda x}-1}$. First note that both functions are continuous and differentiable on $[0,\infty)$. Further, $g$ is a positive, strictly increasing and unbounded function on $\mathbb{R}_{\geq0}$, and $h$ is a positive and strictly decreasing function on $\mathbb{R}_{\geq0}$. Now, notice the following properties:
\begin{align}
\begin{split}
g(l_k)&=h(l_k)+l_k ,\\
g(l_k)&=l_{k}{\me^{\lambda l_{k}}\over\me^{\lambda l_k}-1}>l_k,\\
h(l_k) &= {l_{k}\over\me^{\lambda l_{k}}-1} <{l_k\over\lambda l_k} = {1\over\lambda},
\label{eq:gh_Properties}
\end{split}
\end{align}
where the last inequality follows from $\me^y\geq y+1$ with equality if and only if $y=0$. Now consider the length of the $(N-2)$nd bin. By utilizing the properties in \eqref{eq:gh_Properties} on the recursion in \eqref{eq:compactRecursion}, we get
\begin{align}
g(l_{N-2}) &= {2\over\lambda}+2b-h(l_{N-1}) = g(l_{N-1})-h(l_{N-1}) = l_{N-1} \nn\\ &\Rightarrow l_{N-1} = g(l_{N-2}) =l_{N-2}+h(l_{N-2}) \nn\\
&\Rightarrow  l_{N-2} < l_{N-1} < l_{N-2}+{1\over\lambda}.
\end{align}
Similarly, for the $(N-3)$rd bin, the following relations hold:
\begin{align}
&g(l_{N-3}) = {2\over\lambda}+2b-h(l_{N-2}) 
\nn\\&\qquad= g(l_{N-2})+h(l_{N-1})-h(l_{N-2}) = l_{N-2}+h(l_{N-1}) \nn\\
&g(l_{N-3})=l_{N-2}+h(l_{N-1})<l_{N-2}+h(l_{N-2})=g(l_{N-2}) \nn\\&\qquad\Rightarrow l_{N-3}<l_{N-2}\nn\\
&l_{N-2}<l_{N-2}+h(l_{N-1})=g(l_{N-3})\nn\\&\qquad=l_{N-3}+h(l_{N-3})<l_{N-3}+{1\over\lambda}  \nn\\
&\qquad\Rightarrow  l_{N-3} < l_{N-2} < l_{N-3}+{1\over\lambda} .
\end{align}
Now suppose that $l_{N-1}>l_{N-2}>\ldots>l_{k}$ is obtained. Then, consider the $(k-1)$st bin:
\begin{align}
&g(l_{k-1}) = {2\over\lambda}+2b-h(l_{k}) = g(l_{k})+h(l_{k+1})-h(l_{k}) \nn\\&\qquad = l_{k}+h(l_{k+1}) \nn\\
&g(l_{k-1})=l_{k}+h(l_{k+1})<l_{k}+h(l_{k})=g(l_{k}) \Rightarrow l_{k-1}<l_{k}\nn\\
&l_{k}<l_{k}+h(l_{k+1})=g(l_{k-1})=l_{k-1}+h(l_{k-1})<l_{k-1}+{1\over\lambda}  \nn\\
&\qquad\Rightarrow  l_{k-1} < l_{k} < l_{k-1}+{1\over\lambda} .
\end{align}
Thus, the bin-lengths form a monotonically increasing sequence. 
\hspace*{\fill}\qed

\subsection{Proof of Theorem~\ref{thm:expExistence}}\label{sec:expExistence}
\underline{Proof of (i):} Consider the two bins $[0,\mm_1)$ and $[\mm_1,\infty)$. Then, the centroids of the bins (the decoder actions) are $u_1=\mathbb{E}[M|M<\mm_1]={1\over\lambda}-{\mm_1\over\me^{\lambda\mm_1}-1}$ and $u_2=\mathbb{E}[M|\mm_1\leq M]={1\over\lambda}+\mm_1$. In view of \eqref{eq:centroidBoundaryEq}, an equilibrium with these two bins exists if and only if
\begin{align}
\mm_1 &= {u_1+u_2\over2}+b = {{1\over\lambda}-{\mm_1\over\me^{\lambda\mm_1}-1}+{1\over\lambda}+\mm_1\over2}+b \nn\\ &\Rightarrow {\mm_1\over2}{\me^{\lambda\mm_1}\over\me^{\lambda\mm_1}-1}={1\over\lambda}+b \nn\\ &\Rightarrow \me^{\lambda\mm_1}\left({1\over\lambda}+b-{\mm_1\over2}\right) = {1\over\lambda}+b .
\label{eq:twoBinsEq}
\end{align}
Note that in \eqref{eq:twoBinsEq}, $\mm_1=0$ is always a solution; however, in order to have an equilibrium with two bins, we need a non-zero solution to \eqref{eq:twoBinsEq}; i.e., $\mm_1>0$. For this purpose, the Lambert $W$-function will be used. Although the Lambert $W$-function is defined for complex variables, we restrict our attention to the real-valued $W$-function; i.e., the $W$-function is defined as
\begin{align*}
W(x\me^x) &= x \text{ for } x\geq 0 ,\nn\\
W_0(x\me^x) &= x \text{ for } -1\leq x<0 ,\nn\\
W_{-1}(x\me^x) &= x \text{ for } x\leq -1 .
\end{align*}
As it can be seen, for $x\geq0$, $W(x\me^x)$ is a well-defined single-valued function. However, for $x<0$, $W(x\me^x)$ is doubly valued, e.g., when $W(x\me^x)\in(-{1\over\me},0)$, there exist $x_1$ and $x_2$ that satisfy $x_1\me^{x_1}=x_2\me^{x_2}$ where $x_1\in(-1,0)$ and $x_2\in(-\infty,-1)$. In order to differentiate these values, the principal branch of the Lambert $W$-function is defined to represent the values greater than $-1$; e.g., $x_1=W_0(x_1\me^{x_1})=W_0(x_2\me^{x_2})$. Similarly, the lower branch of the Lambert $W$-function represents the values smaller than $-1$; e.g, $x_2=W_{-1}(x_1\me^{x_1})=W_{-1}(x_2\me^{x_2})$. Further, for $x=-1$, two branches of the $W$-function coincide; i.e., $-1=W_0(-{1\over\me})=W_{-1}(-{1\over\me})$. Regarding the definition above, by letting $t\triangleq2\lambda\left({\mm_1\over2}-{1\over\lambda}-b\right)$, the solution of \eqref{eq:twoBinsEq} can be found as follows:
\begin{align}
\me^{t+2+2\lambda b}\left(-t\over2\lambda\right) &= {1\over\lambda}+b \Rightarrow t\me^t = -(2+2\lambda b)\me^{-(2+2\lambda b)} \nn\\ &\Rightarrow t = W\left(-(2+2\lambda b)\me^{-(2+2\lambda b)}\right) .
\label{eq:twoBinsLambert}
\end{align}
Note that, in \eqref{eq:twoBinsLambert}, depending on the values of $-(2+2\lambda b)$, the following cases can be distinguished: 
\begin{enumerate}
\item [(i)] \underline{$-(2+2\lambda b) \geq 0$} : $t\me^t = -(2+2\lambda b)\me^{-(2+2\lambda b)} \Rightarrow t = -(2+2\lambda b) \Rightarrow \mm_1=0$, which implies a non-informative equilibrium; i.e., an equilibrium with only one bin.
\item [(ii)] \underline{$-1 < -(2+2\lambda b) < 0$} : Since $t\me^t = -(2+2\lambda b)\me^{-(2+2\lambda b)}$, there are two possible solutions: 
\begin{enumerate}
\item If $t=W_0\left(-(2+2\lambda b)\me^{-(2+2\lambda b)}\right)=-(2+2\lambda b)$, we have $\mm_1=0$, as in the previous case.
\item If $t=W_{-1}\left(-(2+2\lambda b)\me^{-(2+2\lambda b)}\right) \Rightarrow t<-1 \Rightarrow -1>t=2\lambda\left({\mm_1\over2}-{1\over\lambda}-b\right)=\lambda\mm_1-2-2\lambda b > \lambda\mm_1-1 \Rightarrow \lambda\mm_1<0$, which is not possible.
\end{enumerate} 
\item [(iii)] \underline{$-(2+2\lambda b) = -1$} : Since $t\me^t = -(2+2\lambda b)\me^{-(2+2\lambda b)}$, there is only one solution, $t=-(2+2\lambda b)=-1\Rightarrow\mm_1=0$, which implies that the equilibrium is non-informative.
\item [(iv)] \underline{$-(2+2\lambda b) < -1$} : Since $t\me^t = -(2+2\lambda b)\me^{-(2+2\lambda b)}$, there are two possible solutions:
\begin{enumerate}
\item If $t=W_{-1}\left(-(2+2\lambda b)\me^{-(2+2\lambda b)}\right)=-(2+2\lambda b)$, we have $\mm_1=0$; i.e., an equilibrium with only one bin.
\item If $t=W_0\left(-(2+2\lambda b)\me^{-(2+2\lambda b)}\right)$, we have $-1<t<0 \Rightarrow -1<\lambda\mm_1-2-2\lambda b<0 \Rightarrow {1\over\lambda}+2b<\mm_1<{2\over\lambda}+2b$. Thus, if we have ${1\over\lambda}+2b>0 \Rightarrow b>-{1\over2\lambda}$, then $\mm_1$ must be positive, which implies the existence of an equilibrium with two bins.
\end{enumerate} 
\end{enumerate}
Thus, as long as $b\leq-{1\over2\lambda}$, there exists only one bin at the equilibrium; i.e., there exist only non-informative equilibria; and the equilibrium with two bins can be achieved only if $b>-{1\over2\lambda}$. In this case, $\mm_1={1\over\lambda}W_0\left(-(2+2\lambda b)\me^{-(2+2\lambda b)}\right)+2\left({1\over\lambda}+b\right)$. Note that, since $-1<W_0(\cdot)<0$, the boundary between two bins lies within the interval ${1\over\lambda}+2b<\mm_1<{2\over\lambda}+2b$.

\underline{Proof of (ii):} In order to have an equilibrium with at least three bins, $l_{N-2}>0$ must be satisfied. If we let $c_{k}\triangleq {2\over\lambda}+2b - {l_{k+1}\over\me^{\lambda l_{k+1}}-1}$, the solution to \eqref{eq:bin_k} is
\begin{align}
l_{k} = {1\over\lambda}W_0\left(-\lambda c_{k}\me^{-\lambda c_{k}}\right)+c_{k}. \label{eq:bin_k_lambert}
\end{align}
From \eqref{eq:bin_k_lambert}, if $-\lambda c_{N-2}<-1$ is satisfied, then the solution to $l_{N-2}$ will be positive. Then,
\begin{align}
&-\lambda c_{N-2} = -\lambda\left({2\over\lambda}+2b-h(l_{N-1})\right)\nn\\
&=-\lambda\left(g(l_{N-1})-h(l_{N-1})\right)=-\lambda l_{N-1} <-1 \nn\\
&\Rightarrow l_{N-1}={1\over\lambda}W_0\left(-(2+2\lambda b)\me^{-(2+2\lambda b)}\right)\nn\\
&\qquad\qquad\qquad\qquad\qquad+2\left({1\over\lambda}+b\right)>{1\over\lambda} \nn\\
&\Rightarrow W_0\left(-(2+2\lambda b)\me^{-(2+2\lambda b)}\right) > -1-2\lambda b .
\label{eq:threeBinsanalysis}
\end{align}
Let $t\triangleq W_0\left(-(2+2\lambda b)\me^{-(2+2\lambda b)}\right)$, then $t\me^t=-(2+2\lambda b)\me^{-(2+2\lambda b)}$ and $-1<t<0$. Then, from \eqref{eq:threeBinsanalysis}, since $t\me^t$ is increasing function of $t$ for $t>-1$,
\begin{align}
t>&-1-2\lambda b \Rightarrow t\me^t=-(2+2\lambda b)\me^{-(2+2\lambda b)} \nn\\
&> -(1+2\lambda b)\me^{-(1+2\lambda b)} \nn\\
\Rightarrow & 2+2\lambda b < (1+2\lambda b)\me \Rightarrow b > -{1\over2\lambda}{\me-2\over\me-1} .
\end{align}
Thus, if $b>-{1\over2\lambda}{\me-2\over\me-1}$, an equilibrium with at least three bins is obtained; otherwise; i.e., $b\leq-{1\over2\lambda}{\me-2\over\me-1}$, there can exist at most two bins at the equilibrium.

\underline{Proof of (iii):} Now, we focus on equilibria with infinitely many bins. For any equilibrium, consider a bin with a finite length, say the $k$th bin, and by utilizing \eqref{eq:compactRecursion} and \eqref{eq:gh_Properties}, we have the following inequalities:
\begin{align*}
{2\over\lambda}+2b =& g(l_k)+h(l_{k+1}) = g(l_k) + g(l_{k+1}) - l_{k+1} \nn\\
>& l_k +l_{k+1} - l_{k+1} = l_k \Rightarrow l_k < {2\over\lambda}+2b ,\nn\\
{2\over\lambda}+2b =& g(l_k)+h(l_{k+1}) = h(l_k)+l_k +h(l_{k+1}) \nn\\
<& {1\over\lambda} + l_k + {1\over\lambda} = {2\over\lambda} + l_k 
\Rightarrow l_k > 2b .
\end{align*}
Thus, all bin-lengths are bounded from above and below: $2b < l_k < {2\over\lambda}+2b$. Now consider the fixed-point solution of the recursion in \eqref{eq:compactRecursion}; i.e., $g(l^*)={2\over\lambda}+2b-h(l^*)$. Then, by letting $c\triangleq{2\over\lambda}+2b$,
\begin{align}
l^*{\me^{\lambda l^*}\over\me^{\lambda l^*}-1} &= c - {l^*\over\me^{\lambda l^*}-1} \Rightarrow l^*{\me^{\lambda l^*}+1\over\me^{\lambda l^*}-1} = c \nn\\
&\Rightarrow (c-l^*)\me^{\lambda l^*}-(c+l^*)=0 .
\label{eq:expInfFixedEq}
\end{align}
In order to investigate if \eqref{eq:expInfFixedEq} has a unique solution $l^*$ such that $2b<l^*<{2\over\lambda}+2b$, let $\Psi(s)\triangleq (c-s)\me^{\lambda s}-(c+s)$ for $s\in\left(2b,{2\over\lambda}+2b\right)$ and notice that $\Psi(s)$ is a concave function of $s$ for $s\in\left(2b,{2\over\lambda}+2b\right)$, $\Psi(2b)>0$, $\Psi(s)$ reaches its maximum value in the interval $\left(2b,{2\over\lambda}+2b\right)$; i.e., when $\Psi^\prime(s^*)=0$, and $\Psi({2\over\lambda}+2b)<0$; thus, $\Psi(s)$ crosses the $s$-axis only once, which implies that  $\Psi(s)=0$ has a unique solution in the interval $\left(2b,{2\over\lambda}+2b\right)$. In other words, the fixed-point solution of the recursion in \eqref{eq:compactRecursion} is unique; i.e., $\Psi(l^*)=0$.

Hence, if the length of the first bin is $l^*$; i.e., $l_1=l^*$, then, all bins must have a length of $l^*$; i.e., $l_1=l_2=l_3=\ldots = l^*$. Thus, there exists an equilibrium with infinitely many equi-length bins.

Now, suppose that $l_1<l^*$. Then, by \eqref{eq:compactRecursion}, $h(l_2) = {2\over\lambda}+2b - g(l_1)$. Since $g$ in an increasing function, $l_1<l^*\Rightarrow g(l_1)<g(l^*)$. Let $g(l^*)-g(l_1)\triangleq\Delta>0$, then,
\begin{align}
g(l^*) + h(l^*) &= g(l_1) + h(l_2) = {2\over\lambda}+2b \nn\\&\Rightarrow \Delta = g(l^*)-g(l_1) = h(l_2) - h(l^*) .
\label{eq:l_1SmallCase}
\end{align}
From Proposition~\ref{prop:expProp2}, we know that $h(s)={s\over\me^{\lambda s}-1}$ is a decreasing function with $h^\prime(s)=-\frac{\me^{\lambda s}(\lambda s -1)+1}{(\me^{\lambda s}-1)^2}<0$ for $s>0$ and  $h^\prime(0)=-{1\over2}$. Since $h^\prime(s)$ is an increasing function of $s$, $h^\prime(0)=-{1\over2}$, and $h^\prime(s)>-{1\over2}$ for $s>0$, it follows that ${h(l^*)-h(l_2)\over l^*-l_2}>-{1\over2} \Rightarrow {-\Delta\over l^*-l_2}>-{1\over2} \Rightarrow l^*-l_2>2\Delta$. From \eqref{eq:l_1SmallCase},
\begin{align}
\Delta + \Delta &= (g(l^*)-g(l_1)) + (h(l_2) - h(l^*)) \nn\\
&= g(l^*)-g(l_1) + (g(l_2)-l_2) - (g(l^*)-l^*) \nn\\
&= g(l_2) - g(l_1) + \underbrace{l^* - l_2}_{>2\Delta} \nn\\
&\Rightarrow g(l_2) - g(l_1) <0 \Rightarrow l_2<l_1 .
\end{align}
Proceeding similarly, $l^*>l_1>l_2>\ldots$ can be obtained. Now, notice that, since $h(l_k)$ is a monotone function and $2b < l_k < {2\over\lambda}+2b$, the recursion in \eqref{eq:compactRecursion} can be satisfied if 
\begin{align}
&g(l_k) = {2\over\lambda}+2b-h(l_{k+1}) \nn\\
&\Rightarrow {2\over\lambda}+2b-h(2b) < g(l_k) < {2\over\lambda}+2b-h\left({2\over\lambda}+2b\right) .
\end{align}
Let $\underline{l}$ and $\overline{l}$ and defined as $g(\underline{l})={2\over\lambda}+2b-h(2b)$ and $g(\overline{l})={2\over\lambda}+2b-h\left({2\over\lambda}+2b\right)$, respectively. Thus, if $l_k \notin (\underline{l},\overline{l})$, then there is no solution to $l_{k+1}$ for the recursion in \eqref{eq:compactRecursion}. Since the sequence of the bin-lengths is monotonically decreasing, there is a natural number $K$ such that $l_K>\underline{l}$ and $l_{K+1}\leq\underline{l}$, which implies that there is no solution to $l_{K+2}$. Thus, there cannot be any equilibrium with infinitely many bins if $l_1<l^*$.

A similar approach can be taken for $l_1>l^*$: Since $g$ is an increasing function, $l_1>l^*\Rightarrow g(l_1)>g(l^*)$. Let $g(l_1)-g(l^*)\triangleq\widetilde{\Delta}>0 \Rightarrow g(l_1)-g(l^*) = h(l^*) - h(l_2) = \widetilde{\Delta}$. Then, since $h^\prime(s)>-{1\over2}$ for $s>0$, ${h(l_2)-h(l^*)\over l_2-l^*}>-{1\over2} \Rightarrow {-\Delta\over l_2-l^*}>-{1\over2} \Rightarrow l_2-l^*>2\Delta$. From \eqref{eq:l_1SmallCase},
\begin{align}
\widetilde{\Delta} + \widetilde{\Delta} &= (g(l_1)-g(l^*)) + (h(l^*) - h(l_2)) \nn\\
&= g(l_1)-g(l^*) + (g(l^*)-l^*) - (g(l_2)-l_2) \nn\\
&= g(l_1) - g(l_2) + \underbrace{l_2 - l^*}_{>2\Delta} \nn\\
&\Rightarrow g(l_1) - g(l_2) <0 \Rightarrow l_1<l_2 .
\end{align}
Proceeding similarly, $l^*<l_1<l_2<\ldots$ can be obtained. Since the sequence of the bin-lengths is monotonically increasing, there is a natural number $\widetilde{K}$ such that $l_{\widetilde{K}}<\overline{l}$ and $l_{\widetilde{K}+1}\geq\overline{l}$, which implies that there is no solution to $l_{\widetilde{K}+2}$. Thus, there cannot be any equilibrium with infinite number of bins if $l_1>l^*$. Notice that, it is possible to have an equilibrium with finite number of bins since for the last bin with a finite length, \eqref{eq:compactRecursionLast} is used. Further, it is shown that, at the equilibrium, any finite bin-length must be greater than or equal to $l^*$; i.e., $2b<l^*\leq l_k<{2\over\lambda}+2b$ must be satisfied.
\hspace*{\fill}\qed

\subsection{Proof of Theorem~\ref{thm:expMoreInformative}}\label{sec:expMoreInformative}

\underline{Proof of (i):} Suppose that there exists an equilibrium with $N$ bins, and the corresponding bin-lengths are $l_1<l_2<\ldots<l_N=\infty$ with the bin edges $0=\mm_0<\mm_1<\ldots<\mm_{N-1}<\mm_N=\infty$. Then, the decoder cost is
\begin{align}
J^{d,N} 
&= \sum_{i=1}^{N} \mathrm{Var}\left(M|\mm_{i-1}<M<\mm_i\right)\nn\\
&\qquad\qquad\qquad\qquad\times \mathrm{Pr}(\mm_{i-1}<M<\mm_i) \nn\\
&= \sum_{i=1}^{N} \left({1\over\lambda^2} - {l_i^2\over\me^{\lambda l_i}+\me^{-\lambda l_i}-2}\right) \nn\\
&\qquad\qquad\qquad\qquad\times \left(\me^{-\lambda \mm_{i-1}}\left(1-\me^{-\lambda l_i}\right)\right).
\end{align}
Now, consider an equilibrium with $N+1$ bins with the bin-lengths $\widetilde{l}_1<\widetilde{l}_2<\ldots<\widetilde{l}_{N+1}=\infty$ and the bin edges $0=\widetilde{m}_0<\widetilde{m}_1<\ldots<\widetilde{m}_{N}<\widetilde{m}_{N+1}=\infty$. The relation between the bin-lengths and the bin edges can be expressed as $l_k=\widetilde{l}_{k+1}$ and $\mm_k=\widetilde{m}_{k+1}-\widetilde{l}_1$, respectively, for $k=1,2,\ldots,N$ by Proposition~\ref{prop:expProp2}. Then, the decoder cost at the equilibrium with $N+1$ bins can be written as
\begin{align}
&J^{d,N+1} =  \sum_{i=1}^{N+1} \left({1\over\lambda^2} - {\widetilde{l}_i^2\over\me^{\lambda \widetilde{l}_i}+\me^{-\lambda \widetilde{l}_i}-2}\right) \nn\\
&\qquad\qquad\qquad\qquad\times \left(\me^{-\lambda \widetilde{m}_{i-1}}\left(1-\me^{-\lambda \widetilde{l}_i}\right)\right) \nn\\
&= \left({1\over\lambda^2} - {\widetilde{l}_1^2\over\me^{\lambda \widetilde{l}_1}+\me^{-\lambda \widetilde{l}_1}-2}\right) \left(\me^{-\lambda \widetilde{m}_0}\left(1-\me^{-\lambda \widetilde{l}_1}\right)\right) \nn\\
&+ \sum_{i=2}^{N+1} \left({1\over\lambda^2} - {\widetilde{l}_i^2\over\me^{\lambda \widetilde{l}_i}+\me^{-\lambda \widetilde{l}_i}-2}\right) \left(\me^{-\lambda \widetilde{m}_{i-1}}\left(1-\me^{-\lambda \widetilde{l}_i}\right)\right) \nn\\
&= \left({1\over\lambda^2} - {\widetilde{l}_1^2\over\me^{\lambda \widetilde{l}_1}+\me^{-\lambda \widetilde{l}_1}-2}\right) \left(1-\me^{-\lambda \widetilde{l}_1}\right) \nn\\
&+ \sum_{i=2}^{N+1} \left({1\over\lambda^2} - {l_{i-1}^2\over\me^{\lambda l_{i-1}}+\me^{-\lambda l_{i-1}}-2}\right) \nn\\
&\qquad\qquad\qquad\qquad\times \left(\me^{-\lambda (\mm_{i-2}+\widetilde{l}_1)}\left(1-\me^{-\lambda l_{i-1}}\right)\right) \nn\\
&= \left({1\over\lambda^2} - {\widetilde{l}_1^2\over\me^{\lambda \widetilde{l}_1}+\me^{-\lambda \widetilde{l}_1}-2}\right) \left(1-\me^{-\lambda \widetilde{l}_1}\right) + \me^{-\lambda\widetilde{l}_1} \nn\\
&\times\underbrace{\left(\sum_{i=1}^{N} \left({1\over\lambda^2} - {l_{i}^2\over\me^{\lambda l_{i}}+\me^{-\lambda l_{i}}-2}\right) \left(\me^{-\lambda \mm_{i-1}}\left(1-\me^{-\lambda l_{i}}\right)\right)\right)}_{J^{d,N}} \nn\\
&\overset{(a)}{<} J^{d,N} \left(1-\me^{-\lambda \widetilde{l}_1}\right) + J^{d,N} \me^{-\lambda \widetilde{l}_1} = J^{d,N} .
\label{eq:expInformativeIneq}
\end{align}
Thus, $J^{d,N+1}<J^{d,N}$ is obtained, which implies that the equilibrium with more bins is more informative. Here, (a) follows from the following:
\begin{align}
J^{d,N} &= \sum_{i=1}^{N} \left({1\over\lambda^2} - {l_i^2\over\me^{\lambda l_i}+\me^{-\lambda l_i}-2}\right) \nn\\
&\qquad\qquad\qquad\qquad\times\left(\me^{-\lambda \mm_{i-1}}\left(1-\me^{-\lambda l_i}\right)\right) \nn\\
&> \sum_{i=1}^{N} \left({1\over\lambda^2} - {\widetilde{l}_1^2\over\me^{\lambda \widetilde{l}_1}+\me^{-\lambda \widetilde{l}_1}-2}\right) \nn\\
&\qquad\qquad\qquad\qquad\times\mathrm{Pr}(\mm_{i-1}<m<\mm_i) \nn\\ 
&= \left({1\over\lambda^2} - {\widetilde{l}_1^2\over\me^{\lambda \widetilde{l}_1}+\me^{-\lambda \widetilde{l}_1}-2}\right) ,
\end{align}
where the inequality holds since $\widetilde{l}_1<l_1<l_2<\ldots<l_N$ and $\varphi(s)\triangleq {s^2\over\me^{\lambda s}+\me^{-\lambda s}-2}$ is a decreasing function of $s$. 

\underline{Proof of (ii):} Now consider an equilibrium with infinitely many bins. By Theorem~\ref{thm:expExistence}, the bin-lengths are $l_1=l_2=\ldots=l^*$, where $l^*$ is the fixed-point solution of the recursion in \eqref{eq:compactRecursion}; i.e., $g(l^*)={2\over\lambda}+2b-h(l^*)$, and the bin edges are $\mm_k = k l^*$. Then, the decoder cost is
\begin{align}
&J^{d,\infty} 
= \left({1\over\lambda^2} - {(l^*)^2\over\me^{\lambda l^*}+\me^{-\lambda l^*}-2}\right) \left(1-\me^{-\lambda l^*}\right) \nn\\
&+\sum_{i=2}^{\infty} \left({1\over\lambda^2} - {(l^*)^2\over\me^{\lambda l^*}+\me^{-\lambda l^*}-2}\right) \nn\\
&\qquad\times\left(\me^{-\lambda (i-1) l^*}\left(1-\me^{-\lambda l^*}\right)\right) \nn\\
=& \left({1\over\lambda^2} - {(l^*)^2\over\me^{\lambda l^*}+\me^{-\lambda l^*}-2}\right) \left(1-\me^{-\lambda l^*}\right) +\me^{-\lambda l^*}\nn\\
\times&\underbrace{\sum_{i=1}^{\infty} \left({1\over\lambda^2} - {(l^*)^2\over\me^{\lambda l^*}+\me^{-\lambda l^*}-2}\right) \left(\me^{-\lambda (i-1) l^*}\left(1-\me^{-\lambda l^*}\right)\right)}_{J^{d,\infty}} \nn\\
\Rightarrow & J^{d,\infty} =  \left({1\over\lambda^2} - {(l^*)^2\over\me^{\lambda l^*}+\me^{-\lambda l^*}-2}\right) .
\end{align}
Since the bin-lengths at the equilibria with finitely many bins are greater than $l^*$ by Theorem~\ref{thm:expExistence}, and due to a similar reasoning in \eqref{eq:expInformativeIneq} (indeed, by replacing $\widetilde{l}_1$ with $l^*$), $J^{d,\infty}<J^{d,N}$ can be obtained for any finite $N$. Actually, $J^{d,N}$ is a monotonically decreasing sequence with limit $\lim_{N\to\infty} J^{d,N} = J^{d,\infty}$. Thus, the lowest equilibrium cost is achieved with infinitely many bins.
\hspace*{\fill}\qed

\subsection{Proof of Proposition~\ref{prop:gaussMonotonicity}} \label{ProofGaussMonotonicity}
\begin{itemize}
\item[(i)] 	Consider an equilibrium with $N$ bins for a Gaussian source $M\sim\mathcal{N}(\mu,\sigma^2)$: the $k$-th bin is $[\mm_{k-1},\mm_k)$, and the centroid of the $k$-th bin (i..e, the corresponding action of the decoder) is $u_k=\mathbb{E}[M|\mm_{k-1}\leq M<\mm_k]$ so that $-\infty=\mm_0<u_1<\mm_1<u_2<\mm_2<\ldots<\mm_{N-2}<u_{N-1}<\mm_{N-1}<u_N<\mm_N=\infty$. Further, assume that $\mu$ is in the $t$-th bin; i.e., $\mm_{t-1}\leq \mu<\mm_t$. Due to the nearest neighbor condition (the best response of the encoder) we have $u_{k+1}-\mm_k = (\mm_k - u_k) - 2b$; and due to the centroid condition (the best response of the decoder), we have $u_k=\mathbb{E}[M|\mm_{k-1}\leq M<\mm_k]=\mu - \sigma{\phi({\mm_k-\mu\over\sigma})-\phi({\mm_{k-1}-\mu\over\sigma})\over\Phi({\mm_k-\mu\over\sigma})-\Phi({\mm_{k-1}-\mu\over\sigma})}$. Then, for any bin in $[\mu,\infty)$, since $\mm_{k}>\mu$, the following holds:
\begin{align}
& u_{k} - \mm_{k-1} = \mathbb{E}[M|\mm_{k-1}\leq M<\mm_{k}]- \mm_{k-1} \nn\\
&<\mathbb{E}[M|\mm_{k-1}\leq M<\infty]- \mm_{k-1}\nn\\
&=\mu + \sigma{\phi({\mm_{k-1}-\mu\over\sigma})\over1-\Phi({\mm_{k-1}-\mu\over\sigma})}	- \mm_{k-1} \nn\\
& \stackrel{\mathclap{\text{(a)}}}{<} \mu + \sigma {\sqrt{\left({\mm_{k-1}-\mu\over\sigma}\right)^2+4}+{\mm_{k-1}-\mu\over\sigma}\over2}- \mm_{k-1} \nn\\
&< {\sigma\over2}\Bigg(\sqrt{\left({\mm_{k-1}-\mu\over\sigma}\right)^2+4\left({\mm_{k-1}-\mu\over\sigma}\right)+4}\nn\\
&\qquad\qquad\qquad-{\mm_{k-1}-\mu\over\sigma}\Bigg) = \sigma .
\label{eq:gaussRightIneq}
\end{align}
Here, (a) in due to an inequality on the upper bound of the Mill's ratio \cite{birnbaumMillsRatio1942}; i.e., ${\phi(c)\over1-\Phi(c)}<{\sqrt{c^2+4}+c\over2}$ for $c>0$. Now, observe the following:
\begin{align*}
\sigma > u_{N} - \mm_{N-1} 	&= (\mm_{N-1} - u_{N-1}) - 2b\\
&> (u_{N-1} - \mm_{N-2}) - 2b\\
&\vdots\\
&> - (N-t)(2b),
\end{align*}
where the inequalities follow from the fact that the Gaussian pdf of $M$ is monotonically decreasing on $[\mu,\infty)$, which implies that bin-lengths
are monotonically increasing in the interval $[\mu,\infty)$ when $b<0$. Further, for $b<0$, $N-t<-{\sigma\over2b}$ is obtained, which implies that the number of bins in $[\mu,\infty)$ is bounded by $\big\lfloor -{\sigma\over2b} \big \rfloor$. 
\item[(ii)] Similarly, for any bin in $(-\infty,\mu]$, since $\mm_{k}<\mu$, the following holds:
\begin{align}
&\mm_k - u_{k} = \mm_k - \mathbb{E}[M|\mm_{k-1}<M<\mm_k]\nn\\
&<\mm_k - \mathbb{E}[M|-\infty<M<\mm_k]\nn\\
&= \mm_k - \mu + \sigma{\phi({\mm_{k}-\mu\over\sigma})\over\Phi({\mm_{k}-\mu\over\sigma})} \nn\\
& \stackrel{\mathclap{\text{(a)}}}{=} \sigma{\phi({\mu-\mm_{k}\over\sigma})\over1-\Phi({\mu-\mm_{k}\over\sigma})}-\sigma{\mu-\mm_{k}\over\sigma} \nn\\
& \stackrel{\mathclap{\text{(b)}}}{<} \sigma \left({\sqrt{\left({\mu-\mm_{k}\over\sigma}\right)^2+4}+{\mu-\mm_{k}\over\sigma}\over2}-{\mu-\mm_{k}\over\sigma}\right)  \nn\\
&< {\sigma\over2}\Bigg(\sqrt{\left({\mu-\mm_{k}\over\sigma}\right)^2+4\left({\mu-\mm_{k}\over\sigma}\right)+4}-{\mu-\mm_{k}\over\sigma}\Bigg) \nn\\
&= \sigma .
\label{eq:gaussLeftIneq}
\end{align}
Here, (a) holds since $\phi(x)=\phi(-x)$ and $\Phi(x)=1-\Phi(-x)$, and (b) follows from an inequality on the upper bound of the Mill's ratio \cite{birnbaumMillsRatio1942}; i.e., ${\phi(c)\over1-\Phi(c)}<{\sqrt{c^2+4}+c\over2}$ for $c>0$. Now, observe the following:
\begin{align*}
\sigma > \mm_1 - u_1 &= u_2 - \mm_1 + 2b\\
&> \mm_2 - u_2 + 2b\\
&\vdots\\
&> (t-1)(2b),
\end{align*}
where the inequalities follow from the fact that the Gaussian pdf of $M$ is monotonically increasing on $(-\infty,\mu]$, which implies that bin-lengths
are monotonically decreasing in the interval $(-\infty,\mu]$ when $b>0$. Further, for $b>0$, $t-1<{\sigma\over2b}$ is obtained, which implies that the number of bins in $(-\infty,\mu]$ is bounded by $\big\lfloor {\sigma\over2b} \big \rfloor$. 
\end{itemize}
\hspace*{\fill}\qed

\subsection{Proof of Remark~\ref{rem:gaussInfBinLength}} \label{ProofGaussInfBinLength}
For $b>0$, we can characterize what the bins looks like as the bin edges get very large with the following analysis:
\begin{align}
&\lim_{i\to\infty} \mathbb{E}[M|\mm_{i-1}^*<M<\mm_{i}^*] - \mm_{i-1}^* \nn\\
&= \lim_{i\to\infty} \mathbb{E}[M|\mm_{i-1}^*<M<\mm_{i-1}^*+l_i^*] - \mm_{i-1}^*\nn\\
&= \lim_{\mm_{i-1}^*\to\infty} \mu - \sigma{\phi({\mm_{i-1}^*+l_i^*-\mu\over\sigma})-\phi({\mm_{i-1}^*-\mu\over\sigma})\over\Phi({\mm_{i-1}^*+l_i^*-\mu\over\sigma})-\Phi({\mm_{i-1}^*-\mu\over\sigma})} - \mm_{i-1}^* \nn\\
& \stackrel{H}{=} \lim_{\mm_{i-1}^*\to\infty} \mu-\mm_{i-1}^* \nn\\
&- \sigma {\phi({\mm_{i-1}^*+l_i^*-\mu\over\sigma}){-\mm_{i-1}^*-l_i^*+\mu\over\sigma}{1\over\sigma}-\phi({\mm_{i-1}^*-\mu\over\sigma}){-\mm_{i-1}^*+\mu\over\sigma}{1\over\sigma}\over\phi({\mm_{i-1}^*+l_i^*-\mu\over\sigma}){1\over\sigma}-\phi({\mm_{i-1}^*-\mu\over\sigma}){1\over\sigma}} \nn\\
&= \lim_{\mm_{i-1}^*\to\infty} \mu-\mm_{i-1}^* + \mm_{i-1}^* - \mu + {l_i^*\over 1 - {\phi({\mm_{i-1}^*-\mu\over\sigma})\over\phi({\mm_{i-1}^*+l_i^*-\mu\over\sigma})}} \nn\\
&\stackrel{(a)}{=} 0
\label{eq:gaussRightInfLength}
\end{align}	 
Here, $\stackrel{H}{=}$ represents l'H\^ospital's rule and (a) follows from
\begin{align*}
&\lim_{\mm_{i-1}^*\to\infty}{\phi({\mm_{i-1}^*-\mu\over\sigma})\over\phi({\mm_{i-1}^*+l_i^*-\mu\over\sigma})}\\
&=\lim_{\mm_{i-1}^*\to\infty}\me^{{-\left({\mm_{i-1}^*-\mu\over\sigma}\right)^2+\left({\mm_{i-1}^*+l_i^*-\mu\over\sigma}\right)^2\over2}}=\infty.
\end{align*}
Then, \eqref{eq:centroidBoundaryEq} reduces to 
\begin{align}
\lim_{i\to\infty} \mm_i-&\mathbb{E}[M|\mm_{i-1}<M<\mm_{i}]\nn\\
&= \lim_{i\to\infty} \mathbb{E}[M|\mm_{i}<M<\mm_{i+1}]-\mm_i+2b \nn\\ \Rightarrow&\lim_{i\to\infty} \mm_i-\mm_{i-1}=\lim_{i\to\infty}\mm_i-\mm_i+2b \nn\\
\Rightarrow& \lim_{i\to\infty} \mm_i-\mm_{i-1}=2b .
\end{align}
In other words, the distance between the centroid and the lower edge of the bin converges to zero (i.e., the centroid of the bin converges to the left-edge), and length of the bins converge to $2b$.

Similarly, for $b<0$, we can characterize what the bins looks like as the bin edges get very large (in absolute value) with the following analysis:
\begin{align}
&\lim_{i\to\infty}  \mm_{r-i}^* - \mathbb{E}[M|\mm_{r-i-1}^*<M<\mm_{r-i}^*]  \nn\\
&= \lim_{i\to\infty} \mm_{r-i}^* - \mathbb{E}[M|\mm_{r-i}^*-l_{r-i}^*<M<\mm_{r-i}^*] \nn\\
&= \lim_{\mm_{r-i}^*\to-\infty} \mm_{r-i}^* - \mu + \sigma{\phi({\mm_{r-i}^*-\mu\over\sigma})-\phi({\mm_{r-i}^*-l_{r-i}^*-\mu\over\sigma})\over\Phi({\mm_{r-i}^*-\mu\over\sigma})-\Phi({\mm_{r-i}^*-l_{r-i}^*-\mu\over\sigma})} \nn\\
&\stackrel{H}{=} \lim_{\mm_{r-i}^*\to-\infty} \mm_{r-i}^* - \mu \nn\\
&+ \sigma {\phi({\mm_{r-i}^*-\mu\over\sigma}){-\mm_{r-i}^*+\mu\over\sigma}{1\over\sigma}-\phi({\mm_{r-i}^*-l_{r-i}^*-\mu\over\sigma}){-\mm_{r-i}^*+l_{r-i}^*+\mu\over\sigma}{1\over\sigma}\over\phi({\mm_{r-i}^*-\mu\over\sigma}){1\over\sigma}-\phi({\mm_{r-i}^*-l_{r-i}^*-\mu\over\sigma}){1\over\sigma}} \nn\\
&= \lim_{\mm_{r-i}^*\to-\infty} \mm_{r-i}^* - \mu - \mm_{r-i}^* + \mu - {l_i^*\over  {\phi({\mm_{r-i}^*-\mu\over\sigma})\over\phi({\mm_{r-i}^*-l_{r-i}^*-\mu\over\sigma})}-1} \nn\\
&\stackrel{(a)}{=} 0 .
\end{align}	 
Here, (a) follows from 
\begin{align*}
&\lim_{\mm_{r-i}^*\to-\infty}{\phi({\mm_{r-i}^*-\mu\over\sigma})\over\phi({\mm_{r-i}^*-l_{r-i}^*-\mu\over\sigma})} \\
&=\lim_{\mm_{r-i}^*\to\infty}\me^{{-({\mm_{r-i}^*-\mu\over\sigma})^2+({\mm_{r-i}^*-l_{r-i}^*-\mu\over\sigma})^2\over2}}
=\infty.
\end{align*} 
Similar to the $b>0$ case, the distance between the centroid and the upper edge of the bin converges to zero (i.e., the centroid of the bin converges to the right-edge), and length of the bins converge to $-2b$.
\hspace*{\fill}\qed

\section*{Acknowledgment}

Some of the results, in particular Proposition~\ref{prop:expProp1} and Proposition~\ref{prop:expProp2} build on the project report \cite{FurrerReport} written by Philippe Furrer, Stephen Kerner, Stanislav Fabricius, Serdar~Y{\"u}ksel and Tam\'{a}s Linder.

\bibliographystyle{IEEEtran}
\bibliography{NumberOfBins}

\begin{thebibliography}{10}
\providecommand{\url}[1]{#1}
\csname url@samestyle\endcsname
\providecommand{\newblock}{\relax}
\providecommand{\bibinfo}[2]{#2}
\providecommand{\BIBentrySTDinterwordspacing}{\spaceskip=0pt\relax}
\providecommand{\BIBentryALTinterwordstretchfactor}{4}
\providecommand{\BIBentryALTinterwordspacing}{\spaceskip=\fontdimen2\font plus
\BIBentryALTinterwordstretchfactor\fontdimen3\font minus
  \fontdimen4\font\relax}
\providecommand{\BIBforeignlanguage}[2]{{%
\expandafter\ifx\csname l@#1\endcsname\relax
\typeout{** WARNING: IEEEtran.bst: No hyphenation pattern has been}%
\typeout{** loaded for the language `#1'. Using the pattern for}%
\typeout{** the default language instead.}%
\else
\language=\csname l@#1\endcsname
\fi
#2}}
\providecommand{\BIBdecl}{\relax}
\BIBdecl

\bibitem{NumberOfBinsISIT2019}
S.~Sar{\i}ta\c{s}, P.~{Furrer}, S.~{Gezici}, T.~{Linder}, and S.~{Y\"uksel},
  ``On the number of bins in equilibria for signaling games,'' in \emph{IEEE
  International Symposium on Information Theory (ISIT)}, 2019, pp. 972--976.

\bibitem{SignalingGames}
V.~P. Crawford and J.~Sobel, ``Strategic information transmission,''
  \emph{Econometrica}, vol.~50, pp. 1431--1451, 1982.

\bibitem{tacWorkArxiv}
S.~Sar{\i}ta\c{s}, S.~Y{\"{u}}ksel, and S.~Gezici, ``Quadratic
  multi-dimensional signaling games and affine equilibria,'' \emph{IEEE
  Transactions on Automatic Control}, vol.~62, no.~2, pp. 605--619, Feb. 2017.

\bibitem{dynamicGameArxiv}
S.~Sar{\i}ta\c{s}, S.~Y\"uksel, and S.~Gezici, ``Dynamic signaling games with
  quadratic criteria under {N}ash and {S}tackelberg equilibria,''
  \emph{Automatica}, vol. 115, p. 108883, May 2020.

\bibitem{eqSelectionBook}
J.~Harsanyi and R.~Selten, \emph{A General Theory of Equilibrium Selection in
  Games}.\hskip 1em plus 0.5em minus 0.4em\relax Cambridge, MA: The MIT Press,
  1988.

\bibitem{basols99}
T.~Ba\c{s}ar and G.~J. Olsder, \emph{Dynamic Noncooperative Game Theory}.\hskip
  1em plus 0.5em minus 0.4em\relax Philadelphia, PA: SIAM Classics in Applied
  Mathematics, 1999.

\bibitem{CedricWork}
F.~Farokhi, A.~M.~H. Teixeira, and C.~Langbort, ``Estimation with strategic
  sensors,'' \emph{IEEE Transactions on Automatic Control}, vol.~62, no.~2, pp.
  724--739, Feb. 2017.

\bibitem{akyolITapproachGame}
E.~Akyol, C.~Langbort, and T.~Ba\c{s}ar, ``Information-theoretic approach to
  strategic communication as a hierarchical game,'' \emph{Proceedings of the
  IEEE}, vol. 105, no.~2, pp. 205--218, Feb. 2017.

\bibitem{omerHierarchial}
M.~O. Sayin, E.~Akyol, and T.~Ba\c{s}ar, ``Hierarchical multistage {G}aussian
  signaling games in noncooperative communication and control systems,''
  \emph{Automatica}, vol. 107, pp. 9--20, 2019.

\bibitem{LloydIT82}
S.~Lloyd, ``Least squares quantization in {PCM},'' \emph{IEEE Transactions on
  Information Theory}, vol.~28, no.~2, pp. 129--137, 1982.

\bibitem{misBehavingAgents}
I.~Shames, A.~M.~H. Teixeira, H.~Sandberg, and K.~H. Johansson, ``Agents
  misbehaving in a network: a vice or a virtue?'' \emph{IEEE Network}, vol.~26,
  no.~3, pp. 35--40, May 2012.

\bibitem{sobelSignal}
J.~Sobel, ``Signaling games,'' in \emph{Encyclopedia of Complexity and Systems
  Science}, R.~A. Meyers, Ed.\hskip 1em plus 0.5em minus 0.4em\relax Springer
  New York, 2009, pp. 8125--8139.

\bibitem{bayesianPersuasion}
E.~Kamenica and M.~Gentzkow, ``Bayesian persuasion,'' \emph{American Economic
  Review}, vol. 101, no.~6, pp. 2590--2615, Oct. 2011.

\bibitem{LeTreustAllerton16}
M.~Le~Treust and T.~Tomala, ``Information design for strategic coordination of
  autonomous devices with non-aligned utilities,'' in \emph{Annual Allerton
  Conference on Communication, Control, and Computing (Allerton)}, 2016, pp.
  233--242.

\bibitem{LeTreustJET19}
M.~{Le Treust} and T.~Tomala, ``Persuasion with limited communication
  capacity,'' \emph{Journal of Economic Theory}, vol. 184, p. 104940, 2019.

\bibitem{LeTreustArxiv20}
------, ``Strategic communication with side information at the decoder,''
  \emph{arXiv preprint arXiv:1911.04950}, 2020.

\bibitem{subjectiveBiasSIT}
V.~S.~S. {Nadendla}, C.~{Langbort}, and T.~Ba\c{s}ar, ``Effects of subjective
  biases on strategic information transmission,'' \emph{IEEE Transactions on
  Communications}, vol.~66, no.~12, pp. 6040--6049, Dec. 2018.

\bibitem{BattagliniEconometrica2002}
M.~Battaglini, ``Multiple referrals and multidimensional cheap talk,''
  \emph{Econometrica}, vol.~70, no.~4, pp. 1379--1401, 2002.

\bibitem{ComparativeCheap07}
A.~Chakraborty and R.~Harbaugh, ``Comparative cheap talk,'' \emph{Journal of
  Economic Theory}, vol. 132, no.~1, pp. 70--94, 2007.

\bibitem{MiuraGAB2014}
S.~Miura, ``Multidimensional cheap talk with sequential messages,'' \emph{Games
  and Economic Behavior}, vol.~87, pp. 419--441, 2014.

\bibitem{ambrus2008multi}
A.~Ambrus and S.~Takahashi, ``Multi-sender cheap talk with restricted state
  spaces,'' \emph{Theoretical Economics}, vol.~3, no.~1, pp. 1--27, 2008.

\bibitem{blume2007noisy}
A.~Blume, O.~J. Board, and K.~Kawamura, ``Noisy talk,'' \emph{Theoretical
  Economics}, vol.~2, no.~4, pp. 395--440, 2007.

\bibitem{LongCheapTalk03}
R.~J. Aumann and S.~Hart, ``Long cheap talk,'' \emph{Econometrica}, vol.~71,
  no.~6, pp. 1619--1660, 2003.

\bibitem{DynamicSIT14}
M.~Golosov, V.~Skreta, A.~Tsyvinski, and A.~Wilson, ``Dynamic strategic
  information transmission,'' \emph{Journal of Economic Theory}, vol. 151, pp.
  304--341, 2014.

\bibitem{KartikJET07}
N.~Kartik, M.~Ottaviani, and F.~Squintani, ``Credulity, lies, and costly
  talk,'' \emph{Journal of Economic Theory}, vol. 134, no.~1, pp. 93--116,
  2007.

\bibitem{GescheGEB21}
T.~Gesche, ``De-biasing strategic communication,'' \emph{Games and Economic
  Behavior}, vol. 130, pp. 452--464, 2021.

\bibitem{CheapTalkTwoSenderInformative}
A.~McGee and H.~Yang, ``Cheap talk with two senders and complementary
  information,'' \emph{Games and Economic Behavior}, vol.~79, pp. 181--191,
  2013.

\bibitem{quantizationSurvey}
R.~M. Gray and D.~L. Neuhoff, ``Quantization,'' \emph{IEEE Transactions on
  Information Theory}, vol.~44, no.~6, pp. 2325--2383, Oct. 1998.

\bibitem{Fleischer64}
P.~E. Fleischer, ``Sufficient conditions for achieving minimum distortion in a
  quantizer,'' \emph{IEEE International Convention Record, Part I}, vol.~12,
  pp. 104--111, 1964.

\bibitem{trushkin82}
A.~Trushkin, ``Sufficient conditions for uniqueness of a locally optimal
  quantizer for a class of convex error weighting functions,'' \emph{IEEE
  Transactions on Information Theory}, vol.~28, no.~2, pp. 187--198, Mar. 1982.

\bibitem{kieffer83}
J.~Kieffer, ``Uniqueness of locally optimal quantizer for log-concave density
  and convex error weighting function,'' \emph{IEEE Transactions on Information
  Theory}, vol.~29, no.~1, pp. 42--47, Jan. 1983.

\bibitem{wu92}
X.~Wu, ``On convergence of {L}loyd's method {I},'' \emph{IEEE Transactions on
  Information Theory}, vol.~38, no.~1, pp. 171--174, Jan. 1992.

\bibitem{gyorgyLinder2003}
A.~Gyorgy, T.~Linder, P.~A. Chou, and B.~J. Betts, ``Do optimal
  entropy-constrained quantizers have a finite or infinite number of
  codewords?'' \emph{IEEE Transactions on Information Theory}, vol.~49, no.~11,
  pp. 3031--3037, Nov. 2003.

\bibitem{csLloydMax}
B.~Larrousse, O.~Beaude, and S.~Lasaulce, ``Crawford-{S}obel meet {L}loyd-{M}ax
  on the grid,'' in \emph{IEEE International Conference on Acoustics, Speech
  and Signal Processing (ICASSP)}, May 2014, pp. 6127--6131.

\bibitem{QuantizationGameTSP21}
A.~Mani, L.~R. Varshney, and A.~Pentland, ``Quantization games on social
  networks and language evolution,'' \emph{IEEE Transactions on Signal
  Processing}, vol.~69, pp. 3922--3934, 2021.

\bibitem{eqSelectionRefinement}
J.~S. Banks and J.~Sobel, ``Equilibrium selection in signaling games,''
  \emph{Econometrica}, vol.~55, no.~3, pp. 647--661, 1987.

\bibitem{eqSelectionSignaling}
M.~Mitzkewitz, ``Equilibrium selection and simple signaling games,'' 2017,
  working paper no. 9.

\bibitem{SobelSelectionEconometrica08}
Y.~Chen, N.~Kartik, and J.~Sobel, ``Selecting cheap-talk equilibria,''
  \emph{Econometrica}, vol.~76, no.~1, pp. 117--136, 2008.

\bibitem{VoraCDC20}
A.~S. Vora and A.~A. Kulkarni, ``Information extraction from a strategic sender
  over a noisy channel,'' in \emph{IEEE Conference on Decision and Control
  (CDC)}, 2020, pp. 354--359.

\bibitem{VoraISIT20}
------, ``Achievable rates for strategic communication,'' in \emph{{IEEE}
  International Symposium on Information Theory {(ISIT)}}, 2020, pp.
  1379--1384.

\bibitem{AkyolITW15}
E.~Akyol, C.~Langbort, and T.~Ba{\c{s}}ar, ``Strategic compression and
  transmission of information,'' in \emph{IEEE Information Theory Workshop -
  Fall (ITW)}, 2015, pp. 219--223.

\bibitem{AkyolISIT16}
------, ``On the role of side information in strategic communication,'' in
  \emph{IEEE International Symposium on Information Theory (ISIT)}, 2016, pp.
  1626--1630.

\bibitem{CSCorrectionEconometrica21}
H.~Kono and M.~Kandori, ``Corrigendum to {C}rawford and {S}obel (1982)
  ``{S}trategic information transmission'','' \emph{Econometrica}, vol.~89,
  no.~4, pp. 1--10, 2021.

\bibitem{bagnoli2005log}
M.~Bagnoli and T.~Bergstrom, ``Log-concave probability and its applications,''
  \emph{Economic theory}, vol.~26, no.~2, pp. 445--469, 2005.

\bibitem{CuleSamworthEJS10}
M.~Cule and R.~Samworth, ``{Theoretical properties of the log-concave maximum
  likelihood estimator of a multidimensional density},'' \emph{Electronic
  Journal of Statistics}, vol.~4, pp. 254 -- 270, 2010.

\bibitem{fixedPointBook}
A.~Granas and J.~Dugundji, \emph{Fixed Point Theory}.\hskip 1em plus 0.5em
  minus 0.4em\relax Springer-Verlag New York, 2003.

\bibitem{Szalay12}
D.~Szalay, ``Strategic information transmission and stochastic orders,'' 2012,
  unpublished.

\bibitem{IbragimovUnimodal56}
I.~A. Ibragimov, ``On the composition of unimodal distributions,'' \emph{Theory
  of Probability and Its Applications}, vol.~1, no.~2, pp. 255--260, 1956.

\bibitem{niculescu2018convex}
C.~Niculescu and L.-E. Persson, \emph{Convex Functions and Their Applications:
  A Contemporary Approach}, 2nd~ed.\hskip 1em plus 0.5em minus 0.4em\relax New
  York, NY, USA: Springer, 2018.

\bibitem{guide2006infinite}
C.~D. Aliprantis and K.~C. Border, \emph{Infinite Dimensional Analysis: A
  Hitchikers Guide}.\hskip 1em plus 0.5em minus 0.4em\relax Berlin:
  Springer-Verlag, 2006.

\bibitem{birnbaumMillsRatio1942}
Z.~W. Birnbaum, ``An inequality for {M}ill's ratio,'' \emph{The Annals of
  Mathematical Statistics}, vol.~13, no.~2, pp. 245--246, June 1942.

\bibitem{FurrerReport}
P.~Furrer, S.~Kerner, S.~Fabricius, S.~Y{\"u}ksel, and T.~Linder, ``Game theory
  and information, {MTHE} 493 {T}echnical {R}eport,'' Apr. 2014.

\end{thebibliography}

\end{document}